\documentclass{article}
\usepackage{amsmath,amsfonts,amssymb,amsthm,pstricks,
pst-plot,pst-node,eucal}
\usepackage{graphicx, epsfig}

\newcommand{\RR}{\mathbb R}
\newcommand{\CC}{\mathbb C}
\newcommand{\ZZ}{\mathbb Z}
\newcommand{\M}{\mathcal M}
\newcommand{\R}{\mathcal R}
\newcommand{\cL}{\mathcal L}
\newcommand{\cH}{\mathcal H}
\newcommand{\VR}{\mbox{\v{R}}}
\newcommand{\g}{\mathfrak g}
\newcommand{\w}{\omega}

\newcommand{\al}{\alpha}
\newcommand{\pa}{\partial}
\newcommand{\e}{\epsilon}

\newtheorem{theorem}{Theorem}
\newtheorem{proposition}{Proposition}

\newtheorem{remark}{Remark}
\newtheorem{definition}{Definition}

\begin{document}

\title
{Lectures on the integrability of the $6$-vertex model.}
\author{Nicolai Reshetikhin
\thanks{Department of Mathematics,
    University of California at Berkeley,
Berkeley, CA 94720-3840.
and KDV Institute for Mathematics,
Universiteit van Amsterdam,
Plantage Muidergracht 24,
1018 TV, Amsterdam, The Netherlands.E-mail:
reshetik@math.berkeley.edu }}
\date{}
\maketitle

\tableofcontents

\section{Introduction}

The goal of these notes is to outline the relation between
solvable models in statistical mechanics, and classical and quantum integrable spin chains. The main examples are the $6$-vertex model
in statistical mechanics,  and  spin chains related to the loop algebra $Lsl_2$.

The $6$-vertex model emerged as a version of the Pauling's ice model,
more generally, as a two-dimensional model of ferroelectricity.
The free energy per site was computed exactly in the thermodynamical
limit by E. Lieb \cite{LiebM}. The free energy as a function of
electric fields was computed by Sutherland and Yang \cite{SY}.
For details and more references on earlier works on the $6$-vertex
model see \cite{LW}. The structure of the free energy as a function of electric
fields was also studied in \cite{Nold}\cite{BS}.

Baxter descovered that Boltzman weights of the $6$-vertex model can be
arranged into a matrix which staisfies what is now known as the
Yang-Baxter equation. For more references on the $6$-vertex model
and on the consequences of the Yang-Baxter equation for the
weights of the $6$-vertex model see \cite{Ba}.

The $6$-vertex model and similar `integrable' models
in statistical mechanics became the subject of renewed
research activity after the discovery of Sklyanin
of the relation between the Yang-Baxter relation
in the $6$-vertex model and the quantization of classical
integrable systems \cite{Skl-1}. It lead to
the discovery of many `hidden' algebraic structures
of the $6$-vertex model and to the construction of quantizations of a number of important classical field theories
\cite{FST}\cite{FT}. For more references see, for example \cite{KBI}.
Further development of this subject resulted in
the development of quantum groups \cite{Dr} and further
understanding of algebraic nature of integrability.

Classical spin chains related to the Lie group $SL_2$
is an important family of classical integrable systems
related to the $6$-vertex model. The continuum  version
of this model is known as Landau-Lifshitz model.  One of its
quantum counterparts, the Heisenberg spin chain
has a long history. Eigenvectors of its Hamiltonian of
the Heisenberg model
were constructed by Bethe in 1931 using the substitution which
is known now as the Bethe ansatz. He expressed
the eigenvalues of the Hamiltonian in terms of solutions to
a system of algebraic equations known now as Bethe equations.
An algebraic version of this substitution was found in
\cite{FT}.

These lectures consist of three major parts. First part
is a survey of some basic facts about classical integrable
spin chains. The second part is a survey of
the corresponding quantum spin chains. The third part is
focus on the $6$-vertex model and on the limit shape
phenomenon. The Appendix has a number of random useful
facts.

The author is happy to thank A. Okounkov, K. Palamarchuk, E. Sklyanin,
and F. Smirnov  for
discussions, B. Sturmfels for an important remark about the
solutions to Bethe equations, and the organizers of the
school for the opportunity to give these lectures.

This work was supported by the Danish National
Research Foundation through the
Niels Bohr initiative, the author is
grateful to the Aarhus University for the
hospitality. It was also
supported by the NSF grant DMS-0601912.

\section{Classical integrable spin chains}

The notion of a Poisson Lie group developed from the
study of integrable systems and their relation to
solvable models in statistical mechanics.

It is a geometrical structure behind Poisson structures
on Lax operators, which emerged in the analysis of Hamiltonian
structures in integrable partial differential equations.
First examples of these
structures are related to taking semiclassical limit of Baxter's $R$-matrix for the 8-vertex model. The notion of Poisson Lie groups and Lie bialgebras places these examples into the context of Lie theory and provides a natural versions of such systems related to
other simple Lie algebras.

\subsection{Classical $r$-matrices and the construction of classical
integrable spin chains}\label{cl-int-syst}

Here we will recall the construction of classical integrable
systems based on classical $r$-matrices.

A classical $r$-matrix with (an additive) spectral parameter is a
holomorphic function on $\CC$ with values in $End(V)^{\otimes 2}$ which satisfies the classical Yang-Baxter equation:
\begin{equation}\label{cl-YBE}
[r_{12}(u),r_{13}(u+v)]+[r_{12}(u),r_{23}(v)]+[r_{13}(u+v),r_{23}(v)]=0
\end{equation}
where $r_{ij}(u)$ act in $V^{\otimes 3}$, such that $r_{12}(u)=r(u)\otimes 1, \ \ r_{23}(u)=1\otimes r(u)$, etc.

Let $L(u)$ be a holomorphic function on $\CC$ of certain type (for example a polynomial). Matrix elements of coefficients of this function satisfy quadratic $r$-matrix Poisson brackets if
Their generating function $L(u)$ has the Poisson brackets
\begin{equation}\label{r-br}
\{L_1(u), L_2(v)\}=[r(u), L_1(u)L_2(v)]
\end{equation}
where $L_1(u)=L(u)\otimes 1, L_2(u)=1\otimes L(u)$.
The expression on the left is the collection of Poisson
brackets $\{L_{ij}(u), L_{kl}(v)\}$.

Consider the product
\[
T(u)=L^{(N)}(u-a_N)\dots L^{(1)}_1(u-a_1)
\]
Matrix elements of $T(u)$ are functions on
$P_N\times \dots \times P_1$. The factor
$L^{(i)}$ is the function on $P_i$. The
Poisson structure on $P_i$ is as above.

The $r$-matrix Poisson brackets for
$L^{(i)}(u)$ imply similar Poisson brackets
for $T_{ij}(u)$:
\[
\{T_1(u), T_2(v)\}=[r(u-v), T_1(u)T_2(v)]
\]

Taking the trace in this formula we se that
\[
\{t(u), t(v)\}=0
\]

Fix symplectic leaves $S_i\subset P_i$ for each $i$.The
restriction of the generating function $t(u)$ to
the product of symplectic leaves $S=S_N\times \dots \times S_1\subset P_N\times \dots \times P_1$ gives the generating
function for commuting functions on this symplectic
manifold.

Under the right circumstances the generating function $t(u)$ will
produce the necessary number of independent functions to produce
a completely integrable system, i.e.
$dim(S)/2$.

One of the main conceptual questions in this construction is: Why classical $r$-matrices exist? Indeed, if $n=dim(V)$, the equation (\ref{cl-YBE}) is a system of $n^6$ functional equations
for $n^4$ functions.
The construction of the Drinfeld double of a Lie bialgebra
provides an answer to this question \cite{Dr}.

\subsection{Classical $L$-operators related to $\widehat{sl_2}$}

In this section we will focus on classical $L$-operators
for the classical $r$-matrix corresponding to the standard Lie bialgebra structure on $\widehat{sl_2}$.

Such $L$-operators describe finite dimensional Poisson submanifolds in
the infinite dimensional Poisson Lie group $LSL_2$. For some basic
facts and references see Appendix \ref{be}.
Up to a scalar multiple they are polynomials in the spectral
variable $z$. One of such simplest Poisson submanifolds correspond
to polynomials of first degree of the following form:
\begin{equation}\label{orbit}
L(u)=\left(\begin{array}{cccc}a+ a'z^2 &  b'\\
z^2 b & c+c' z^2 \end{array}\right)
\end{equation}
where $z=\exp(u)$. The
$r$-matrix Poisson brackets on $L(u)$ with the $r$-matrix
given by (\ref{cl-r}) induce a Poisson algebra structure on
the algebra of polynomials in $a, b, ...$. This Poisson algebra can be specialized further (by quotienting with respect to
corresponding Poisson ideal). As a result we arrive to the following L-operator:
\begin{equation}\label{el-L}
L(u)=\left(
\begin{array}{cccc}
  zk-z^{-1}k^{-1} & z^{-1}f   \\
  zf &  zk^{-1}-z^{-1}k
\end{array} \right)
\end{equation}
which satisfies the $r$-matrix Poisson brackets (\ref{r-br})
with the following brackets on $k, e, f$:
\begin{equation}\label{br-1}
\{k,e\}=\e ke, \ \ \{k, f\}=- \e kf
\end{equation}
\begin{equation}\label{br-2}
\{e,f\}=2\e (k^2-k^{-2}),
\end{equation}

The function
\begin{equation}\label{Poisson-cas}
c=ef+ k^2+k^{-2},
\end{equation}
Poisson commute with all other elements of this Poisson algebra, i.e. it is the Casimir function. It is easy to show that this
is the only Casimir function on this Poisson manifold.

It is easy to check that
\begin{equation}\label{cl-crossing}
L(-u)^t=-D_z\sigma^y_2L(u)\sigma^y_2 D_z^{-1}
\end{equation}

On the level surface of $c$ parameterized as
$c=t^2+t^{-2}$ this matrix satisfies the extra identity
\[
L(u)\theta(L(-u))^t=(zt-z^{-1}t^{-1})(tz^{-1}-t^{-1}z)I
\]
where $I$ is the identity matrix and
\[
\theta(e)=f, \ \ \theta(f)=e, \ \ \theta(k)=k
\]
is the anti-Poisson involution: $\{\theta(a)\theta(b)\}=-\theta(\{a,b\})$, $\theta^2=id$.

Its easy to find the determinant of $L(u)$:
\[
det(L(u))=(zt-z^{-1}t^{-1})(t^{-1}z-tz^{-1})
\]

When $z=t, t^{-1}$ the matrix $L$ degenerates
to one dimensional projectors.

\begin{equation}\label{L-deg-1}
L(t)=\left(
\begin{array}{c}
  \frac{tk-t^{-1}k^{-1}}{t^{-1}e}    \\
  1
\end{array} \right)\otimes \left(t^{-1}e , tk^{-1}-t^{-1}k\right)
\end{equation}
\begin{equation}\label{L-deg-2}
 L(t^{-1})=\left(
\begin{array}{c}
  \frac{t^{-1}k-tk^{-1}}{te}    \\
  1
\end{array} \right)\otimes \left(te , t^{-1}k^{-1}-tk\right)
\end{equation}

\subsection{Real forms}\label{cl-r-form}

\subsubsection{} Here we will describe the Poisson manifold
$su_2^*$ .
Let $S_1, S_2, S_3$ be the coordinates $su_2^*$ corresponding to
the usual orthonormal basis in $su_2$. The Poisson brackets
between these coordinate functions are
\[
\{S_1, S_2\}=2S_3, \ \ \{S_2, S_3\}=2S_1, \ \ \{S_3, S_1\}=2S_2,
\]
These coordinates are also known as classical spin coordinates.
The center of this Poisson algebra is generated by $C=S_1^2+S_2^2+
S_3^2$.

It is convenient to introduce $S^+=(S_1+iS^2)/2$, $S^-=(S_1-iS^2)/2$. Since $S_i$ are real, $\overline{S^+}=S^-$.

The Poisson brackets between these coordinates are:
\[
\{S^+, S^-\}=-iS_3, \ \ \{S_3, S^\pm\}=\mp 2i S^\pm
\]

Level surfaces of $C$ are spheres. On the level surface with
$C=l^2$ with two point $S_3=\pm l$ being removed we have the following Darboux coordinates:
\[
S^+=e^{i\phi}\sqrt{pl-p^2}, \ \ S^-=e^{-i\phi}\sqrt{pl-p^2},
\ \ S_3=l-2p
\]
where $0<p<l$ and $0<\phi \leq 2\pi$ and $\{p,\phi\}=1$.

\subsubsection{} The Poisson algebra
\[
\{k, e\}=\e ke, \ \ \{k, f\}=-2\e kf,
\]
\[
\{e, f\}=2\e (k^2-k^{-2})
\]
has the two real forms which are important for spin chains with
compact phase spaces:
\begin{itemize}
\item The real form with $\e=1$, $|k|=1$, $e=\overline{f}$.
\item And the real form where $\e=i$, $e=\overline{f}$, and $\overline{k}=k$.
\end{itemize}

 In the first case the level surface of (\ref{Poisson-cas}) with $c=2\cos 2R$ without two points $e=f=0$ has the
following Darboux coordinates:
\[
e=2e^{i\phi}\sqrt{\sin(p)\sin(2R-p)}, \ \
f=2e^{-i\phi}\sqrt{\sin(p)\sin(2R-p)}, \ \ k=e^{i(R-p)}
\]
with $\{p, \phi\}=1$, and $0<p<l$ and $0<\phi \leq 2\pi$.

Similarly, in the second case  the level
surface $c=2\cosh R$ without two points $e=f=0$ have
Darboux coordinates
\[
e=2e^{i\phi}\sqrt{\sinh(p)\sinh(2R-p)}, \ \
f=2e^{-i\phi}\sqrt{\sinh(p)\sinh(2R-p)}, \ \ k=e^{R-p}
\]
where $\phi \in \RR$, $0<p<R$, and $\{p, \phi\}=1$.

We will denote these level surfaces by $S^{(R)}$. In the
compact case $S^{(R)}$ is diffeomorphic to a sphere. It can
be realized as the unit sphere with the symplectic form
dependent on $R$.

\subsection{Integrable classical local $SU_2$-spin chains}\label{cl-Ham-sec}

Because the $S^{(R)}$ is a symplectic leaf of the Poisson Lie group $LGL_2$, the Cartesian product
\[
\M^{R_1,\dots, R_N}=
S^{(R_1)}\times \dots \times S^{(R_N)}
\]
is a Poisson submanifold in $LGL_2$, which means matrix elements of the monodromy matrix
\begin{equation}\label{cl-tran}
T(z)=L^{(R_1)}(za_1)L^{(R_2)}(za_2)\dots L^{(R_N)}(za_N)
\end{equation}
satisfy the $r$-matrix Poisson brackets.

In order to obtain local Hamiltonians in the homogeneous classical spin chain $a_1=\dots=a_N=1$, $R_1=\dots=R_N=R$ one can use the
degenerations (\ref{L-deg-1}), (\ref{L-deg-2}). Combining these
formulae (\ref{L-deg-1}), (\ref{L-deg-2}), and (\ref{cl-crossing})
we obtain the following identities:
\[
L(t)=\alpha\otimes \beta^t, \ \ L(t^{-1})=-\sigma^yD\beta\otimes \alpha^tD^{-1}\sigma^y
\]
where column vector $\alpha$ and row vector $\beta^t$ are
given in (\ref{L-deg-1}). These identities imply
\[
tr(T(t))=\prod_{n=1}^N(\beta_n,\alpha_{n+1}), \ \ tr(T(t^{-1}))=(-1)^N\prod_{n=1}^N(\alpha_n,\beta_{n+1})
\]
Here we assume the periodicity $\alpha_{N+1}=\alpha_1$ and
$\beta_{N+1}=\beta_1$. Now notice that
\[
tr(L_n(t)L_{n+1}(t))=(\alpha_n,\beta_{n+1})(\beta_n,\alpha_{n+1})
\]
One the other hand this trace can be computed explicitly:
\begin{multline}
tr(L_n(t)L_{n+1}(t))=e_nf_{n+1}+f_ne_{n+1}+
(tk_n-t^{-1}k_n^{-1})(tk_{n+1}-t^{-1}k_{n+1}^{-1})+\\
(tk_n^{-1}-t^{-1}k_n)(tk_{n+1}^{-1}-t^{-1}k_{n+1})
\end{multline}

This gives the first local Hamiltonian
\begin{multline}
H=\log(tr(T(t))tr(T(t^{-1})))=\sum_{n=1}^N\log(e_nf_{n+1}+f_ne_{n+1}+
\\ (tk_n-t^{-1}k_n^{-1})(tk_{n+1}-t^{-1}k_{n+1}^{-1})+
(tk_n^{-1}-t^{-1}k_n)(tk_{n+1}^{-1}-t^{-1}k_{n+1}))
\end{multline}

Other local spin Hamiltonians can be chosen as logarithmic derivatives of $tr(T(z))$ at when $z=t^{\pm 1}$, for details see
\cite{FT} and references therein.

When $N$ is even and inhomogeneities are alternating $a_1=a,a_2=a^{-1},a_3=a,\dots, a_N=a^{-1}$ there is a similar
construction of local Hamiltinians also based on
degenerations (\ref{L-deg-1}), (\ref{L-deg-2}).
Again, all logarithmic derivatives of $T(z)$ at points
$z=at^{\pm 1}, a^{-1}t^{\pm 1}$
are local spin Hamiltonians.

In the continuum limit the Hamiltian dynamics generated by these Hamiltians converges to the Landau-Lifshitz equation, see \cite{FT} and references therein.

\section{Quantization of local integrable spin chains}

\subsection{Quantum integrable spin chains}

In the appendix \ref{q-in-sys} there is a short discussion of integrable quantization of classical integrable system.

A quantization of a local classical integrable
spin chin is an integrable quantization of a
classical local integrable spin chain such
that quantized Hamiltonians remain local. That is
the collection of the following data:

\begin{itemize}
\item A choice of the quantization of the algebra of observables
of the classical system (in a sense of the section Appendix \ref{quantization}), i.e a family of associative algebras
with a $*$-involution which deform the classical
Poisson algebra of observables.
\item  A choice of  a maximal commutative subalgebra in the algebra which is a quantization of Poisson commuting algebra of classical integrals.

\item In addition,locality of the quantization means that the quantized algebra of observables is the tensor product of local algebras (one for each site of our one-dimensional lattice): $A_h=\otimes_{n=1}^NB_h^{(n)}$, and that the quantum Hamiltonian has the same local structure as the classical Hamiltonian  (\ref{cl-loc}):
\[
H=\sum_n H_n
\]
where $H_n=1\otimes \dots\otimes H^{(k)}\otimes \dots\otimes 1$ and $H^{(k)}\in B^{(n)}_h\otimes B^{(n+1)}_h\otimes\dots\otimes  B^{(n+k)}_h $.

\item A $*$-representation of the algebra of observables (the
space of pure states of the system).

\end{itemize}

\subsection{The Yang-Baxter equation and the quantization}

Here we will describe the approach to the quantization
of classical spin chains with $r$-matrix Poisson bracket
for polynomial Lax matrices based on construction of
corresponding quantum $R$-matrices and quantum Lax matrices.

In modern language the construction of quantum $R$-matrix means
the construction of the corresponding quantum group, and the
construction of the quantum $L$-matrix means the construction of
the corresponding representation of the quantum group.

Suppose we have a classical integrable system with
commuting integrals obtained as coefficients of
the generating function $\tau(u)$ described in section \ref{cl-int-syst}.

The $R$-matrix quantization means the following:

\begin{itemize}
\item Find a family $R(u,h)$ of invertible linear operators acting
in $V\otimes V$ such that for each $h$ they satisfy the quantum Yang-Baxter equation
\[
R_{12}(u,h)R_{13}(u+v,h)R_{23}(v,h)=R_{23}(v,h)R_{13}(u+v,h)R_{12}(u,h)
\]
and when $h\to 0$
\[
R(u,h)=1+hr(u)+O(h^2)
\]
where $r(u)$ is the classical $r$-matrix.

\item Let the classical Lax matrix $L(u)$ be a matrix valued
function of $u$ of certain type (for example a polynomial of fixed degree), with matrix elements generating a Poisson algebra with Poisson brackets
(\ref{r-br}). For given $R(u,h)$ define the quantization of
this Poisson algebra as the associative algebra generated by
matrix elements of the matrix $\cL(u)$ of the same type as
$L(u)$ (for example a polynomial of the same degree) with defining relations
\begin{equation}\label{R-L}
R(u,h) \cL(u+v)\otimes \cL(v)=(1\otimes \cL(v))
(\cL(u+v)\otimes 1) R(u,h),
\end{equation}
Denote such algebra by $B_h$.

\item Consider the generating function
\begin{equation}\label{T}
T(u)=\cL_1(u-w_1)\otimes\dots \cL_N(u-w_N)
\end{equation}
acting in $End(V)\otimes B_h^{(1)}\otimes \dots\otimes B_h^{(N)}$. It is easy to see that the commutation relations
(\ref{R-L}) imply
\begin{equation}\label{R-T}
R(u,h) T(u+v)\otimes T(v)=(1\otimes T(v))
(T(u+v)\otimes 1) R(u,h),
\end{equation}
The invertibility of $R(u,h)$ together with the relations
(\ref{R-T}) imply that $t(u)=tr_V(T(u))\in B_h^{(N)}\otimes \dots\otimes B_h^{(1)}$ is a generating function for a commutative subalgebra in $B_h^{(1)}\otimes \dots\otimes B_h^{(N)}$:
\[
[t(u),t(v)]=0
\]
Under the right circumstances this commutative subalgebra is maximal and defines an integrable quantization of the corresponding classical integrable spin chain.
\end{itemize}

The $R$-matrix was found by Baxter. Sklyanin discovered that when $h\to 0$ the classical $R$-matrix defines a Possion structure
on $LGL_2$ defined by the formula (\ref{loop-P}) and that it implies the Poisson commutativity of traces.

There is an algebraic way to derive the Baxter's R-matrix from the universal $R$-matrix for $U_q(\widehat{gl_2})$. It is outlined in the appendix.

\subsection{Quantum Lax operators and representation theory}

\subsubsection{Quantum $LSL_2$}
Here is the formal definition of $C_q(\widehat{SL}_2)$
in terms of generators and relations.
Let $q$ be a nonzero complex number. The algebra  $C_q(\widehat{SL}_2)$
is a complex algebra generated by elements the $ T_{ij}^{(k)}$,
where $i,j=1,2$ and $k\in {\Bbb Z}$. Consider the matrix ${\mathcal T}(z)$
which is
the generating function for
the elements  $ T_{ij}^{(k)}$
\begin{equation}\label{gen-aff}
{\mathcal T}(z)=\sum_{k=1}^\infty T^{(k)} z^{2k}+ \left(
\begin{array}{cc}T_{11}^{(0)}& T_{12}^{(0)}\\ 0& T_{22}^{(0)}\end{array}\right).
\end{equation}
The determining relations in $C_q(\widehat{SL}_2)$
can be written as the following
matrix identity with entries in  $C_q(\widehat{SL}_2)$:
\begin{eqnarray} \label{relations}
R(z) {\mathcal T}(zw)\otimes {\mathcal T}(w)&=&(1\otimes {\mathcal T}(w))
({\mathcal T}(zw)\otimes 1) R(z),  \\
{\mathcal T}(qz)_{11}{\mathcal T}(z)_{22}&-&{\mathcal T}(qz)_{12}{\mathcal T}(z)_{21}=1,
\nonumber
\end{eqnarray}
where the tensor product is the tensor product of matrices. Matrix elements in this formula are multiplied as elements of $C_q(\widehat{GL}_2)$ in the order in which they appear.

The matrix $R(z)$ acts in ${\Bbb C}^2\otimes{\Bbb C}^2$ and
has the following structure in the tensor product basis:
\begin{equation}    \label{Rmatrix}
R(z) = \begin{pmatrix}
1 & 0 & 0 & 0 \\
0 & f(z) & z^{-1}g(z) & 0 \\
0 & zg(z) & f(z) & 0 \\
0 & 0 & 0 & 1
\end{pmatrix}.
\end{equation}
where
\[
f(z)=\dfrac{z-z^{-1}}{zq-z^{-1}q^{-1}}, \ \ g(z)=\dfrac{(q-q^{-1})}{zq-z^{-1}q^{-1}}
\]

It satisfies the Yang-Baxter equation.

\begin{remark} There is important function $s(z)$:
$$
s(z)=q^{-1/2} \frac{(z^2q^2;q^4)_\infty^2 }
{(z^2;q^4)_\infty (z^2q^4;q^4)_\infty },
$$
where
$$
(x;p)_\infty=\prod_{k=1}^\infty (1-xp^n).
$$

the matrix
\[
\R(z)=s(z) R(z)
\]
satisfies what is known in physics unitarity and the crossing symmetry:
\[
\R(z)\R(z^{-1})^t=1, \ \ \R(z^{-1})^{t_2}=C_2\R(zq^{-1})C_2
\]
where $t$ is the transposition with respect to the
standard scalar product in ${\CC^2}^{\otimes}$, $t_2$
is the transposition with respect to second factor in the tensor
product, and $C_2=1\otimes C$ where
\[
C=\left(\begin{array}{cc} 0& -i\\
i & 0
\end{array}\right)
\]

\end{remark}

The Hopf algebra structure on  $C_q(\widehat{SL}_2)$ is determined by the
following action of the antipode on the generators:
\begin{equation}   \label{comult}
\Delta ({\cal T}_{ij}(z))= \sum_{k={1,2}}{\mathcal T}_{ik}(z)\otimes
{\mathcal T}_{kj}(z).
\end{equation}
The right side here, as well as in the second relation in \eqref{relations},
is understood as a product of Laurent power series. The antipode is
determined by the relation
$$
{\mathcal T}(z)(S\otimes {id}{\mathcal T}(z))=1.
$$

Let $d$ be a nonzero complex number.
The identity \eqref{relations} implies that the power series
\begin{equation}  \label{t-m}
\tau_1(z;d)=d {\mathcal T}_{11}(z)+d^{-1} {\mathcal T}_{22}(z)
\end{equation}
generates a commutative subalgebra in $C_q(\widehat{SL}_2)$.

Chose a linear basis (for example ordered monomials in $T_{ij}^{(k)}$ ). The relations between generators will
give the multiplication rule for monomials which will depend on $q$. This multiplication turns into the commutative multiplication of coordinate functions when $q=1$. The commutator of two monomials, divided by $q-1$, at $q=1$  becomes the Poisson bracket. It is easy to check that this Poisson brackets is exactly the one
defined by the classical $r$-matrix (\ref{cl-r}).

\begin{remark} Let $D$ be a diagonal matrix. It is easy to see that $[D\otimes 1+1\otimes D, R(x)]=0$. It is easy to show that
if
\[
R(z){\mathcal T}_1(zw){\mathcal T}_2(w)={\mathcal T}_2(w){\mathcal T}_1(zw)R(z)
\]
then
\[
\tilde{R}(z)=(z^D\otimes 1)R(z)(z^{-D}\otimes 1), \ \ \tilde{{\mathcal T}}(z)=z^{D}{\mathcal T}(z)z^{-D}
\]
satisfy the same relation
\[
\tilde{R}(z)\tilde{{\mathcal T}}_1(zw)\tilde{{\mathcal T}}_2(w)=
\tilde{{\mathcal T}}_2(w)\tilde{{\mathcal T}}_1(zw)\tilde{R}(z)
\]
In particular $\tilde{R}(z)$ satisfies the Yang-Baxter equation.
Choosing $D=diag(-1/2,1/2)$ gives the $R$-matrix
(\ref{Rmatrix}) but with no factors $z^{\pm 1}$ off-diagonal.
This symmetric version of the R-matrix is the matrix of Boltzmann
weights in the 6-vertex model.
\end{remark}
\begin{remark}
If $A$ is an invertible diagonal matrix such that $(A\otimes A)R(z)=R(z)(A\otimes A)$ and ${\mathcal T}(z)$ is as above then
\[
{\mathcal T}^A(z)=A{\mathcal T}(z)A^{-1}
\]
also satisfies the relations (\ref{relations}).
\end{remark}

\subsection{Irreducible representations}\label{irreps}
\subsubsection{}It is easy to check that the following matrix satisfies the $R$-matrix commutation relations (\ref{relations})
\begin{equation}\label{q-L}
{\mathcal L}(z)=\left(
\begin{array}{cccc}
  zkq^{\frac{1}{2}}-z^{-1}k^{-1}q^{-\frac{1}{2}} & z^{-1}q^{-\frac{1}{2}}f   \\
   zq^{\frac{1}{2}}e &  zk^{-1}q^{\frac{1}{2}}-z^{-1}kq^{-\frac{1}{2}}
\end{array} \right)
\end{equation}

if $e, f, k$ commute as
\begin{eqnarray*}
ke&=&qek, ~~\ kf=q^{-1}fk, \\
ef-fe&=&(q-q^{-1})(k^2-k^{-2}).
\end{eqnarray*}
Denote this algebra $C_q$.

The element
\begin{equation}
c=fe+k^2q+k^{-2}q^{-1}
\end{equation}
generates the center of this algebra.

It is clear that this algebra quantizes the Poisson algebra
(\ref{br-1})(\ref{br-2}).  Indeed, the algebra $C_1$ is the commutative algebra generated by $e,f, k^{\pm 1}$. Consider the monomial basis $e^nk^mf^l$ in $C_q$.
Fix the isomorphism between $C_q$ and $C_1$ identifying these bases. The associative multiplication in $C_q$ is given
in this basis the function of $q$:
\[
e_ie_j=\sum_k m_{ij}^k(q)e_k
\]
it is clear that when $q=1$ this multiplication is the
usual multiplication in the commutative algebra generated by
$e,f, k, k^{-1}$. The skew symmetric part of the derivative of $m(q)$ at $q=1$
is the Poisson structure.

The algebra $C_q$ is a Hopf algebra with the comultiplication
acting on generators as
\[
\Delta k= k\otimes k, \ \ \Delta e=e\otimes k +k^{-1}\otimes e,
\Delta f=f\otimes k+ k^{-1}\otimes f
\]

The algebra $C_q$ is closely related to the quantized universal
enveloping algebra for $sl_2$. Indeed, elements $E=ek/(q-q^{-1}),
F=k^{-1}f/(q-q^{-1}), K=k^2$ are generators for $U_q(sl_2)$:
\[
KE=q^2EK, \ \ KF=q^{-2}FK, \ \ EF-FE=\frac{K-K^{-1}}{q-q^{-1}}
\]
with
\[
\Delta K=K\otimes K, \ \ \Delta E=E\otimes K+1\otimes E,
\Delta F= F\otimes 1+ K^{-1} \otimes F
\]

\subsubsection{} Assume that $q$ is generic. Denote by $V^{(m)}$ the irreducible $m+1$-dimensional representation of $C_q$, and by $v^{(m)}_0$ the highest weight vector in this representation:
\[
kv^{(m)}_0=q^{\frac{m}{2}} v^{(m)}_0, \ \ ev^{(m)}_0=0
\]
The weight basis in this representation can be obtained
by acting $f$ on the highest weight vector.
Properly normalized the action of $C_q$ on the weight basis is:
\[
kv^{(m)}_n=q^{\frac{m}{2}-n}v^{(m)}_n, \ \
fv^{(m)}_n=(q^{m-n}-q^{-m+n})v^{(m)}_{n+1}, \ \
ev^{(m)}_n=(q^n-q^{-n})v^{(m)}_{n-1}.
\]

The Casimir element $c$ acts on $V{(m)}$ by the multiplication
on $q^{m+1}+q^{-m-1}$.

Because the algebra $C_q$ is a Hopf algebra, it acts naturally on
the tensor product of representations.
\subsubsection{} Denote the matrix (\ref{q-L}) evaluated in the irreducible representation $V^{(m)}$ by ${\mathcal L}^{(m)}(z)$.
It is easy to check that it satisfies the following identities:
\[
{\mathcal L}^{(m)}(z^{-1})^t=-D_zC{\mathcal L}^{(m)}(zq)C^{-1}D_z^{-1}
\]
\[
{\mathcal L}^{(m)}(z)^T{\mathcal L}^{(m)}(z^{-1})=(zq^{\frac{m+1}{2}}-z^{-1}q^{-\frac{m+1}{2}})(z^{-1}q^{\frac{m+1}{2}}-zq^{-\frac{m+1}{2}})I
\]
where $D_z$ is a diagonal matrix.

Here $t$ is the transposition with respect to the
standard scalar product in $\CC^2$ and $T$ is the transposition
$t$ combined with the transposition in $V^{(m)}$ with respect to
the scalar product $(v^{(m)}_n,v^{(m)}_{n'})=\delta_{n,n'}$.

\subsubsection{} The matrix (\ref{q-L}) defines a family of $2$-dimensional representations of $C_q(\widehat{SL_2})$.
If $a$ is a non-zero complex number such representation is
\[
{\mathcal T}(z)\mapsto g_m(z){\mathcal L}^{(m)}(za)
\]
where the factor $g_m(z)$ is important only if we
want to satisfy the second relations which play the
role of quantum counter-parts of the the unimodularity
(i.e. $det=1$) of ${\mathcal T}(z)$.
\[
g_m(z)=-z\frac{(t^{-1}q^3z^2;q^4)_\infty (tq^5z^2;q^4)_\infty}
{(t^{-1}qz^2;q^4)_\infty (tq^3z^2;q^4)_\infty}
\]

\subsubsection{} The comultiplication defines the tensor
product of irreducible representations described above
\begin{equation} \label{T-repr}
{\mathcal T}(z)\mapsto \prod_{i=1}^Ng_{m_i}(z/a_i)T^{(m_1,\dots,m_N)}(z|a_1,\dots,a_N)
\end{equation}
where
\begin{equation} \label{T-matr}
T^{(m_1,\dots,m_N)}(z|a_1,\dots,a_N) ={\mathcal L}_1(z/a_1) \dots
{\mathcal L}_N(z/a_N)  ,
\end{equation}
where we have taken the matrix product of the matrices ${\mathcal L}(z)$ as
$2\times 2$ matrices. The $n$-th factor in \eqref{T-matr}
acts on the $n$-th factor of $V^{(m_1)}\otimes \dots \otimes V^{(m_N)}$.

Notice that this is also a tensor product of representations of $C_q$.

\subsubsection{Real forms}\label{q-r-form} As in the classical case there are two
real form of the algebra $C_q$ which are important for
finite dimensional spin chains.

Recall that a $*$-involution of a complex associative algebra
is an ani-involution of
the algebra, i.e. $(ab)^*=b^*a^*$ which is complex anti-linear:
$(\lambda a)^*=\bar{\lambda}a^*$. A real form of a complex $A$
corresponding to this involutions is real algebra which is the real subspace in $A$ spanned by the $*$-invariant
elements.

When $|q|=1$, we will write $q=\exp(i\gamma)$, the relevant real form of $C_q$, is characterized by the $*$-involution which acts on generators as
\[
e^*=f, \ \ k^*=k^{-1}
\]

When $q$ is real positive we will write $q=\exp(\eta)$.
In this case  the relevant real form is characterized by the $*$ involution acting on generators as:
\[
e^*=f, \ \ f^*=e, \ \ k^*=k
\]

As $\gamma\to 0$ or $\eta\to 0$ these real forms become real forms of the corresponding Poisson algebras described in section
\ref{cl-r-form} with $\epsilon=i$ and
$\epsilon=1$ respectively.

\subsection{The fusion of $R$-matrices and the degeneration of tensor products of irreducibles}

This section is the analog of the construction of
quantum $L$-operators by taking the tensor product of
2-dimensional representations.

Consider the product of $R$-matrices acting in the tensor
product of $n+m$ copies of $\CC^2$:
\[
\begin{array}{ccccc}R_{1'2'\dots m',12\dots n}(z)=&
 R_{1'1}(z)&R_{1'2}(zq)&\dots & R_{1'n}(zq^{n-1})\\
 & R_{2'1}(zq)&R_{2'2}(zq^2)&\dots & R_{2'n}(zq^{n})\\
 & \dots& \dots &\dots & \dots\\
 & R_{m'1}(zq^{m-1})&R_{m'2}(zq^m)&\dots & R_{m'n}(zq^{n+m-2})
\end{array}
\]

The operator $R(z)$ satisfies the identities
\[
R(z)R^t(z^{-1})=(zq-z^{-1}q^{-1})(z^{-1}q-zq^{-1})
\]
\[
PR(z)P=R^t(z)
\]
where $t$ is the transpiration operation (with respect
to the tensor product of the standard basis in $\CC^2$), and
\[
det(R(z))=(zq-z^{-1}q^{-1})^3(z^{-1}q-zq^{-1})
\]
From here we conclude that the matrix $PR(z)$ degenerates
at $z=q$ and $z=q^{-1}$ to the matrices of rank $3$ and $1$
respectively.

Define
\begin{equation}\label{p-proj}
\begin{array}{ccccc}P^{+}_{12\dots n}=&\VR_{12}(q)& \VR_{23}(q^2)& \dots & \VR_{1n}(q^{n-1})\\
 & & \VR_{23}(q) & \dots & \VR_{2n}(q^{n-2})\\
 & & & \dots & \\
 & & &  & \VR_{n-1 n}(q)
\end{array}
\end{equation}
where $\VR(z)=PR(z)$ and $P$ is the permutation matrix,
$P(x\otimes y)=y\otimes x$.

Consider $(\CC^2)^{\otimes n}$ as a representation of
$C_q$. Because all finite dimensional representations of
this algebra are completely reducible, it decomposes
into the direct sum of irreducible components. The irreducible
representation $V^{(n)}$ appears in this decomposition
with multiplicity $1$. One can show that the operator (\ref{p-proj}) is the orthogonal projector to $V^{(n)}$.
The proof can be found in \cite{KRS}.

Also, it is not difficult to show that
\begin{multline}
P^+_{12\dots n}R_{1',12\dots n}(zq^{-\frac{n-1}{2}})P^+_{12\dots n}=\\ (zq^{-\frac{n-3}{2}}-z^{-1}q^{\frac{n-3}{2}})\dots (zq^{\frac{n-1}{2}}-z^{-1}q^{-\frac{n-1}{2}}) R_{1',[12\dots n]}^{(1,n)}(z)
\end{multline}
where the linear operator $R^{1,n}(z)$ acts in $\CC^2\otimes V^{(n)}$ and the second factor appears as the $q$-symmetrized part
of the tensor product ${\CC^2}^{\otimes N}$. Moreover, it is easy to show that this operators is conjugate by a diagonal matrix to
${\mathcal L}(z)$:
\[
R^{(1,n)}(z)\simeq \left(
\begin{array}{cccc}
  zkq^{\frac{1}{2}}-z^{-1}k^{-1}q^{-\frac{1}{2}} & z^{-1}q^{-\frac{1}{2}}f  \\
   zq^{\frac{1}{2}}e &  zk^{-1}q^{\frac{1}{2}}-z^{-1}kq^{-\frac{1}{2}}
\end{array} \right)
\]
where $e,f, k$ act in the $n+1$ dimensional irreducible representation as it is described in section \ref{irreps}. In this realization of the irreducible representation weight vectors
appear as $v^{(n)}_k=P^{+}_{1\dots n}e_1\otimes e_1\otimes e_2
\otimes e_2$ where we have $n-k$ copies of $e_1$ and $k$ copies of $e_2$ in this tensor product.

Similarly
\[
\begin{array}{cccc} P^+_{1'2'\dots m'}P^+_{12\dots n} & R_{1'1}(zq^{-\frac{n+m-2}{2}})&\dots & R_{1'n}(zq^{\frac{n-m}{2}})\\
 &R_{2'1}(zq^{-\frac{n+m-4}{2}})&\dots & R_{2'n}(zq^{\frac{n-m+2}{2}})\\
 & \dots & \dots & \dots \\
 & P^+_{1'2'\dots m'}& P^+_{12\dots n}=&
 \end{array}
 \]
\[
\begin{array}{ccc}
(zq^{-\frac{n+m-4}{2}}-z^{-1}q^{\frac{n+m-4}{2}})&\dots & (zq^{-\frac{n-m}{2}}-z^{-1}q^{-\frac{n-m}{2}})\\
\dots & \dots & \dots \\
(zq^{-\frac{n-m-2}{2}}-z^{-1}q^{-\frac{n-m-2}{2}})&\dots & (zq^{\frac{n+m-2}{2}}-z^{-1}q^{-\frac{n+m-2}{2}})\\
 & & R^{(m,n)}(z)
\end{array}
\]
Here we assume that $m<n$. The matrix elements of $R^{(m,n)}(z)$
are Laurent polynomials of the form $z^{-m}P(z^2)$ where $P(t)$
is a polynomial of degree $m$.

The matrix $R^{(n,m)}$ also can be expressed in terms $e, f, k^{\pm 1}$.

The matrices $R^{(k,l)}$ satisfy the the Yang-Baxter equation
\[
R_{12}^{(l,m)}(z)R_{13}^{(l,n)}(zw)R_{23}^{(m,n)}(w)= R_{23}^{(m,n)}(w)R_{13}^{(l,n)}(zw)R_{12}^{(l,m)}(z)
\]

In addition to this they satisfy identities
\[
R^{(l,m)}(z)R^{(m,l)}(z^{-1})^T=s_{ml}(z)s_{ml}(z^{-1})
\]
and
\[
R_{12}^{(l,m)}(z)^{t_1}=(-1)^l (D_zC^{(m)}\otimes 1)R_{12}^{(l,m)}(zq)( D_z^{-1}{C^{(m)}}^{-1}\otimes 1)
\]
where $s_{ml}(z)$ is a Laurent polynomial in $z$ which is easy to
compute, $C^{(m)}=P^+_{12\dots n}\otimes C P^+_{12\dots n}$,
 and we assume that $l\leq m$.

\subsection{Higher transfer-matrices}
Define $C_q(\widehat{SL}_2)$-valued matrices
\[
{\mathcal T}^{(m)}(z)=P^+_{12\dots n}{\mathcal T}_1(zq^{{m-1\over 2}})\dots
{\mathcal T}_m(zq^{-{m-1\over 2}})P^+_{12\dots n}
\]
where $P^+_{12\dots n}$ is defined above and ${\mathcal T}$ is the
matrix (\ref{gen-aff}).

They satisfy the relations
\[
R^{(l,m)}(z)_{12}{\mathcal T}^{(l)}_1(zw){\mathcal T}^{(m)}_2(w)=
{\mathcal T}^{(m)}_2(w){\mathcal T}^{(l)}_1(zw)R^{(l,m)}(z)_{12}
\]

For non-zero $d$ define the following elements of $C_q(\widehat{SL}_2)$
\[
\tau_\ell(z) = ({\mbox{id}}\otimes tr_{V^{(\ell)}})
({\mathcal T}^{(\ell)}(z)d^{(l)}) \ .
\]
where $d^{(l)}=diag(d^l,d^{l-2},\dots, d^{-l})$ and $d\neq 0$.

The fusion relations for ${\mathcal T}(z)$  imply the following
recursive relations for $\tau_\ell(z)$:
\[
\tau_1(z)\tau_\ell(zq) = \tau_{\ell +1}(z)+\tau_{\ell -1}(zq^2)
\]
which can be solved in terms of determinants \cite{BR}:
\[
\tau_\ell(z) = {\mbox{det}} \left(
\begin{array}{cccc}
  \tau_1(z) & 1 & {} & 0 \\
  1 & \tau_1(zq) & {\ddots} & {} \\
 {} & \ddots & \ddots & {1} \\
  0 & {} & 1 & \tau_1(zq^{\ell -1})
\end{array}\right).
\]
The remarkable fact is that elements $\{\tau_\ell(z)\}$ also satisfy
another set of relations which also follow from the fusion
relations:
\begin{equation}  \label{F}
\tau_\ell(zq^{\frac 12})\tau_\ell(zq^{-\frac 12}) =
\tau_{\ell +1}(z) \tau_{\ell -1}(z)+1.
\end{equation}

\subsection{Local integrable quantum spin Hamiltonians}
Transfer-matrices
\[
t_m(u)=tr_a(R^{m,m_1}_{a1}(z/a_1)\dots R^{m,m_1}_{a1}(z/a_1)d_a^{(l)})
\]
form a commuting family of operators
in $V^{(m_1)} \otimes \dots \otimes V^{(m_N)}$.
They quantize the generating functions for classical
spin chains and can be used to construct local quantum spin chains.
Below we outline two common constructions of local Hamiltonians
from such transfer-matrices.

\subsubsection{Homogeneous $SU(2)$ spin chains}\label{local-homo}

The homogeneous Heisenberg model of spin $S$
corresponds to the choice $m_1=\dots =m_N=l$ and $a_N=\dots= a_1=1$.
We will denote corresponding transfer matrices as $\tau^{(l)}_m(u)$

In the case $m=1$, the transfer matrix $t^{(1)}_1(z)$ is
\begin{equation}  \label{6v}
t_1^{(1)}(z) = tr_0(R_{0N}(z)\dots R_{01}(z))
\end{equation}
where $R(z)$ is the matrix \eqref{Rmatrix}.

The linear operator \eqref{6v} is the transfer matrix of the 6-vertex model
\cite{LiebM} \cite{Ba}. It is also a generating function for local spin
Hamiltonians:
\begin{equation}  \label{loc-ham-1}
H_1 = \frac{d}{du} \log (t_1^{(1)}(u))|_{u=0}~ =~ \sum^N_{n=1}
(\sigma^x_n \sigma^y_{n+1} +\sigma^y_n \sigma^y_{n+1} +
\Delta\sigma^z_n \sigma^z_{n+1}),
\end{equation}
\begin{equation}  \label{loc-ham}
H_k =(\frac{d}{du})^k \log t_1^{(1)}(u)|_{u=0} =
\sum^N_{n=1} H^{(k)} (\sigma_n,\dots ,\sigma_{n+k}).
\end{equation}
Here $\sigma^x,\sigma^y,\sigma^z$ are Pauli matrices
and $\sigma_n$ is the collection of Pauli matrices acting nontrivially in the n-th factor of the tensor product.
\begin{equation}\label{sigma-m}
\sigma^x=\left(\begin{array}{ll} 0 & 1 \\ 1 & 0
\end{array}\right) \ ,
\qquad
\sigma^y=\left(\begin{array}{ll} 0 & i \\ -i & 0
\end{array}\right) \  ,
\qquad
\sigma^z=\left(\begin{array}{ll} 1 & 0 \\ 0 & -1
\end{array}\right) \ .
\end{equation}

If $l>1$ similar analysis can be done for the transfer-matrix
\[
t^{(l)}_l(z)=tr_a(R_{a1}^{(l,l)}(z)\dots R_{aN}^{(l,l)}(z))
\]
Since $R^{(l,l)}(1)=P$, we have:
\[
t^{(l)}_l(1)=tr_a(P_{a1}\dots P_{aN})=tr_a(P_{12}P_{13}\dots P_{1N}P_{a1})=P_{12}P_{13}\dots P_{1N}
\]
Here we used the identities $P_{a1}A_aP_{a1}=A_1$
and $tr_a(P_{a1})=I_1$.

The operator $T=t^{(l)}_l(1)$ is the translation matrix:
\[
T(x_1\otimes x_2\dots \otimes x_n)=x_N\otimes x_1\otimes \dots x_{N-1}
\]

Differentiating $t^{(l)}_l(z)$ at $z=1$ we have:
\[
{t^{(l)}_l(1)}'=\sum_{i=1}^N tr_a(P_{a1}\dots P_{ai-1}{R^{(l,l)}_{ai}(1)}'
P_{ai+1}\dots P_{aN})=T\sum_{i=1}^N H^{(l)}_{ii+1}
\]

Similarly, taking higher logarithmic derivatives  of $t_{l}^{(l)}(z)$ at $z=1$ we will have
higher local Hamiltonians acting in
$({\Bbb C}^{l+1})^{\otimes N}$:
\[
H_k =(z\frac{d}{dz})^k t^{(l)}_{l}(z)|_{z=1} =
\sum^N_{n=1} H^{(k)}_{n,\dots ,n+k} \ .
\]
Here the matrix $H^{(k)}$ acts in
$({\Bbb C}^{l+1})^{\otimes {k+1}}$.
The subindices show how this matrix acts in
$({\Bbb C}^{l+1})^{\otimes N}$.

One can show that these local quantum spin chain Hamiltonians in the limit $q\to 1$ and $l\to \infty$ become classical Hamiltonians
described in section \ref{cl-Ham-sec}, assuming that $q^l$ is fixed.

\subsubsection{Inhomogeneous $SU(2)$ spin chains}
The construction using degeneration points.

The construction of inhomogeneous local operators
is easy to illustrate on the spin chain where the inhomogeneities
alternate.
\[
t_m(z)=tr_a(R_{a1}^{(m,l_1)}(za^{-1})R_{a2}^{(m,l_2)}(za)\dots R_{a,2N-1}^{(m,l_1)}(za^{-1})R_{a,2N}^{(m,l_2)}(za))
\]

Now we have two sublattices and two translation operators
\[
T_{even}=P_{24}P_{26}\dots P_{2,2N}, \ \  T_{odd}=P_{13}P_{15}\dots P_{1,2N-1}
\]

It is easy to find the following special values of
the transfer-matrix:
\[
t_{l_1}(a)=T^{even}R^{(l_2,l_1)}_{21}(a^{-2})\dots R^{(l_2,l_1)}_{2N,2N-1}(a^{-2})
\]
\[
t_{l_2}(a^{-1})=R^{(l_1,l_2)}_{12}(a^{2})\dots R^{(l_1,l_2)}_{2N-1,2N}(a^{2})T^{odd}
\]
These operators commute and
\[
t_{l_1}(a)t_{l_2}(a^{-1})=T^{even}T^{odd}
\]
\begin{multline}
t_{l_1}(a)t_{l_2}(a^{-1})^{-1}=T^{even}(T^{odd})^{-1}\\
R^{(l_2,l_1)}_{2,2N-1}(a^{-2})R^{(l_2,l_1)}_{4,1}(a^{-2})\dots R^{(l_2,l_1)}_{2N,2N-3}(a^{-2})
R^{(l_1,l_2)}_{2N-1,2N}(a^{-2})\dots R^{(l_1,l_2)}_{1,2}(a^{-2})
\end{multline}
Taking logarithmic derivatives of $t_{l_1}(z)$ at $z=a$ and of $t_{l_2}(z)$ at $z=a^{-1}$ we again will have local operators, for
example:
\begin{multline}
z\frac{d}{dz} \log t_{l_1}(z)|_{z=a}=\sum_{n=1}^N R^{-1}_{2n+1,2n}(a^2)R'_{2n+1,2n}(a^2)\\
+\sum_{n=1}^N R^{-1}_{2n+1,2n}(a^2)P_{2n+1,2n-1}R'_{2n+1,2n-1}(1)R_{2n+1,2n}(a^2)
\end{multline}

These Hamiltonians in the semiclassical limit reproduce inhomogeneous classical spin chains described earlier.

\section{The spectrum of transfer-matrices}\label{Bethe-ans}

\subsection{Diagonalizability of transfer-matrices}

Assume that $q$ is real. Let $t(u)^*$ be the Hermitian conjugation
of $t(u)$ with respect to the standard Hermitian scalar product on $(\CC^2)^{\otimes N}$. It is easy to prove, using the identities
for $R(u)$ that
\[
t(z|a_1,\dots, a_N)^*=(-1)^N t(\overline{z}^{-1}q^{-1}|\overline{a_1}^{-1}\dots \overline{a_N}^{-1})
\]
where $\overline{z}$ is the complex conjugate to $z$.

Because $t(z)$ is the commutative family of operators,
the operators $t(z)$ is normal when $\overline{a_i}=a_i^{-1}$.
Therefore for these values of $a_i$ it is diagonalizable.
Since $t(z)$ is linear (up to a scalar factor) in $a_i^2$,
this imply that $t(z)$ is diagonalizable for all generic complex
values of $a_i$, and for the same reasons for all generic
complex values of $q$.

\subsection{Bethe ansatz for $sl_2$}
In this section we will recall the algebraic Bethe ansatz
for the inhomogeneous finite dimensional spin chain.

The quantum monodromy matrix for such spin chain is:
\begin{equation}\label{mon-matr}
T(z)=\cL_1^{(m_1)}(z/a_1)\dots \cL_N^{(m_N)}(z/a_N)D
\end{equation}
where
\[
D=\left(\begin{array}{cc} Z & 0 \\ 0 & Z^{-1}\end{array}\right)
\]

It is convenient to write it as
\[
T(z)=\left(\begin{array}{cc} A(z) & B(z) \\ C(z) & D(z)
\end{array}\right)
\]
In the basis $e_1\otimes e_1,e_1\otimes e_2,e_2\otimes e_1,e_2\otimes e_2$ of the tensor product $\CC^2\otimes \CC^2$
we have:
\[
T_1(zw)T_2(w)=\left(\begin{array}{cccc} A(zw)A(w) & A(zw)B(w) & B(zw)A(w) & B(zw)B(w) \\
A(zw)C(w) & A(zw)D(w) & B(zw)C(w) & B(zw)D(w)\\
C(zw)A(w)& C(zw)B(w) & D(zw)A(w) & D(zw)B(w) \\
C(zw)C(w) & C(zw)D(w) & D(zw)C(w) & D(zw)D(w)
\end{array}\right)
\]

Writing the $R$-matrix in the tensor product basis
as in (\ref{Rmatrix})
\[
R(z)=\left(\begin{array}{cccc} 1 & 0 & 0 & 0 \\
0 & f(z) & g(z)z^{-1} & 0\\
0 & g(z)z & f(z) & 0 \\
0 & 0 & 0 & 1
\end{array}\right)
\]
the commutation relations (\ref{R-T}) produce the following relations between $A$ and $B$ and $D$ and $B$:
\[
A(z)B(v)={1\over f(vz^{-1})}B(v)A(z)-{g(vz^{-1})zv^{-1}\over f(vz^{-1})}B(z)A(v)
\]
\[
D(z)B(v)={ 1\over f(zv^{-1})}B(v)D(z)-{g(zv^{-1})zv^{-1}\over f(zv^{-1})}B(z)D(v)
\]
where $f(z)={z-z^{-1}\over zq-z^{-1}q^{-1}}, g(z)={q-q^{-1}\over zq-z^{-1}q^{-1}}$.

The $L$-operators act on the vector
\[
\Omega=v_0^{(m_1)}\otimes v_0^{(m_2)}\otimes \dots\otimes v_0^{(m_N)}
\]
in a special way:
\[
\cL_i(z)\Omega=\left(\begin{array}{cc} (zq^{m_i+1\over 2}-z^{-1}q^{-{m_i+1\over 2}})\Omega & * \\ 0 &  (zq^{-{m_i-1\over 2}}-z^{-1}q^{{m_i-1\over 2}})\Omega
\end{array}\right)
\]

From here it is clear that $\Omega$ is an eigenvector
for operators $A$ and $D$ and that $C$ annihilates it:
\[
T(z)\Omega=\left(\begin{array}{cc} \alpha(z)\Omega & * \\ 0 &  \delta(z)\Omega
\end{array}\right)
\]
where
\[
\alpha(z)=Z\prod_{i=1}^N(za_i^{-1}q^{m_i+1\over 2}-z^{-1}a_iq^{-{m_i+1\over 2}})
\]
\[
\delta(z)=Z^{-1}\prod_{i=1}^N(za_i^{-1}q^{-{m_i-1\over 2}}-z^{-1}a_iq^{m_i-1\over 2})
\]

The details of proof of the following construction of eigenvectors can be found in \cite{FT}.
\begin{theorem}
The following identity holds
\[
(A(z)+D(z))B(v_1)\dots B(v_n)\Omega=\Lambda(z|\{v_i\})B(v_1)\dots B(v_n)\Omega
\]
where
\begin{equation}\label{BE-eigenv}
\Lambda(z|\{v_i\})=\alpha(z)\prod_{i=1}^n{v_iz^{-1}q-v_i^{-1}zq^{-1}\over v_iz^{-1}-v_i^{-1}z}+\delta(z)\prod_{i=1}^n{v_i^{-1}zq-v_iz^{-1}q^{-1}\over v_i^{-1}z-v_iz^{-1}}
\end{equation}
if the numbers $v_i$ satisfy the Bethe equations:
\begin{equation}\label{BE}
\prod_{\alpha=1}^N{v_ia_\alpha^{-1}q^{m_\alpha+1\over 2}-z^{-1}a_\alpha q^{-{m_\alpha+1\over 2}}\over v_ia_\alpha^{-1}q^{-{m_\alpha-1\over 2}}-z^{-1}a_\alpha q^{m_\alpha-1\over 2}}=-Z^2\prod_{j=1}^n {v_iv_j^{-1}q-v_i^{-1}v_jq^{-1}\over v_iv_j^{-1}q^{-1}-v_i^{-1}v_jq}
\end{equation}
\end{theorem}

Note that the formula for the eigenvalues in terms
of solutions to Bethe equations is a rational function.
Bethe equations can be regarded as conditions
\[
res_{z=v_j} \Lambda(z|\{v_i\})=0
\]
This agrees with the fact that $t(z)$ is a commuting family
of operators which has no poles at finite $z$.

\subsection{The completeness of Bethe vectors}
The next step is to establish whether the construction outlined above all eigenvectors. We will focus here on the spin chain of
spin $1/2$.

Assume that $q$, $e^{2H}$ , and inhomogeneities $a_i$ are generic.
Let us demonstrate that the vectors
\begin{equation}\label{bethe}
B(v_1)\dots B(v_n)\Omega
\end{equation}
where $v_i$ are solutions to Bethe equations give all $2^N$ eigenvectors of the transfer-matrix.

\subsubsection{} Consider the limit of (\ref{bethe}) when $a_N\to \infty$.

Assume that $v$ is fixed and $a_N\to \infty$. From the definition of $B(v)$ we have:
\[
B(v)=a_Nq^{-\frac{1}{2}}v^{-1}(\tilde{A}(v)\otimes f-\tilde{B}(v)\otimes K)(1+o(1))
\]
where $\tilde{A}(v)$, $\tilde{B}(v)$ are elements of the quantum monodromy
matrix (\ref{mon-matr}) with only $N-1$ first factors.

On the other hand if $a_N\to \infty$ and $v\to \infty$ such that $v=wa_N$ and $w$ is finite the asymptotic is different:
\[
B(v)=w^{-1}q^{-\frac{1}{2}}\prod_{n=1}^{N}(wa_Na_n^{-1}q^{\frac{1}{2}})
k\otimes\dots\otimes k\otimes f (1+o(1))
\]

From here we obtain the asymptotic of Bethe vectors
when all $v_i$ are fixed and $a_N\to \infty$
\begin{multline}\label{as-B-1}
B(v_1)\dots B(v_n)\Omega_N\to q^{\frac{m_N+1}{2}n}(-a_Nq^{\frac{1}{2}})^n\prod_{i=1}^n v_i^{-1}\\(\tilde{B}(v_1)\dots \tilde{B}(v_{n})\Omega_{N-1}\otimes v_0^{(m_N)}-\sum_{i=1}^n q^{-\frac{m_N+1}{2}-i+1}\\
\tilde{B}(v_1)\dots \tilde{A}(v_i)\dots\tilde{B}(v_{n})\Omega_{N-1}\otimes f v_0^{(m_N)})(1+o(1)).
\end{multline}
Similarly, when $v_1,\dots, v_{n-1}$ are fixed
and $a_N\to \infty$ such that $v_n=wa_N$ we have
\begin{multline}\label{as-B-2}
B(v_1)\dots B(v_n)\Omega_N\to q^{\frac{m_N+1}{2}(n-1)-n}(-a_Nq^{\frac{1}{2}})^{n-1}\prod_{i=1}^{n-1} v_i^{-1}q^{\frac{m_i}{2}}\\
\tilde{B}(v_1)\dots \tilde{B}(v_{n-1})\Omega_{N-1}\otimes fv_0^{(m_N)}(1+o(1))
\end{multline}

\subsubsection{}

Solutions to (\ref{BE}) have the following possible asymptotic when $a_N\to \infty$ :

1. For all $j=1,\dots, N$, $\lim_{a_N\to \infty}v_j=v'_j$ where
$\{v'_j\}$ is a solution to the Bethe system for the spin chain of length $N-1$ with inhomogeneities $a_1,\dots, a_{N-1}$ and $Z$.

2. For one of $v_j$'s, say for $v_n$ we have $v_n=a_Nw+O(1)$
and for others  $\lim_{a_N\to \infty}v_j=v'_j$ where
$\{v'_j\}$ is a solution to the Bethe system for the spin chain of length $N-1$ with inhomogeneities $a_1,\dots, a_{N-1}$ and $Zq^{-1}$.
From the Bethe equation for $v_n$ we have
\[
w^2={1-Z^2q^{-N+2n}\over q^2-Zq^{-N+2n}}
\]

3. More then one of $v_i$ is proportional to $a_N$.

Using induction and the asymptotic of Bethe vectors (\ref{as-B-1})
and (\ref{as-B-2})
it is easy to show that only first two options describe the
spectrum of the spin $1/2$ transfer-matrix. Similar arguments were used in \cite{KKR} to prove the completeness of Bethe vectors in
an $SL_n$ spin chain.

This implies immediately that there are $\left(\begin{array}{c} N   \\ n \end{array} \right)$ Bethe vectors
for each $0\leq n\leq N$. And that the total number of Bethe vectors is
\[
2^N=\sum_{n=0}^N \left(\begin{array}{c} N   \\ n \end{array} \right)
\]
Other solutions to Bethe equations describe eigenvectors in infinite dimensional representations of quantized affine algebra with the
same weights. They do not correspond to any eigenvectors of
the inhomogeneous spin $1/2$.

\subsubsection{} For special values of $a_\alpha$ the solutions to
the Bethe equations may degenerate (a level crossing may occur
in the spectrum of $t(z)$). In this case the Bethe ansatz should involve derivatives of vectors (\ref{bethe}).

\section{The thermodynamical limit}

The procedure of "filling Dirac seas" is a way to construct
physical vacua and the eigenvalues of quantum integrals of motion
in integrable spin chains solvable by Bethe ansatz.

To be specific, consider the homogenous spin chain of spins $1/2$.
Let $H_1, H_2, \dots$ be quasilocal Hamiltonians described in the
section \ref{local-homo} with $q=\exp(\eta)$ for some real $\eta$.

Take the linear combination
\begin{equation}\label{l-Hamilt}
H(\lambda)=\sum_k H_k \lambda_k
\end{equation}
This operator is bounded. Let $\Omega_N(\lambda)$ be its
normalized ground state. As $N\to \infty$, matrix elements $(\Omega_N, a\Omega_N)$  converges to the state $\omega_\lambda$
on the inductive limit of the algebra of observables.
The action of local operators on $\omega_\lambda$ generate the
Hilbert space ${\mathcal H}$.

Since the eigenvalues of coefficients of $t(u)$ can be computed
in terms of solutions to Bethe equations, the spectrum of
these operators in the large $N$ is determined by the large $N$
asymptotic of solutions to Bethe equations.

The main assumption in the analysis of the Bethe equations
in the limit $N\to\infty$ is that the numbers $\{v_\alpha^{(0)}\}$
\footnote{solutions to the Bethe equations corresponding to the minimum
eigenvalue of $H(\lambda)$} are
distributed along the real line with some density $\rho(u)$.
The intervals where
$\rho(u)\neq 0$ are called Dirac seas. For Hermitian Hamiltonian
(\ref{l-Hamilt}) there is  strong evidence that Dirac seas
is a finite colelction of intervals $(B^+_1,B^-_1),\dots ,(B^+_n,B^-_n)$.
Here numbers $B_\alpha^\pm$ are boundaries of Dirac seas.
We assume they are increasing from left to right.
The boundaries of Dirac seas are uniquely determined by $\{\lambda_l\}$.

A solution to the Bethe equations is said to contain an $m$-string,
when as $N\to\infty$, there is a subset of $\{v_i\}$ of the form
\[
v^{(m)}+i\frac{\eta}{2}m,
v^{(m)}+i\frac{\eta}{2}(m-2),\dots ,v^{(m)}-i\frac{\eta}{2}(m-2),
v^{(m)}-i\frac{\eta}{2}m .
\]
with some real $v^{(m)}$.

The excitations over the ground states can be of the following types:
\begin{itemize}
\item A hole in the Dirac sea $(B^+_k,B^-_k)$ correspond to the
solution to Bethe equations which has one less number $v_i$, and
as $N\to\infty$ the remaining $v_i$ ``fill'' the same Dirac seas
with the the densities deformed by the fact that one of the numbers
$\{v_\alpha^{(0)}\}$ is missing and other are ``deformed'' by the
missing one. The number which is ``missing'' is a state with a hole
is the ``rapidity'' $v\in (B^+_k,B^-_k)$ of the hole.
\item Particles correspond to ``adding''
one real number to the collection $\{v_\alpha^{(0)}\}$).
\item $m$-strings, $m > 1$  (corresponding to adding one $m$-string
solution to the collection $\{v_\alpha^{(0)}\}$).
\end{itemize}

There are convincing arguments, that the Fock space of the system with
the Hamiltonian (\ref{l-Hamilt}) has the follwoing structure.
It has a vacuum state $\O_\lambda$ corresponding to the solution
of the Bethe equations with the minimal eigenvaue of $H(\lambda)$.
Excited states are eigenvectors of the Hamiltonian (and of all
other integrals) which corresond to solutions of Bethe
euqations with finitely many holes, particles, and $m$-strings. It has the
following structure
\begin{eqnarray*}
{\mathcal H}(\{B^+_i,B^-_i\}^k_{i=1}) & = &
\bigoplus_{ N_h\geq 0, N_p\geq 0,
N\geq 0}
\bigoplus_{ n_1+\dots +n_k=N_h,N_p} \\
& \cdot & \bigotimes^k_{j=1}
L_2^{\mbox{symm}}( {I_1^+}^{\times n_1}\times\dots {I_k^+}^{\times n_k}\times
I^{N_p}\times
{S^1}^{N})
\end{eqnarray*}
Here we used the notation $I_l^+=(B^+_l,B^-_l), I_l^-=(B^+_{l-1},B^-_l)$,
where $B^+_0=-\pi$, and $B^-_{k+1}=\pi$; $n_k$ is the number
of holes in the Dirac sea $(B^+_k,B^-_k)$, $N_p$ is the number of
particles; $I$ is the complement to the Dirac seas, $N$ is the number
of strings with $m\geq 2$.  The symbol ``symm'' means
certain symmetrization procedure which we will not discuss here
(see, for example, \cite{JM} for a discussion of the aniferromagnetic
ground state).

Varying $\{\lambda_k\}$, or equivalently, positions $\{ B_k^\pm\}$ of
Dirac seas we obtain a ``large'' part of the space of states of the
spin chain in the limit $N\to \infty$:
$({\Bbb C}^2)^{\otimes N}$ in the limit $N\to\infty$ into
the direct integral of
separable Hilbert spaces:
\begin{equation}       \label{decomp}
({\Bbb C}^2)^{\otimes N} \to \bigoplus_{k\geq 0}
\int^\oplus_{[-\pi,\pi]^{\times2k}}
{\mathcal H}(\{B^+_i,B^-_i\}^k_{i=1})
\end{equation}

\section{The 6-vertex  model}

\subsection{The 6-vertex configurations and boundary conditions}
The $6$-vertex model is a model is statistical mechanics
where states are configurations of arrows on a square planar grid, see an example on Fig. \ref{dwbc_conf}. The weights are assigned
to vertices of the grid. They depend on the arrows on edges surrounding the vertex non-zero weights correspond to the
configurations on Fig. \ref{vertices}, to configurations where
the number of incoming arrows in equal to the number of outgoing
arrows. This is also known as the ice rule \cite{LW}.

\begin{figure}[t]
\begin{center}
\includegraphics{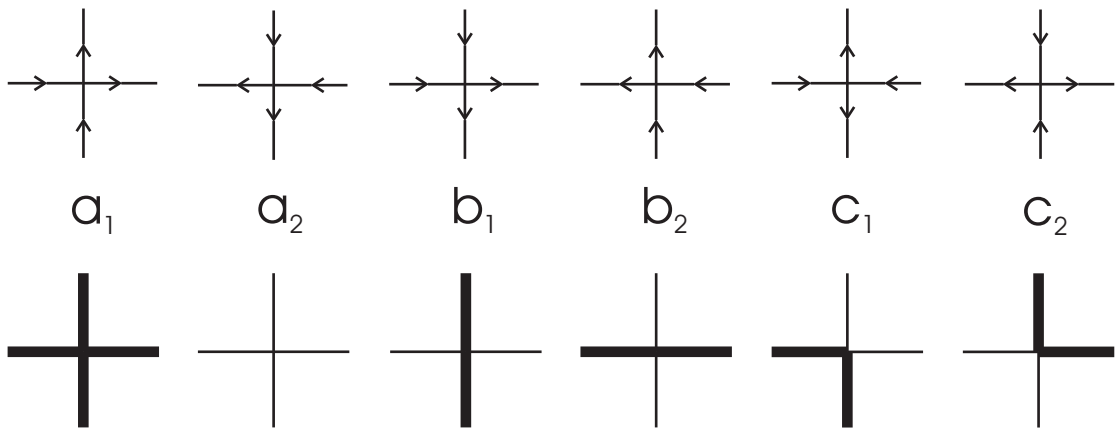}
\end{center}
\caption{The $6$ types of vertices and the corresponding
thin and thick edges configurations.}
\label{vertices}
\end{figure}

Each configuration of arrows on the lattice can be equivalently
described as the configuration of ``thin'' and ``thick'' edges (or
``empty'' and ``occupied'' edges) as it is shown on figure \ref{vertices}.
There should be an even number of thick edges at each vertex as
a consequence of the ice rule.

The thick edges form paths going from NorthWest (NW) to SouthEast (SE). We assume that when there are 4 think edges meeting at a vertex, the
corresponding paths meet at this point and then are going apart.
So, equivalently, configurations of the $6$-vertex model can be
regarded as configurations of paths going from NW to SE satisfying the rules from Fig. \ref{vertices}\footnote[1]{One can consider such configurations on any 4-valent graph. But only for special graphs and special Boltzmann weights
one can compute the partition function pe site.}.

\subsection{Boundary conditions}
It is natural to consider the 6-vertex model on surface grids.

If the surface is a  domain on a plane we will say that an edge is
{\it outer} if it intersects the boundary of a domain. We assume that the boundary is chosen such that it intersects each edge at most once. Outer edges are attached to a 4-valent vertex by one side and to the boundary by the other side.

Fixed boundary conditions means that fixing the 6-vertex
configurations on outer edges. An example of fixed boundary conditions known as domain wall (DW) boundary conditions on a
square domain is shown on fig. \ref{dwbc_conf}.

We will be interested in three types of boundary conditions:

\begin{itemize}
\item A {\it domain} (connected simply connected on a plane)
with {\it fixed boundary conditions}, see Fig. \ref{domain}. We will
also call this Dirichlet boundary conditions.
\begin{figure}[t]
\begin{center}
\includegraphics[height=3cm,width=3cm]{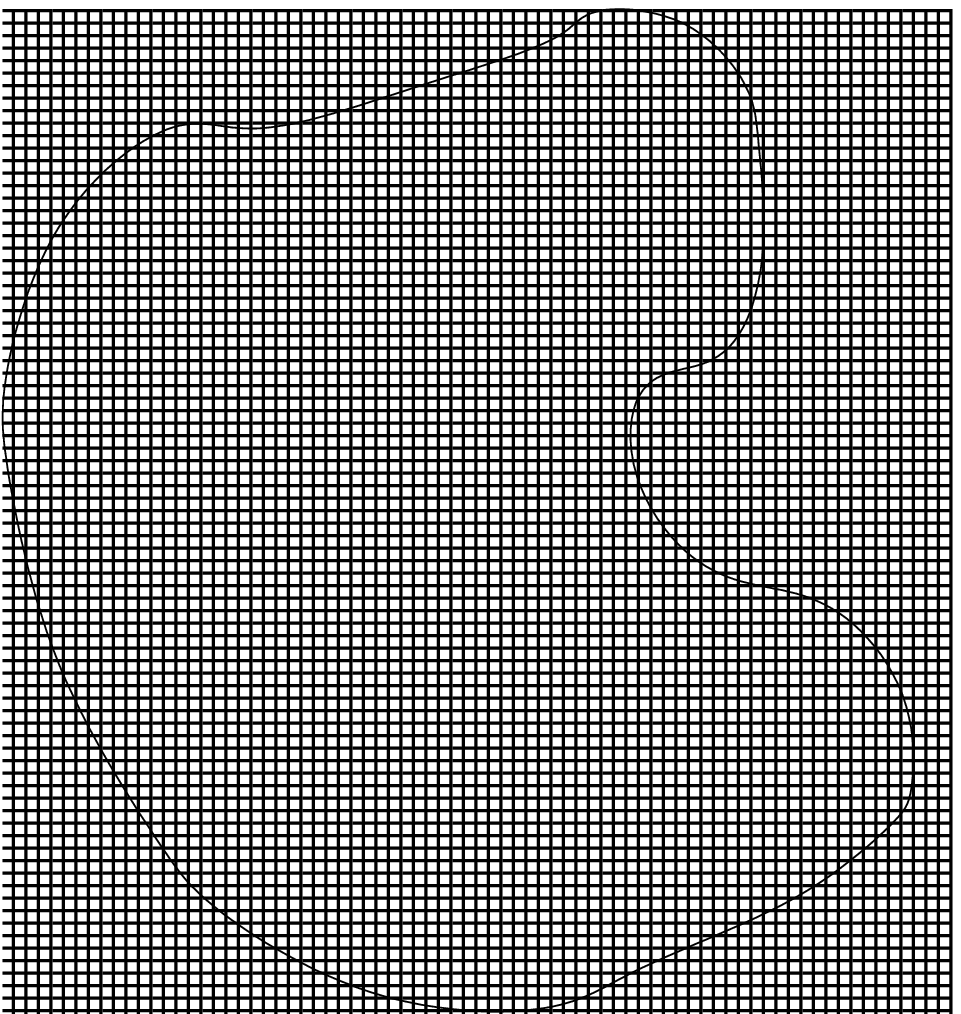}
\end{center}
\caption{A domain, connected, simply-connected.}
\label{domain}
\end{figure}
\item  A {\it cylinder with fixed boundary conditions}, see Fig. \ref{cyliner}. This
case can be regarded as a domain with states on outer edges of two
sides being identified and with fixed boundary conditions on other
sides.
\begin{figure}[t]
\begin{center}
\includegraphics[height=3cm,width=3cm]{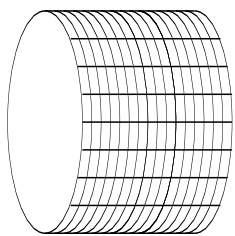}
\end{center}
\caption{Cylindric boundary conditions.}
\label{cylinder}
\end{figure}
\item Identification of states on outer edges of opposite sides of a
  rectangle gives the states for the 6-vertex model on a torus, see
  Fif. \ref{torus}. It is also
\begin{figure}[t]
\begin{center}
\includegraphics[height=3cm,width=3cm]{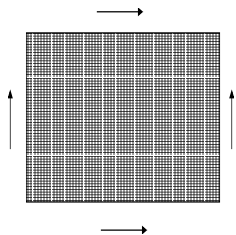}
\end{center}
\caption{Toric boundary conditions.}
\label{torus}
\end{figure}
known as the $6$-vertex model with periodical boundary conditions.
\end{itemize}

\subsection{The partition function and local correlation functions}

To each configuration  $a_1$, $a_2$, $b_1$, $b_2$, $c_1$, and $c_2$ on Fig.
\ref{vertices} we assign Boltzmann weights which
we denote by the same letters. The physical meaning of a Boltzmann weight is $\exp(-\frac{E}{T})$, where $E$ is the energy of a state and $T$ is the temperature (in appropriate units), so all numbers
$a_1$, $a_2$, $b_1$, $b_2$, $c_1$, and $c_2$ should be positive.

The weight of the configuration is the product of weights
corresponding to vertices inside the domain of weights assigned
to each vertex by the 6-vertex rules.

The 6-vertex model is called {\it homogeneous} if the weight
assigned to a vertex depends only on the configuration of arrows
on adjacent edges and not on the vertex itself.

When the weights also depend on the vertex the model is
called {\it inhomogeneous}.

The partition function is the sum of weights of all states of the model
$$
Z=\sum_{\rm states}\prod_{\rm vertices} w(\rm vertex),
$$
where $w(\rm vertex)$ is the weights of a vertex assigned according
to  Fig. \ref{vertices}.

Weights of states  define the probabilistic measure on the set of states of
the $6$-vertex model. The probability of a state is given by the
ratio of the weight of the state to the partition function of the
model
\begin{equation}\label{6v-prob}
P(state)=\frac{\prod_{\rm vertices} w(\rm vertex)}{Z}.
\end{equation}

The{\it characteristic function} of an edge $e$ is the function defined on the set of 6-vertex states
$$
\sigma_e(state)=\left\{
\begin{array}{cc} 1, & \ e\,\,\mbox{is occupied by a path},\\
0, &  \mbox{otherwise}.
\end{array} \right .
$$

A {\it local correlation function} is the expectation value of the product of such characteristic functions with respect to the measure (\ref{6v-prob}):
$$
\label{loc_cor}
\langle\sigma_{e_1}\sigma_{e_2}..\sigma_{e_n}\rangle
=\sum_{states}P(state)\prod_{i=1}^n\sigma_{e_i}(state).
$$

It is convenient to
write the Boltzmann weights in exponential form
\begin{eqnarray}
&& a_1=ae^{H+V}, \qquad a_2=ae^{-H-V}, \nonumber
\\
&& b_1=be^{+H-V}, \qquad b_2=be^{-H+V},
\nonumber
\\
&& c_1=ce^{-E}, \qquad c_2=ce^{E}, \nonumber
\end{eqnarray}
From now on we will assume that $E=0$. If the weight
are homogenous (do not depend on a vertex), local correlation functions for a domain, cylinder or torus do not depend on $E$.
Also, the parameters $H$ and $V$ have a clear physical meaning,
they can be regarded as horizontal and vertical electric fields. Indeed, if $E=0$ the weight of the state
can be written as
\begin{multline}
w(S)=\prod_{v\in vertices} w(v|S) \\
\exp(\sum_{e\in E_h}
\sigma_e(S)(H(e_+)+H(e_-))/2+\sum_{e\in E_v}
\sigma_e(S)(V(e_+)+V(e_-))/2)
\end{multline}
Here $S$ is a state of the model, $E_h$ is a set of horizontal
edges and $E_v$ is the set of vertical edges, $w(v|S)$ is the weight of the vertex $v$ in the state $S$ where $a_1=a_2=a$ and
$b_1=b_2=b$. The symbol $\sigma_e(S)$ is the characteristic function of $e$: $\sigma_e(S)=1$ if the arrow pointing
up or to the left, and $\sigma_e(S)=-1$ if the arrow is pointing
down or to the right.

For a given {\it domain} let us chose its boundary edge and enumerate all other boundary edges counter-clock wise. The partition function for a domain for various fixed boundary conditions can be considered as a vector in $(\CC^2)^{\otimes N}$ where the factors in the tensor product counted from left to right correspond to enumerated boundary vertices.

Similarly, the partition function for a (vertical) {\it cylinder} of size $N\times M$ (horizontal, vertical) can be regarded as a linear operators acting in $(\CC^2)^{\otimes N}$ where the factors in the tensor
product correspond to states on boundary edges.

The partition function for the {\it torus} of size $N\times M$ is a number and it can be regarded as a the trace of the partition
for the vertical cylinder of size $N\times M$ over the states
on its horizontal sides, i.e. over $(\CC^2)^{\otimes N}$.
It can also be regarded as trace of the partition
for the horizontal cylinder of size $M\times N$ over the states
on its vertical sides, i.e. over $(\CC^2)^{\otimes M}$.
The result is an identity which we will discuss later.

\begin{figure}[t]
\begin{center}
\includegraphics[height=5cm,width=5cm]{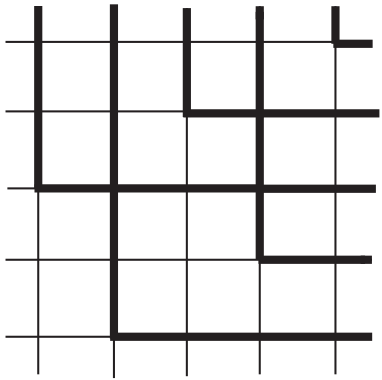}
\end{center}
\caption{A possible configuration of paths on a $5\times 5$ square grid for
the DW boundary conditions.}
\label{dwbc_conf}
\end{figure}

\subsection{Transfer-matrices}
Let is write the matrix of Boltzmann weights for a vertex
as the $4\times 4$ matrix acting in the tensor product of
spaces of states on adjacent edges.

Let $e_1$ be the vector corresponding to the arrow pointing
up on a vertical edge, and left on a horizontal edge. Let
$e_2$ be the vector corresponding to the arrow pointing down
on a vertical edge and right on a horizontal edge.

The matrix of Boltzmann weights with zero electric fields
acts as
\begin{eqnarray}
R e_1\otimes e_1&=&a e_1\otimes e_1\\
R e_1\otimes e_2&=&b e_1\otimes e_2+c e_2\otimes e_1\\
R e_2\otimes e_1&=&b e_2\otimes e_1+c e_1\otimes e_2\\
R e_2\otimes e_2&=&a e_2\otimes e_2
\end{eqnarray}
In the tensor product basis $e_1\otimes e_1, e_1\otimes e_2, e_2\otimes e_1, e_2\otimes e_2$ it is the $4\times 4$ matrix
\begin{equation}\label{R-bar}
\overline{R}=\left(\begin{array}{cccc} a & 0 & 0 & 0 \\
0 & b & c & 0 \\
0 & c & b & 0 \\
0 & 0 & 0 & a
\end{array}
\right)
\end{equation}
The $6$-vertex rules imply that the operator $R$ commutes with
the operator representing the total number of thick vertical edges, i.e.
\[
[D^H\otimes D^H, \overline{R}]=0
\]
where
\begin{equation*}
D^H = \begin{pmatrix}
e^{H/2} & 0  \\
0 & e^{-H/2}
\end{pmatrix}.
\end{equation*}

The row-to-row transfer-matrices with open boundary conditions
also known as the (quantum) {\it monodromy matrix} is defined as
\[
T_a=D_a^{H(a,1)}R_{a1}D_a^{H(a,1)}\dots D_a^{H(a,1)}R_{aN}
\]
It acts in the tensor product $\CC^2_a\otimes \CC^2_1\otimes \dots \otimes \CC^2_N$ of spaces corresponding to
horizontal edges   and vertical edges . Each matrix $R_{ai}$ is of the form (\ref{R-weights}), it acts trivially (as the identity matrix) in all factors of the tensor
product except $a$ and $i$. In the inhomogeneous case matrix elements of $\overline{R}_{ai}$ depend
on $i$.

A matrix element of $T$ is the partition function of the 6-vertex
model on a single row with fixed boundary conditions.

Define operators
\[
D^{(a)}=D^{V(a,1)}\otimes \dots \otimes D^{V(a,N)}
\]
\[
D_a^H=1\otimes \dots \otimes D^H\otimes \dots \otimes 1
\]

The row-to-row transfer-matrix corresponding to the
cylinder with a single horizontal row $a$ with
with electric field $H(a,i)$
applied to the $i$-th horizontal edge of the $a$-th
horizontal line is the following trace
\[
t_a=tr_a(T_a)=tr(D_a^{H(a,1)}R_{a1}\dots D_a^{H(a,N)}R_{aN})
\]
It is an operator acting in $(\CC^2)^{\otimes N}$. Its matrix element is the partition function of the $6$-vertex model
on the cylinder with a single row with fixed boundary conditions on vertical edges.

The partition function of an inhomogeneous $6$-vertex on a cylinder of height $M$ with fixed boundary condition on outgoing vertical edges is a matrix element of the linear operator
\[
Z^{(C)}=D^{(M)}t_M\dots D^{(1)}t_1
\]
where
\[
D^{(a)}=D^{V(a,1)}\otimes \dots \otimes D^{V(a,N)}
\]

The partition function for the torus with $N$ columns and
$M$ rows is the trace of the partition function for
the cylinder
\[
Z^{(T)}_{N,M}=tr_{(\CC^2)^{\otimes N}}(D^{(M)}t_M\dots D^{(1)}t_1)
\]
Using the $6$-vertex rules this trace
can be transformed to
\[
Z^{(T)}_{N,M}=tr_{(\CC^2)^{\otimes N}}(t_M^{(H_M)}\dots t_1^{(H_1)}(D_1^{V_1}\otimes \dots \otimes D_N^{V_N}))
\]
where $H_a=\sum_{i=1}^N H_{ai}$ and $V_i=\sum_{a=1}^M V_{ai}$
and
\[
t^{(H)}=tr_a(R_{a1}\dots R_{aN}D_a^{H})
\]

The partition function for a generic domain does not have
a natural expression in terms of a transfer-matrix. However,
it is possible in few exceptional cases, such as domain wall
boundary conditions on a square domain, see for example \cite{DW}.
and references therein.

\subsection{The commutativity of transfer-matrices and positivity of weights}

\subsubsection{Commutativity of transfer-matrices}
Baxter discovered that matrices of the form
(\ref{R-weights}) acting in the tensor product of two two-dimensional spaces satisfy the equation
\begin{equation}\label{Bax}
R_{12}R_{13}'R_{23}''=R_{23}''R_{13}'R_{12}
\end{equation}
if
$$
\frac{a^2+b^2-c^2}{2ab}=\frac{{a'}^2+{b'}^2-{c'}^2}{2a'b'}
=\frac{{a''}^2+{b''}^2-{c''}^2}{2a''b''}
$$
This parameter is denoted by $\Delta$:
$$
\Delta=\frac{a^2+b^2-c^2}{2ab}
$$

If each factor in monodromy matrices $T'_a=R'_{a1}\dots R'_{aN}$ and $T''_b=R''_{a1}\dots R''_{aN}$ have the same value of
$\Delta$, the equation (\ref{Bax}) implies that they these monodromy matrices
satisfy the relation
\[
R_{ab}T'_aT''_b= T''_bT'_a R_{ab}
\]
in $V_a\otimes V_b\otimes V_1\otimes \dots \otimes V_N$.

If $R$ is invertible, this relation implies that row-to-row transfer-matrices with periodic boundary conditions commute:
\[
t'=tr_a(T'_a D_a^H), \ \  t''=tr_b(T''_bD_b^H), \ \ [t',t'']=0
\]

It is easy to recognize that $t$ is exactly the generating function for commuting family of local Hamiltonians for spin chains constructed earlier. Thus, the problem of computing the
partition function for periodic and cylindrical boundary conditions for the $6$-vertex model is closely related to finding the spectrum of Hamiltonians for integrable spin chains.

\subsubsection{The parametrization} The set of positive triples of real numbers $a:b:c$ (up to a common multiplier) with fixed values of $\Delta$ has the following parametrization \cite{Ba}.

\begin{enumerate}

\item When $\Delta >1$ there are two cases.

If $a>b+c$, the Boltzmann weights $a,b,$ and $c$ can be parameterized as
$$
a=r\sinh(\lambda+\eta), \quad b=r\sinh(\lambda),\quad c=r\sinh(\eta)
$$
with $\lambda,\eta>0$.

If $b>a+c$, the Boltzmann weights can be parameterized as
$$
a=r\sinh(\lambda-\eta), \quad b=r\sinh(\lambda),\quad c=r\sinh(\eta)
$$
with $0<\eta<\lambda$.
For both of these parameterizations of weights $\Delta=\cosh(\eta)$.

\item When $-1<\Delta\leq 0$
$$
a=r\sin(\lambda-\gamma), \quad b=r\sin(\lambda), \quad c=r\sin(\gamma),
$$
where $0<\gamma<\pi/2$, $\gamma<\lambda<\pi/2$, and $\Delta=-\cos\gamma$.

\item When $0\leq \Delta<1$
$$
a=r\sin(\gamma-\lambda), \quad b=r\sin(\lambda), \quad c=r\sin(\gamma),
$$
where $0<\gamma<\pi/2$, $0<\lambda<\gamma$, and $\Delta=\cos\gamma$.

\item When $\Delta <-1$ the parametrization is
$$
a=r\sinh(\eta-\lambda),\quad  b=r\sinh(\lambda),\quad c=r\sinh(\eta),
$$
where $0<\lambda<\eta$ and $\Delta=-\cosh\eta$.

\end{enumerate}

We will write $a=a(u), b=b(u), c=c(u)$ assuming this parameterizations.

\subsubsection{Topological nature of the partition function
of the $6$-vertex model}

Fix a domain and a collection of simple non-selfintersecting oriented curves with simple (transversive) intersections. The result is a $4$-valent graph embedded into the domain.
Fix $6$-vertex states (arrows) at the boundary edges of this graph. Such graph connects boundary points, and defines a perfect matching between boundary points.\begin{figure}[t]
\begin{center}
\includegraphics[height=3cm,width=4cm]{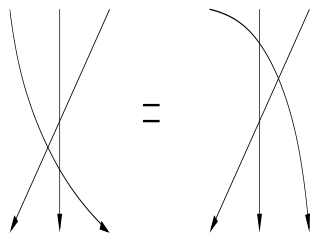}
\end{center}
\caption{The Yang-Baxter equation.}
\label{Red-1}
\end{figure}
\begin{figure}[t]
\begin{center}
\includegraphics[height=3cm,width=4cm]{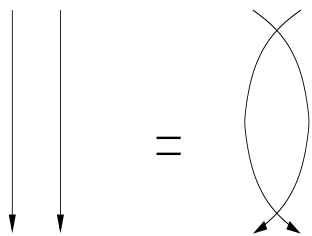}
\end{center}
\caption{The `unitarirty'.}
\label{Red-2}
\end{figure}

The $6$-vertex rules define the partition function on such graph.
The Yang-Baxter relation implies the invariance with respect to
the two moves shown of Fig. \ref{Red-1}, Fig. \ref{Red-2}.
Because of this, the partition function of the $6$-vertex model
depends only on the connection pattern between boundary points,
i.e. it depends only on the perfect matching on boundary points
induced by the graph.

The `unitarity' relation involves the inverse to $R$, and
therefore does not preserve the positivity of weights. But if it is used "even" number of times it gives the equivalence of partition functions with positive weights. For example it can be used to `permute' the rows in case of cylindrical boundary conditions.

\subsubsection{Inhomogeneous models with commuting transfer-matrices} From now on we will focus on the $6$-vertex
model in constant electric fields. If Boltzmann weights of the $6$-vertex model have special inhomogeneity $a_{ij}=a(u_i-w_j),
b_{ij}=b(u_i-w_j), c_{ij}=c(u_i-w_j)$ the partition function
of the $6$-vertex model on the torus is
\begin{equation}\label{6-v-ihom}
Z^{(T)}_{N,M}(\{u\},\{w\}|H,V)=Tr_{\cH_N}( t(u_1)\dots t(u_M)(D^{NV}\otimes \dots\otimes  D^{NV}))
\end{equation}
where $\cH_N=(\CC^2)^{\otimes N}$ and $t(u)=tr_a(T_a(u))$
\begin{equation}\label{6-v-ihom-T}
T_a(u)=R_{a1}(u-w_1)\dots R_{aN}(u-w_N)D^{NH}
\end{equation}

Because the $R$-matrix satisfies the Yang-Baxter equation
we have
\[
R_{ab}(u-v)T_a(u)T_b(v)=T_b(v)T_a(u)R_{ab}(u-v)
\]
As a corollary, traces of these matrices commute:
\[
[t(u), t(v)]=0
\]
and as we have seen the previous sections their spectrum can be
described explicitly by the Bethe ansatz.

Notice that positivity of weights restricts possible value of
inhomogeneities. For example when $\Delta>1$ and $a>b+c$
we should have $-u<w_i<u$.

\subsubsection{} For a diagonal matrix $d$ such that $[d\otimes 1+
1\otimes d, R(u)]=0$ define
\[
R^d(u)=(e^{ud}\otimes 1)R(u)(e^{-ud}\otimes 1)
\]
It is clear that
\[
tr(R^d_{a1}(u-w_1)\dots R^d_{aN}(u-v_N)D^{NH})=Utr(R_{a1}(u-w_1)\dots R_{aN}(u-v_N)D^{NH})U^{-1}
\]
where $U=\exp(-\sum_{i=1}^Nw_id_i)$.

In particular the partition function for a torus
with weights given by $R(u)$ and with weights given by $R^d(u)$ are the same.

\subsubsection{}
Let $u_1,\dots, u_N$ be parameters of inhomogeneities along horizontal directions, and $w_1,\dots, w_M$ be inhomogeneities
along vertical directions. The partition function (\ref{6-v-inhom}) have the following properties:
\begin{itemize}
\item $Z^{(T)}_{N,M}(\{u\},\{w\}|H,V)$ is a symmetric function of $u$ and $w$.

\item It has a form $\prod_{i=}^Ne^{-Mu_i}\prod_{i=1}^Me^{-Nw_i}P(\{e^{2u}\}, \{e^{2w}\})$ where $P$ is a polynomial of degree $M$ in each $e^{u_i}$ and of degree $N$ is each $e^{w_j}$.

\item It is a function of $u_i-w_j$.

\item It satisfies the identity
\begin{equation}\label{modul}
Z^{(T)}_{N,M}(\{u\},\{w\}|H,V)=Z^{(T)}_{M,N}(\{w\},\{h-u\}|V,H)
\end{equation}
where $h=\eta$ or $\gamma$ depending on the regime. This identity correspond to the "rotation" of the torus by $90$ degrees and is known as "crossing-symmetry" or "modular identity".

\end{itemize}
\section{The $6$-vertex model on a torus in the thermodynamic limit }

\subsection{The thermodynamic limit of the 6-vertex model for
the periodic boundary conditions}\label{free-en-tdl}
By the thermodynamical limit here we will mean here the large volume limit, when $N,M\to \infty$.

The free energy per site in this limit is
\[
f=-\lim_{N,M\to\infty}\frac{\log(Z_{N,M})}{NM}
\]
where $Z_{N,M}$ is the partition function with the periodic
boundary conditions on the rectangular grid $L_{N,M}$. It is a
function of the Boltzmann weights and magnetic fields.

For generic $H$ and $V$ the $6$-vertex model in the thermodynamic
limit has the unique translational invariant Gibbs measure with the
slope $(h,v)$:
\begin{equation}\label{slope}
h=-\frac{1}{2}\frac{\pa f}{\pa H}+\frac{1}{2}, \quad
v=-\frac{1}{2}\frac{\pa f}{\pa V}+\frac{1}{2}.
\end{equation}

This Gibbs measure defines local correlation functions in the thermodynamical limit
The parameter
$$
\Delta=\frac{a^2+b^2-c^2}{2ab}.
$$
defines many characteristics of the $6$-vertex model in the
thermodynamic limit.

\subsection{The large $N$ limit of the eigenvalues of the transfer-matrix}
The row-to-row transfer-matrix for the homogeneous
$6$-vertex model on a lattice with periodic boundary condition
with rows of length $N$ is
\[
t(u)=tr_a(R_{a1}(u)\dots R_{aN}(u)e^{H\sigma^z})\exp(V\sum_{a=1}^N\sigma_a^z)
\]
According to (\ref{BE-eigenv}) and (\ref{BE}), the eigenvalues of
this linear operator are:
\begin{equation}\label{6v-eigenv}
\Lambda_{\{v_i\}}=a^Ne^{NH}\prod_{i=1}^n{v_iz^{-1}q-v_i^{-1}zq^{-1}\over v_iz^{-1}-v_i^{-1}z}+b^Ne^{-NH}\prod_{i=1}^n{v_i^{-1}zq-v_iz^{-1}q^{-1}\over v_i^{-1}z-v_iz^{-1}}
\end{equation}
where $0\leq n \leq N$ and $z,q$ are parameterizing $a,b,c$ as $a=r(zq-zq^{-1}, b=r(z-z^{-1}),c=r(q-q^{-1})$. The numbers $v_i$
are solutions to the Bethe equations:
\begin{equation}\label{6v-BE}
\left({v_iq-v_i^{-1} q^{-1}\over v_i-v_i^{-1}}\right)^N=-e^{-2H}\prod_{j=1}^n {v_iv_j^{-1}q-v_i^{-1}v_jq^{-1}\over v_iv_j^{-1}q^{-1}-v_i^{-1}v_jq}
\end{equation}

As for the corresponding spin chains it is expected that
the numbers $v_i$ corresponding to the
largest eigenvalue concentrate, when $N\to \infty$,  on a contour in a complex plane with the finite density. The Bethe equations
provide a linear integral equation for this density.
This conjecture is supported by the numerical evidence and it is proven in some special cases, for example when $\Delta=0$.

The partition function for the homogeneous $6$-vertex model
on an $N\times M$ lattice with periodic boundary conditions
is
\begin{equation}\label{part-fcnc}
Z_{N,M}=\sum_\alpha {\Lambda_\alpha^{(N)}}^M
\end{equation}
where $\alpha$ parameterize eigenvalues of $t(u)$.

Let $\omega_N$ be the eigenvector of $t(u)$ corresponding to the
maximal eigenvalue. The sequence of vectors $\{\Omega_N\}$ as $N\to \infty$ defines the Hilbert space of pure states for the infinite system. Let $\Lambda_0$ be the largest eigenvalue of $t(u)$.
According to the main conjecture about that the largest eigenvalue correspond to numbers $v_i$ filling a contour in a complex
plane as $N$, the largest eigenvalue has the following asymptotic:
\begin{equation}\label{L-0}
\Lambda_0^{(N)}=\exp(-Nf(H,V)+O(1))
\end{equation}
The function $f(H,V)$ as we will see below it
is the free energy of the system. It is computed in
the next section.

The transfer-matrix in this limit has the asymptotic
\[
t(u)=\exp(-Nf(H,V))\widehat{t}(u)
\]
where the operator $\widehat{t}(u)$ acts in the space $H_\infty$
and its eigenvalues are determined by positions of "particles" and
"holes",
similarly to the structure of excitations in the large $N$ limit
in spin chains.

\subsection{Modularity}
\subsubsection{} It is easy to compute now
the asymptotic of the partition function
in the thermodynamical limit. As $M\to \infty$
the leading term in the formula (\ref{part-fcnc})
is given by the largest eigenvalue:
\[
Z_{N,M}={\Lambda_0^{(N)}}^M(1+O(e^{-\alpha M}))
\]
for some positive $\alpha$.

Taking the limit $N\to \infty$ and taking into account the
asymptotic of the largest eigenvalue (\ref{L-0}) we identify
the function $f(H,V)$ in (\ref{L-0}) with the free energy:
\[
Z_{N,M}=e^{NMf(H,V)(1+o(1))}
\]

\subsubsection{} Notice now that we could have changed the role of $N$ and $M$ by first taking the limit $N\to \infty$ and then $M\to \infty$.

In this case we would have to compute first the asymptotic
of
\[
Z_{N,M}=tr(t(u)^M)
\]
as $N\to\infty$ and then take the limit $M\to \infty$.

The large $N$ limit of the trace can be computed by
using the finite temperature technique developed by
Yang and Yang in \cite{YY}. It was done by de Vega and
Destry \cite{DVD}. The leading term of the asymptotic
can be expressed in terms
of the solution to a non-linear integral equation.

This gives an alternative description for the largest eigenvalue.
Similar description exists for all eigenvalues. In other integrable
quantum field theories it was done by Al. Zamolodchikov \cite{Z}.

\section{The $6$-vertex model at the free fermionic point}\label{ff}

When $\Delta=0$ the partition function of the $6$-vertex
model can be expressed in terms of the dimer model on
a decorated square lattice. Because the dimer model
can be regarded as a theory of free fermions, the $6$-vertex model
is said to be free fermionic when $\Delta=0$.

At this point the raw-to-raw transfer-matrix on $N$-sites for a torus can be written in terms of the Clifforrd algebra of $\CC^N$.
The Jordan-Wiegner transform maps local spin operators to the
elements of the Clifford algebra.

In Bethe equations (\ref{BE}) the variables $v_i$ disappear in the
r.h.d which becomes simply $(-1)^{n-1}e^{-2NH}$. After change of
variables these equations can be interpreted as the periodic boundary condition for a fermionic wave function.

\subsection{Homogeneous case}
\subsubsection{} At the free fermionic point $\Delta=0$, i.e. $\gamma=\frac{\pi}{2}$ and the weights of the $6$-vertex model are parameterized
as:
\[
a=\sin(\frac{\pi}{2}-u)=\cos(u), \ \ b=\sin(u), \ \
c=1
\]
Without loosing generality we may assume that $0<u<{\pi \over 4}$.

The eigenvalues of the row-to-row transfer-matrix are given by the
Bethe ansatz formulae:
\begin{equation} \label{ff-hom-eigen}
\Lambda(u)=((\cos(u)^N(-1)^ne^{NH}+(\sin(u))^Ne^{-NH})e^{NV}\prod_{i=1}^n (\cot(u-v_i)e^{-2V})
\end{equation}
Here $0\leq n\leq N$, and $v_i$ are distinct solutions to the Bethe equations:
\[
\cot(v)^N=(-1)^{n-1}e^{-2NH}
\]
or:
\[
\cot(v)=\w e^{-2H}
\]
where $\w^N=(-1)^{n-1}$.

Using the identity:
\[
\cot(u-v)={\cot u \cot v +1\over \cot v-\cot u}=\cot u-{1-\cot^2u\over \omega e^{-2H}-\cot u}
\]
we can write the eigenvalues as:
\begin{equation}\label{eigenval}
\Lambda(u)=((\cos(u)^N(-1)^ne^{NH}+(\sin(u))^Ne^{-NH})e^{(N-2n)V})
\prod_{i=1}^n (\cot u-{1-\cot^2u\over \omega_i e^{-2H}-\cot u})
\end{equation}

\subsubsection{} To find the maximal eigenvalue in the limit $N\to \infty$ we should analyze factors in the
formula (\ref{ff-hom-eigen}).

{\bf 1}. If
\[
\max_{|\omega|=1}\cot(u-v)<e^{2V}
\]
where $\cot(v)=\omega e^{-2H}$, all factor are less then one and
the maximal eigenvalue
\[
\Lambda_{ord}(u)=(\cos(u)^N e^{NH}+(\sin(u))^Ne^{-NH})e^{NV}
\]
is achieved when $n=0$.

As $N\to \infty$ the first term dominates when $
\cot(u)<e^{2H}$. In this case the Gibbs state describing
the 6-vertex model in the thermodynamical limit is
the ordered state $A_1$ from Fig. \ref{ferroelectric_phase}
When $\cot(u)>e^{2H}$, the second term dominates. In
this case the Gibbs state the ordered state $B_2$ shown
on Fig. \ref{ferroelectric_phase}.

{\bf 2}. If $\min_{|\omega|=1}\cot(u-v)>e^{2V}$ all factors $\cot(u-v_j)e^{2V}$ are greater then one by absolute value.
In this case the maximal eigenvalue is achieved when $n=N$.
The corresponding ground state is ordered  and in $A_2$
from Fig. \ref{ferroelectric_phase} when $\cot u >e^{-2H}$
and is $B_1$ when $\cot( u) <e^{-2H}$

{\bf 3}. If
\[
\max_{|\omega|=1}\cot(u-v)<e^{2V}
\]
then there exists $\omega_0=e^{iK}$ such that
\[
|\cot(u-b)|=e^{2V}
\]
where $\cot(b)=e^{\pm iK-2H}$. It is easy to find $K$:
\[
\cos(K)=\frac{(Ue^{2H}+U^{-1}e^{-2H})e^{4V}-(Ue^{-2H}+U^{-1}e^{2H})}{2(e^{4V}+1)}
\]

In this case
\[
|\cot(u-v)|>e^{2V}
\]
when $\cot(v)=e^{i\alpha}e^{-2H}$ with $-K<\alpha< K$
and
\[
|\cot(u-v)|<e^{2V}
\]
when $-\pi\leq \alpha< -K$ or $K<\alpha\leq \pi$.

The maximal eigenvalue
is this case corresponds to maximal $n$ such that ${\pi (n-1)\over N}<K$
and $\omega_j=\exp(i{\pi (n+1-2j)\over N})$, $j=1,\dots, n$ and is
given by (\ref{ff-hom-eigen}).

As $N\to \infty$ the asymptotic of the largest eigenvalue is
given by the integral:
\[
\log(\Lambda_{disord}(u))=N(\log(\cos(u)e^{H})+\frac{1}{2\pi i}\int_{-K}^K\log\left(\frac{\cot(u)\omega e^{-2H}+1}{\omega e^{-2H}-\cot(u)}\right)\frac{d\omega}{\omega})+O(1)
\]
The 6-vertex in this regime is in the disordered phase.

Disordered and ordered phases are separated by the curve
\[
\max_{|\omega|=1}|\cot(u-v)|=e^{-2V}, \ \ \min_{|\omega|=1|}|\cot(u-v)|=e^{-2V} ,
\]
or, more explicitly
\[
\left|{Ue^{-2H}\mp1\over U\pm e^{-2H}}\right|=e^{-2V}
\]
This curve is shown on Fig. \ref{ff-curve}

\begin{figure}[t]
\begin{center}
\includegraphics[width=10cm, height=10cm]{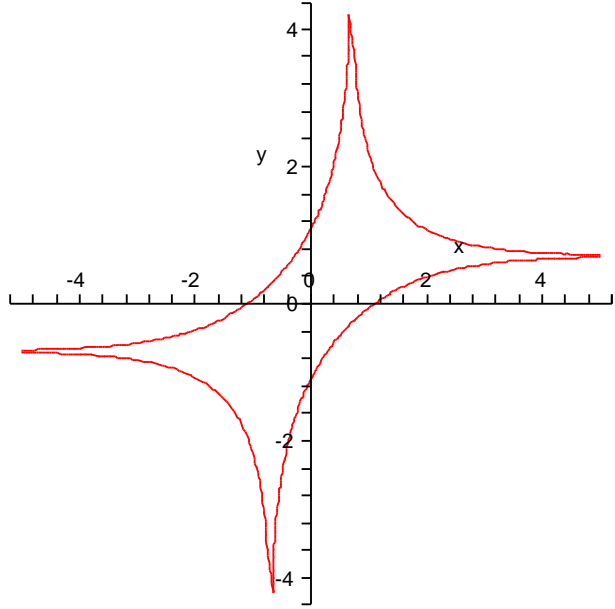}
\end{center}
\caption{The boundary ordered and disordered regions in the $(H,V)$-plane for the homogeneous $6$-vertex model with $\Delta=0$ and $U=\cot(u)=2$} \label{ff-curve}
\end{figure}

\subsection{Horizontally inhomogeneous case}

The partition function for the 6-vertex model at the free fermionic point with inhomogeneous rows is
\[
Z_{N,M}(u,a)=\sum_{n=0}^N \sum_{\omega_1,\dots, \omega_n}\Lambda(u+a|\omega)^M\Lambda(u-a|\omega)^M
\]
where $\omega_i\neq \omega_j$ are solutions to
\[
\omega^N=(-1)^{n-1}
\]
and $\Lambda(u|\omega)$ is given by (\ref{eigenval}).

The boundary between ordered and disordered phases is
given by  equations
\[
\max_{|\omega|=1}|\cot(u+a-v)\cot(u-a-v)|=e^{4V}, \ \ \min_{|\omega|=1}|\cot(u+a-v)\cot(u-a-v)|=e^{4V}
\]
or,
\begin{equation}\label{b-curve}
\left|{U_+e^{-2H}\pm 1 \over U_+\mp e^{-2H}}\right|\left|{U_-e^{-2H}\pm 1 \over
 U_-\mp e^{-2H}}\right|=e^{4V}
\end{equation}
This curve is shown on Fig. \ref{ff-inhom-curve}.

The ordered phases $A1B_2$ and $A_2B_1$ are shown on Fig.
\ref{ff-AB-phases}.

\begin{figure}[t]
\begin{center}
\includegraphics[width=10cm, height=10cm]{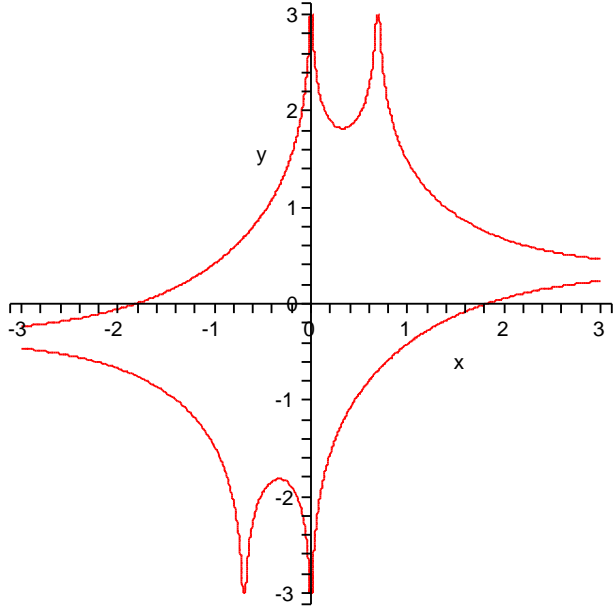}
\end{center}
\caption{The boundary between ordered and disordered regions in the $(H,V)$-plane for the inhomogeneous $6$-vertex model with $\Delta=0$ and $U=\cot(u)=3, T=\cot(a)=1$} \label{ff-inhom-curve}
\end{figure}

\begin{figure}[t]
\begin{center}
\includegraphics[width=10cm, height=4cm]{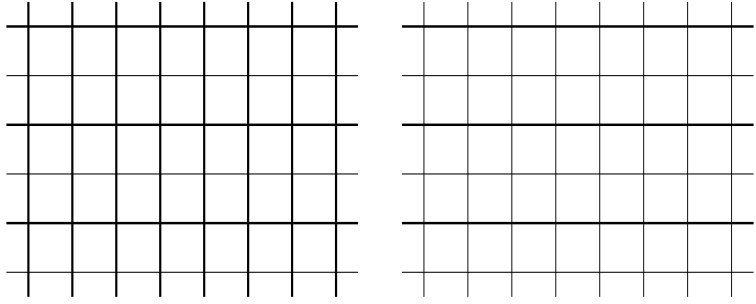}
\end{center}
\caption{Horizontally inhomogeneous ordered phases $A_1B_2$ and $A_2B_1$ respectively} \label{ff-AB-phases}
\end{figure}

The free energy per site in the disordered region is given by
\begin{multline}\label{inhom-free-en}
f(H,V)=\log(cos(u+a)\cos(u-a)+\\
\frac{1}{2\pi i}\int_{-K_1}^{K_1}\log\left(\frac{\cot(u+a)\omega e^{-2H}+1}{\omega e^{-2H}-\cot(u+a)}\frac{\cot(u-a)\omega e^{-2H}+1}{\omega e^{-2H}-\cot(u-a)}\right)\frac{d\omega}{\omega})+\\
\frac{1}{2\pi i}\int_{-K_2}^{K_2}\log\left(\frac{\cot(u+a)\omega e^{-2H}+1}{\omega e^{-2H}-\cot(u+a)}\frac{\cot(u-a)\omega e^{-2H}+1}{\omega e^{-2H}-\cot(u-a)}\right)\frac{d\omega}{\omega})
\end{multline}
where $K_{1,2}$ are defined by equations
\[
|\cot(u+a-b)\cot(u-a-b)|=e^{2V}
\]
with $\cot(b)=e^{\pm i K}e^{-2H}$.

\subsection{Vertically non-homogeneous case}
\subsubsection{}Assume the weights are homogeneous in the vertical direction
and alternate in the horizontal direction with parameters
alternating as $\dots, u+a,u-a,u+a,u-a,\dots$. Positivity of weights in the region $0\leq u\leq \frac{\pi}{4}$
requires $0\leq a\leq u$.

The eigenvalues of the transfer-matrix are given by the Bethe ansatz:
\begin{multline}\label{eigenval-1}
\Lambda(u)=((\cos(u+a)\cos(u-a))^N(-1)^ne^{NH}+(\sin(u+a)\sin(u-a))^Ne^{-NH})\\
e^{(N-2n)V}\prod_{i=1}^n \cot(u-v_i)
\end{multline}
where $v_i$ are distinct solutions to
\[
(\cot(v-a)\cot(v+a))^N=(-1)^{n-1}e^{-2NH}
\]
The left side is the $N$-th power of
\[
\frac{\cot(v)\cot(a)+1}{\cot(a)-\cot(v)}
\frac{-\cot(v)\cot(a)+1}{-\cot(a)-\cot(v)}=
\frac{\cot(v)^2\cot(a)^2-1}{\cot(a)^2-\cot(v)^2}
\]
From here it is easy to find the parametrization of solutions
by roots of unity:
\begin{equation}\label{wv-inhom}
\cot(v)^2=\frac{\w \cot(a)^2e^{-2H}+1}{\cot(a)^2+\w e^{-2H}}
\end{equation}
where
\[
\w^N=(-1)^{n-1}
\]

\subsubsection{Largest eigenvalue}
The analysis of the largest eigenvalue of the transfer-matrix in the limit $N\to \infty$ is similar to the previous cases. When
\begin{equation}\label{b-curve-inh}
\max_{|\omega|=1}|\cot(u-v)|\geq e^{2V}
\end{equation}{
where $v$ and $\omega$ are related as in (\ref{wv-inhom})
the absolute value of all factors $\cot(u-v)|e^{-2V}$ is
greater then one by the absolute value and the largest eigenvalue
correspond to $n=0$. The rest of the analysis is similar.

Let us notations:
\[
T=\cot(a), \ \ x=e^{-H}, \ \ U=\cot(u), \ \ s=\cot(v)
\]
Positivity of weights imply
\[
T\geq U>0
\]
\begin{proposition}
The curve separating ordered phases from the disordered phase
is
\[
X=\left|{YU_+-1 \over Y+U_+}\right|\left|{YU_--1\over Y+U_-}\right|
\]
\end{proposition}

\begin{proof}
In the notations from above:
\[
|\cot(u-v)|^2=|\frac{\cot(u)\cot(v)+1}{\cot(v)-\cot(u)}|^2=
\frac{Us+1}{s-U}\frac{U\bar{s}+1}{\bar{s}-U}=
\frac{U^2|s|^2+2U Re(s)+1}{U^2-2U Re(s)+|s|^2}
\]
The equation of the boundary of the disordered region is
\[
\max_{|\omega|=1} |\cot(u-v)|=e^{2V}
\]

The equation defining $v$ in terms of roots of unity $\omega$ can be solved explicitly for $s=\cot v$:
\[
s^2=\frac{T^2x^2\w+1}{T^2+\w x^2}
\]
Denote $\w=e^{i\alpha}$, then
\[
s^2=\frac{(T^2x^2\cos(\al)+1)+iT^2x^2\sin(\al)}
{(T^2+x^2\cos(\al))+ix^2\sin(\al)}
\]
\[
Re(s^2)=\frac{(T^2x^2\cos(\al)+1)(T^2+x^2\cos(\al))+T^2x^2\sin(\al)x^2\sin(\al)}{(T^2+x^2\cos(\al))^2+
x^4\sin(\al)^2}
\]
\[
|s|^4=\frac{T^4x^4+2T^2x^2\cos(\al)+1}{T^4+2T^2x^2\cos(\al)+x^4}
\]
Now, a simple algebra:
\[
Re(s)=\sqrt{\frac{Re(s^2)+|s|^2}{2}}
\]
As we vary $\omega$ along the unit circle, the function $\cot(u-v)$
has maximum at $\w=1$.

Taking this into account the equation for the boundary curve
becomes
\[
\left|{U\sqrt{1+x^2T^2}\pm \sqrt{T^2+x^2}\over U\sqrt{T^2+x^2}\pm
 \sqrt{1+x^2T^2}}\right|=e^{2V}
\]
Solving this equation for $X=x^2=e^{-2H}$ in terms of $Y=e^{2V}$ we
obtain
\[
X=\left|{YU_+-1 \over Y+U_+}\right|\left|{YU_--1\over Y+U_-}\right|
\]
which is, after changing coordinates to $X^{-1}, Y^{-1}$ is the same curve as (\ref{b-curve}). This is one of the implications of the modular symmetry.
\end{proof}

By modularity, i.e. by "rotating" the lattice with periodic boundary conditions by $90$ degrees, all characteristics
of the model with vertical inhomogeneities can be identifies with
the corresponding characteristics of the model with horizontal inhomogeneities.

\section{The free energy of the $6$-vertex model}

The computation of the free energy
for the $6$-vertex model by taking the large $N,M$ limit of the partition function on a torus was outlines in section \ref{free-en-tdl}. In this section we will describe
the free energy as a function of electric fields $(H,V)$ and
its basic properties.

\subsection{The phase diagram for $\Delta>1$}
\subsubsection{} The weights $a$, $b$, and $c$
in this region satisfy one of the two inequalities, either $a>b+c$
or $b>a+c$.

If $a>b+c$, the Boltzmann weights $a$, $b$, and $c$ can be
parameterized as
\begin{equation}\label{dg1par-1}
a=r\sinh(\lambda+\eta), b=r\sinh(\lambda), c=r\sinh(\eta)
\end{equation}
with $\lambda,\eta>0$.

If $a+c<b$, the Boltzmann weights can be parameterized as
\begin{equation}\label{dg1par-2}
a=r\sinh(\lambda-\eta), b=r\sinh(\lambda), c=r\sinh(\eta)
\end{equation}
with $0<\eta<\lambda$.

For both of these parametrization of weights $\Delta=\cosh(\eta)$.

The phase diagram of the model for $a>b+c$ (and, therefore, $a>b$)
is shown on Fig. \ref{ferro_diagram_1} and for $b>a+c$ (and,
therefore, $a<b$) on Fig. \ref{ferro_diagram_2}.

\subsubsection{} When magnetic fields $(H,V)$ are in one of the regions $A_i, B_i$
of the phase diagram, the system in the thermodynamic limit
is described by
the translationally invariant Gibbs measure supported on the
corresponding frozen (ordered) configurations. There are four frozen configurations $A_1$, $A_2$, $B_1$, and $B_2$, shown on Fig. \ref{ferroelectric_phase}. For a finite but large grid the
probability of any other state is of the order at most $\exp(-\alpha N)$ for some positive $\alpha$.

Local correlation functions in a frozen state are products of
expectation values of characteristic functions of edges.
$$
\lim_{N\to\infty}\langle \sigma_{e_1}\dots\sigma_{e_n}\rangle_N=
\sigma_{e_1}(S)\dots\sigma_{e_n}(S)
$$
where $S$ is the one of the ferromagnetic states $A_i,B_i$.

\begin{figure}[b]
\begin{center}
\includegraphics{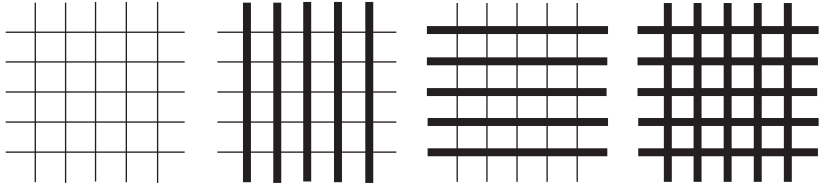}
\end{center}
\caption{Four frozen configurations of the ferromagnetic phase}
\label{ferroelectric_phase}
\end{figure}

\subsubsection{} The boundary between ordered phases in the $(H,V)$-plane and disordered phases, as in the free fermionic case, is determined by the next to the largest eigenvalue of the row-to-row transfer matrix.

Without going into the details of the computations as we did in the free fermionic case we will just give present answers.

$\bullet$ $a>b+c$, see Fig. \ref{ferro_diagram_1},
\begin{eqnarray}
&& \mbox{$A_1$-region:}\qquad V+H\geq 0, \qquad
\cosh(2H)\leq\Delta, \nonumber
\\
&& \qquad (e^{2H}-b/a)(e^{2V}-b/a)\geq(c/a)^2,\qquad e^{2H}>b/a,
\qquad \cosh(2H)>\Delta, \nonumber
\\
&& \mbox{$A_2$-region:}\qquad V+H\leq 0, \qquad
\cosh(2H)\leq\Delta, \nonumber
\\
&& \qquad (e^{-2H}-b/a)(e^{-2V}-b/a)\geq(c/a)^2,\qquad
e^{-2H}>b/a, \qquad \cosh(2H)>\Delta, \nonumber
\\
&& \mbox{$B_1$-region:}\qquad
(e^{2H}-a/b)(e^{-2V}-a/b)\geq(c/b)^2, \qquad e^{2H}>a/b, \nonumber
\\
&& \mbox{$B_2$-region:}\qquad
(e^{-2H}-a/b)(e^{2V}-a/b)\geq(c/b)^2, \qquad e^{-2H}>a/b.
\nonumber
\end{eqnarray}

\begin{figure}[t]
\begin{center}
\includegraphics[width=8cm, height=8cm]{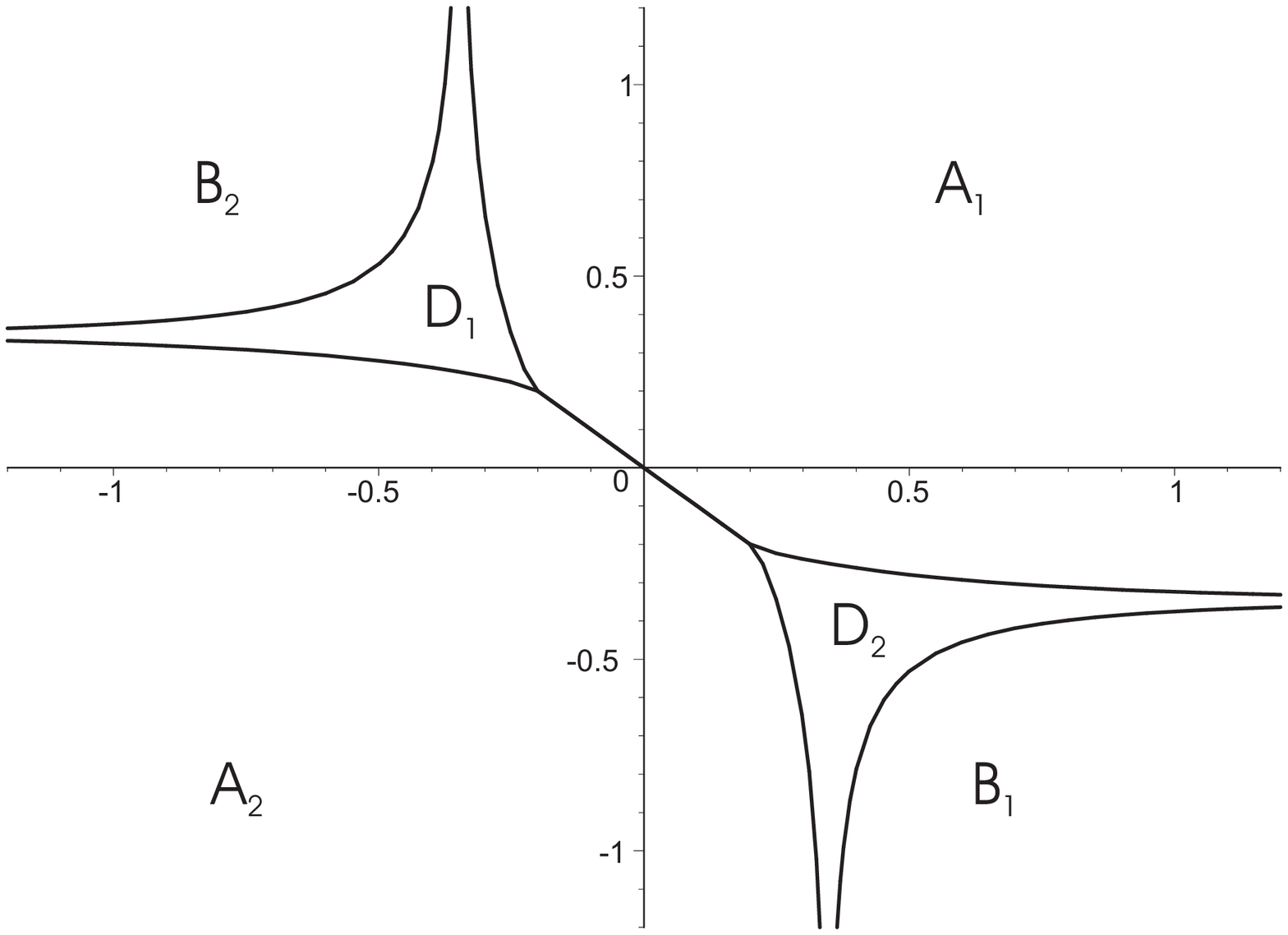}
\end{center}
\caption{The phase diagram in the $(H,V)$-plane for $a=2$, $b=1$,
and $c=0.8$} \label{ferro_diagram_1}
\end{figure}

$\bullet$ $b>a+c$, see Fig. \ref{ferro_diagram_2},
\begin{eqnarray}
&& \mbox{$A_1$-region:}\qquad(e^{2H}-b/a)(e^{2V}-b/a)\geq(c/a)^2,
\qquad e^{2H}> b/a; \nonumber
\\
&&
\mbox{$A_2$-region:}\qquad(e^{-2H}-b/a)(e^{-2V}-b/a)\geq(c/a)^2,
\qquad e^{-2H}> b/a; \nonumber
\\
&& \mbox{$B_1$-region:}\qquad V-H\geq 0, \qquad
\cosh(2H)\leq\Delta, \nonumber
\\
&& \qquad (e^{2H}-a/b)(e^{-2V}-a/b)\geq(c/b)^2, \qquad e^{2H}>a/b,
\qquad \cosh(2H)>\Delta, \nonumber
\\
&& \mbox{$B_2$-region:}\qquad  V-H\leq 0, \qquad
\cosh(2H)\leq\Delta, \nonumber
\\
&& \qquad (e^{-2H}-a/b)(e^{2V}-a/b)\geq(c/b)^2, \qquad
e^{-2H}>a/b, \qquad \cosh(2H)>\Delta. \nonumber
\end{eqnarray}

\begin{figure}[t]
\begin{center}
\includegraphics[width=8cm, height=8cm]{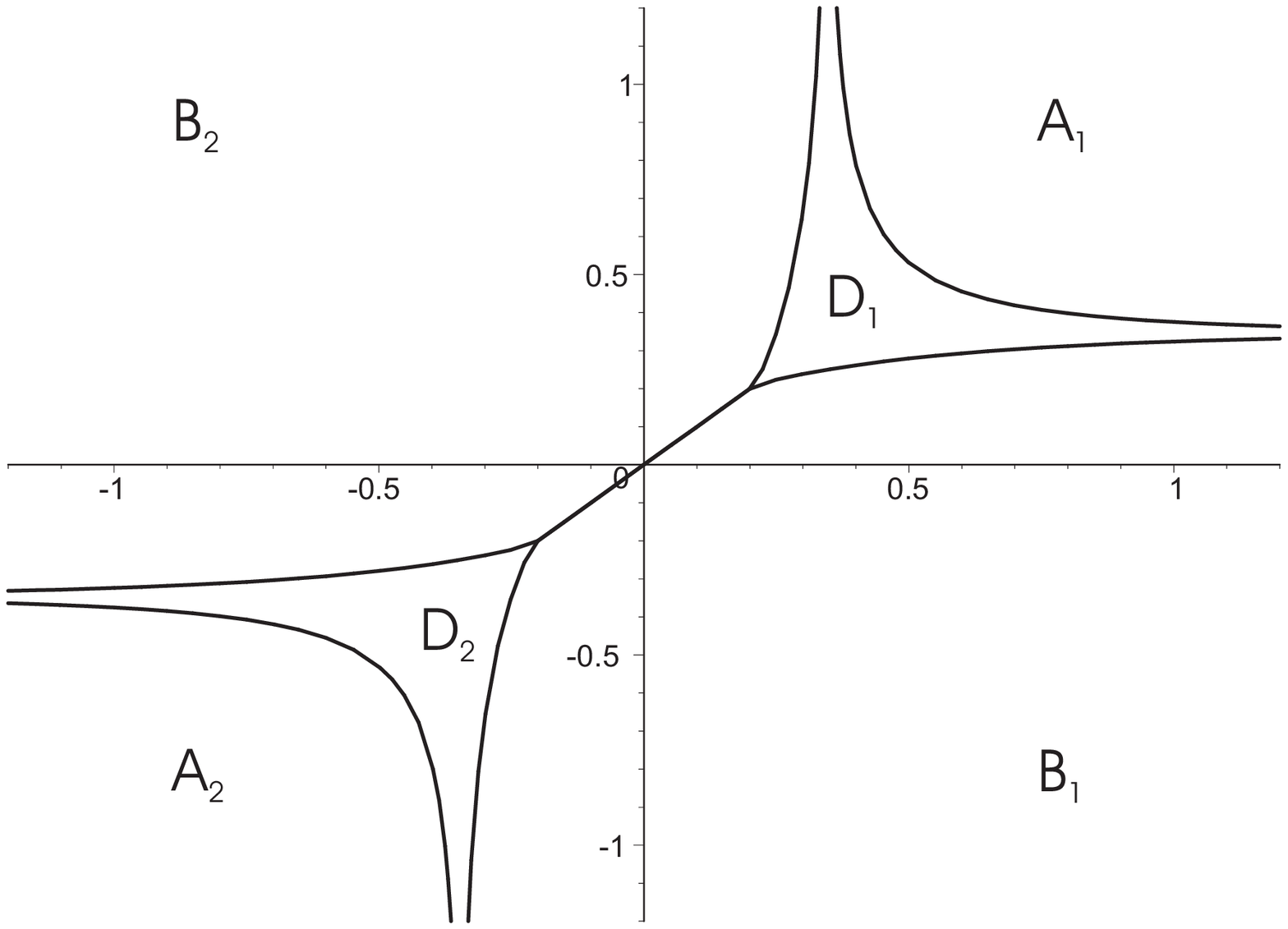}
\end{center}
\caption{The phase diagram in the $(H,V)$-plane for $a=1$, $b=2$,
and $c=0.8$} \label{ferro_diagram_2}
\end{figure}

The free energy is a linear function in $H$ and $V$ in the four
frozen regions:
\begin{eqnarray}
&& f=-\ln a-H-V\qquad\mbox{in}\quad A_1, \nonumber
\\
&& f=-\ln b+H-V\qquad\mbox{in}\quad B_2,
\\
&& f=-\ln a+H+V\qquad\mbox{in}\quad A_2, \nonumber
\\
&& f=-\ln b-H+V\qquad\mbox{in}\quad B_1. \nonumber
\label{linear_energy}
\end{eqnarray}

The regions $D_1$ and $D_2$ are disordered phases. If $(H,V)$ is in one of these regions, local correlation functions are determined by the unique Gibbs measure with the polarization given by the gradient of the free energy.
In this phase the system is disordered,
which means that local correlation functions decay as a power of
the distance $d(e_i,e_j)$ between $e_i$ and $e_j$ when
$d(e_i,e_j)\to \infty$.

In the regions $D_1$ and $D_2$ the free energy is given by
\cite{SY}:
\begin{multline}\label{free_energy}
f(H,V)=
\min(\min_\alpha\left(E_1-H-(1-2\alpha)V-\frac{1}{2\pi i} \int_C\ln(\frac{b}{a}-\frac{c^2}{ab-a^2 z})\rho(z)dz\right),
\\
\min_\alpha\left(E_2+H-(1-2\alpha)V-
\frac{1}{2\pi i}
\int_C\ln(\frac{a^2-c^2}{ab}+\frac{c^2}{ab-a^2
z})\rho(z)dz\right)),
\end{multline}
where $\rho(z)$ can be found from the integral equation
\begin{equation}\label{int-eqn}
\rho(z)=\frac{1}{z}+\frac{1}{2\pi i}\int_C
\frac{\rho(w)}{z-z_2(w)}dw -\frac{1}{2\pi i}\int_C
\frac{\rho(w)}{z-z_1(w)}dw,
\end{equation}
in which
$$
z_1(w)=\frac{1}{2\Delta-w}, \qquad z_2(w)=-\frac{1}{w}+2\Delta.
$$
$\rho(z)$ satisfies the following normalization condition:
$$
\alpha=\frac{1}{2\pi i}\int_C\rho(z)dz.
$$
The contour of integration $C$ (in the complex $z$-plane) is
symmetric with respect to the conjugation $z\rightarrow\bar{z}$,
is dependent on $H$  and is
defined by the condition that the form $\rho(z)dz$ has purely
imaginary values on the vectors tangent to $C$:
$$
{\rm Re}(\rho(z)dz)\Bigr|_{z\in\,\, C}=0.
$$

The formula (\ref{free_energy}) for the free energy follows from
the Bethe Ansatz diagonalization of the row-to-row
transfer-matrix. It relies on a number of conjectures that
are supported by numerical and analytical evidence and in physics
are taken for granted. However, there is no rigorous proof.

There are two points where three phases coexist (two frozen and
one disordered phase). These points are called {\it tricritical}.
The angle $\theta$ between the boundaries of $D_1$ (or $D_2$) at a
tricritical point is given by
$$
\cos(\theta)=\frac{c^2}{c^2+2\min(a,b)^2(\Delta^2-1)}.
$$
The existence of such points makes the $6$-vertex model (and its
degeneration known as the $5$-vertex model \cite{HWKK}) remarkably
different from dimer models \cite{KOS} where generic singularities
in the phase diagram are cusps. Physically, the existence of
singular points where two curves meet at the finite angle
manifests the presence of interaction in the $6$-vertex model.

Notice that when $\Delta= 1$ the phase diagram of the model has a
cusp at the point $H=V=0$. This is the transitional point between
the region $\Delta>1$ and the region $|\Delta|<1$ which is
described below.

\subsection{The phase diagram $|\Delta|<1$} In this case,
the Boltzmann weights have a convenient parametrization by
trigonometric functions. When $1\geq\Delta\geq 1$
\[
a=r\sin(\lambda-\gamma), b=r\sin(\lambda), c=r\sin(\gamma),
\]
where $0\leq \gamma\leq \pi/2$, $\gamma\leq \lambda\leq \pi$, and
$\Delta=\cos\gamma$.

When $0\geq\Delta\geq -1$
\[
a=r\sin(\gamma-\lambda), b=r\sin(\lambda), c=r\sin(\gamma),
\]
where $0\leq \gamma\leq \pi/2$, $\pi- \gamma\leq \lambda\leq \pi$,
and $\Delta=-\cos\gamma$.

The phase diagram of the $6$-vertex model with $|\Delta|<1$ is
shown on Fig. \ref{disordered_diagram}. The phases $A_i, B_i$ are
frozen and identical to the frozen phases for $\Delta>1$. The
phase $D$ is disordered. For magnetic fields $(H,V)$ the Gibbs
measure is translationally invariant with the slope
$(h,v)=(\frac{\pa f(H,V)}{\pa H}, \frac{\pa f(H,V)}{\pa V})$.

\begin{figure}[t]
\begin{center}
\includegraphics[width=10cm, height=10cm]{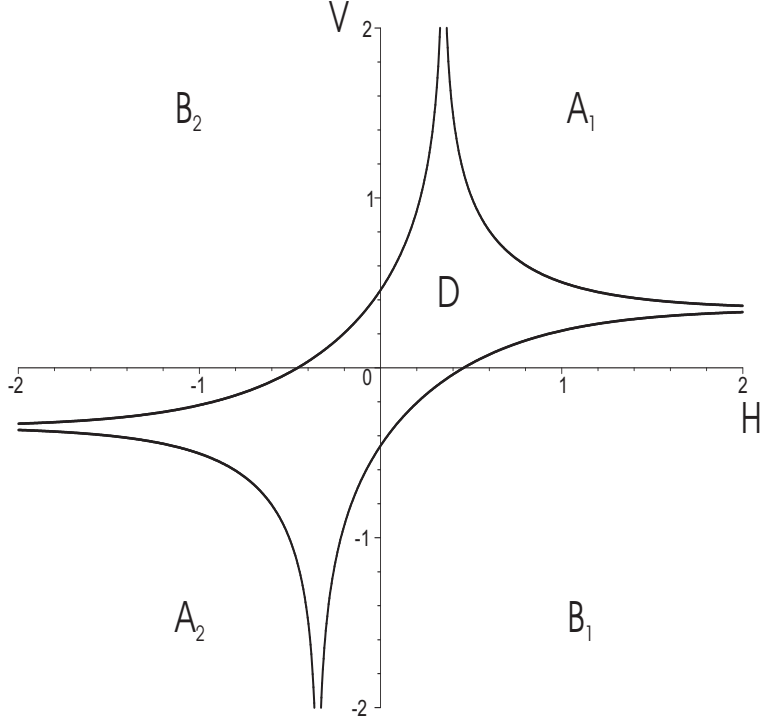}
\end{center}
\caption{The phase diagram in the $(H,V)$-plane for $a=1$, $b=2$,
and $c=2$} \label{disordered_diagram}
\end{figure}

The frozen phases can be described by the following inequalities:
\begin{eqnarray}
&& \mbox{$A_1$-region:}\qquad(e^{2H}-b/a)(e^{2V}-b/a)\geq(c/a)^2,
\qquad e^{2H}> b/a, \nonumber
\\
&&
\mbox{$A_2$-region:}\qquad(e^{-2H}-b/a)(e^{-2V}-b/a)\geq(c/a)^2,
\qquad e^{-2H}> b/a, \nonumber
\\
&& \mbox{$B_1$-region:}\qquad(e^{2H}-a/b)(e^{-2V}-a/b)\geq(c/b)^2,
\qquad e^{2H}>a/b, \label{boundaries_phases}
\\
&& \mbox{$B_2$-region:}\qquad(e^{-2H}-a/b)(e^{2V}-a/b)\geq(c/b)^2,
\qquad e^{-2H}>a/b. \nonumber
\end{eqnarray}

The free energy function in the frozen regions is still given by
the formulae (\ref{linear_energy}). The first derivatives of the
free energy are continuous at the boundary of frozen phases, The
second derivative is continuous in the tangent direction at the
boundary of frozen phases and is singular in the normal direction.

It is smooth  in the disordered region where it is given by
(\ref{free_energy}) which, as in case $\Delta>1$ involves a
solution to the integral equation (\ref{int-eqn}). The contour of
integration in (\ref{int-eqn}) is closed for zero magnetic fields
and, therefore, the equation (\ref{int-eqn}) can be solved
explicitly by the Fourier transformation \cite{Ba} .

The $6$-vertex Gibbs measure with zero magnetic fields converges
in the thermodynamic limit to the superposition of translationally
invariant Gibbs measures with the slope $(1/2,1/2)$. There are two
such measures. They correspond to the double degeneracy of the
largest eigenvalue of the row-to-row transfer-matrix \cite{Ba}.

There is a very interesting relationship between the $6$-vertex
model in zero magnetic fields and the highest weight
representation theory of the corresponding quantum affine algebra.
The double degeneracy of the Gibbs measure with the slope
$(1/2,1/2)$ corresponds to the fact that there are two integrable
irreducible representations of $\widehat{sl_2}$ at level one.
Correlation functions in this case can be computed using
$q$-vertex operators \cite{JM}. For latest developments see
\cite{BJMST}.

\subsection{The phase diagram $\Delta<-1$}

\subsubsection{The phase diagram}The Boltzmann weights for these values of
$\Delta$ can be conveniently parameterized as
\begin{equation}\label{hyp-param}
a=r\sinh(\eta-\lambda), b=r\sinh(\lambda), c=r\sinh(\eta),
\end{equation}
where $0<\lambda<\eta$ and $\Delta=-\cosh\eta$.

The Gibbs measure in thermodynamic limit depends on the value of
magnetic fields. The phase diagram in this case is shown on Fig.
\ref{antiferroelectric_diagram} for $b/a>1$. In the
parameterization (\ref{hyp-param}) this correspond to $0<\lambda<
\eta/2$. When $\eta/2<\lambda <\eta$ the $4$-tentacled ``amoeba''
is tilted in the opposite direction as on Fig.
\ref{ferro_diagram_1}.

When $(H,V)$ is in one of the $A_i,B_i$
regions in the phase diagram the Gibbs measure is supported on the
corresponding frozen configuration, see Fig.
\ref{ferroelectric_phase}.

\begin{figure}[t]
\begin{center}
\includegraphics[height=10cm, width=10cm]{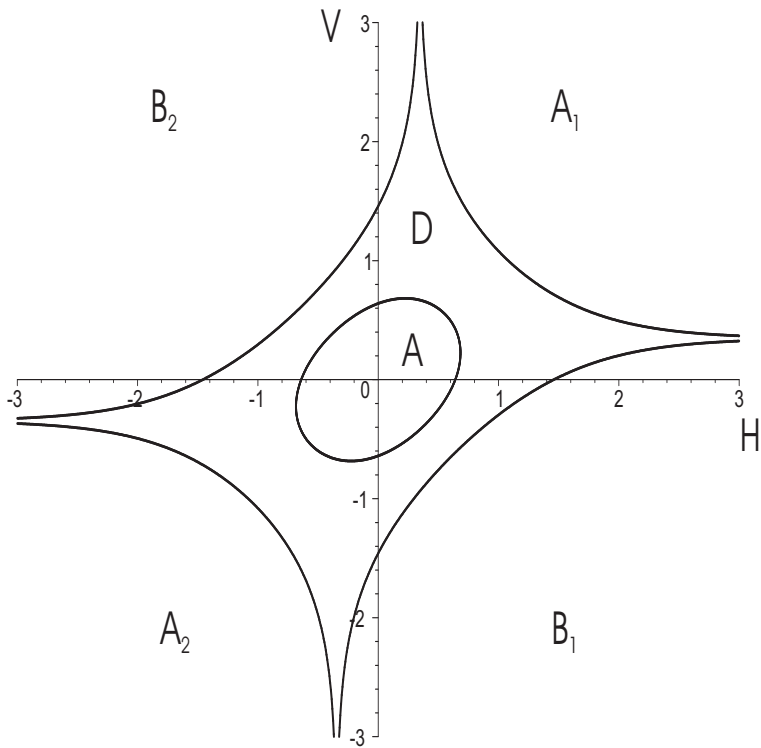}
\end{center}
\caption{The phase diagram in the $(H,V)$-plane for $a=1$, $b=2$,
and $c=6$} \label{antiferroelectric_diagram}
\end{figure}

The boundary between ordered phases $A_i, B_i$ and the disordered phase $D$ is given by inequalities (\ref{boundaries_phases}). The free energy in these regions is linear in electric fields and is given by (\ref{linear_energy}).

If $(H,V)$ is in the region $D$, the Gibbs measure is the
translationally invariant measure with the polarization $(h,v)$
determined by (\ref{slope}). The free energy in this case is
determined by the solution to the linear integral equation
(\ref{int-eqn}) and is given by the formula (\ref{free_energy}).

If $(H,V)$ is in the region $A$, the Gibbs measure is the
superposition of two Gibbs measures with the polyarization $(1/2,1/2)$. In
the limit $\Delta\to -\infty$ these two measures degenerate to two
measures supported on configurations $C_1, C_2$, respectively,
shown on Fig. \ref{antiferr}. For a finite $\Delta$ the support of
these measures consists of configurations which differ from $C_1$
and $C_2$ in finitely many places on the lattice.

\begin{figure}[b]
\begin{center}
\includegraphics[width=8cm, height=4cm]{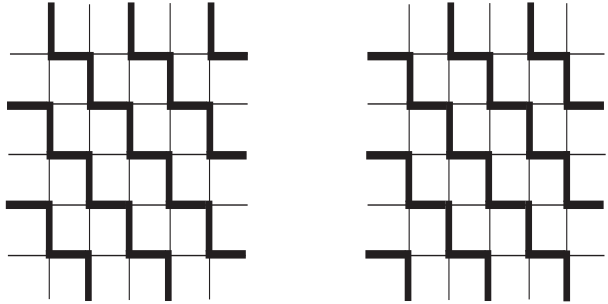}
\end{center}
\caption{The configurations $C_1$ and $C_2$.} \label{antiferr}
\end{figure}

\begin{figure}[h]
\begin{center}
\includegraphics{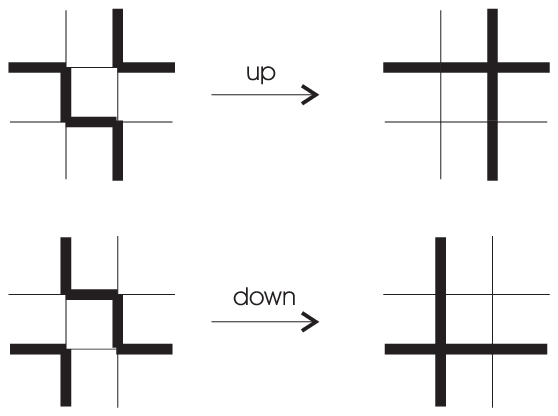}
\end{center}
\caption{The elementary up and down fluctuations in the
antiferromagnetic phase.} \label{fluct}
\end{figure}

\begin{remark} Any two configurations lying in the support of each
of these Gibbs measures can be obtained from $C_1$ or $C_2$ via
flipping the path at a vertex ``up'' or ``down'' as it is shown on
Fig. \ref{fluct} finitely many times. It is also clear that it
takes infinitely many flips to go from $C_1$ to $C_2$.
\end{remark}

\subsubsection{The antiferromagnetic region}\label{AF-region}
The $6$-vertex model in the phase $A$ is disordered and is also
noncritical. Here the non-criticality means that the local correlation function
$\langle\sigma_{e_i}\sigma_{e_j}\rangle$ decays as $\exp(-\alpha
d(e_i,e_j))$ with some positive $\alpha$ as the distance
$d(e_i,e_j)$ between $e_i$ and $e_j$ increases to infinity.

The free energy in the $A$-region can be explicitly computed by
solving the equation (\ref{int-eqn}). In this case the largest eigenvalue will correspond to $n=N/2$ , the contour of integration in (\ref{int-eqn}) is closed, and the equation can be solved by the Fourier transform. \footnote{Strictly speaking this is a conjecture supported by the numerical evidence}.

The boundary between the antiferromagnetic region $A$
and the disordered region $D$ can be derived
similarly to the boundaries of the ferromagnetic regions $A_i$ and
$B_i$  by analyzing next to the largest eigenvalue of the
row-to-row transfer-matrix.
This computation was done in \cite{SY}, \cite{LW}. The result is a
simple closed curve, which can be described parameterically as
$$
H(s)=\Xi(s), \qquad V(s)=\Xi(\eta-\theta_0+s),
$$
where
$$
\Xi(\varphi)=\cosh^{-1}\Bigl(\frac{1}{{\rm
dn}(\frac{K}{\pi}\varphi|1-\nu)}\Bigr),
$$
$$
|s|\leq 2\eta,
$$
and
$$
e^{\theta_0}=\frac{1+\max(b/a,a/b)e^{\eta}}{\max(b/a,a/b)+e^{\eta}}.
$$
The parameter  $\nu$ is defined by the equation $\eta K(\nu)=\pi
K'(\nu)$, where
$$
K(\nu)=\int_0^{\pi/2}(1-\nu\sin^2(\theta))^{-1/2}d\theta \qquad
K'(\nu)=\int_0^{\pi/2}(1-(1-\nu)\sin^2(\theta))^{-1/2}d\theta.
$$

The curve is invariant with respect to the reflections $(H,V)\to
(-H,-V)$ and $(H,V)\to (V,H)$ since the function $\Xi$ satisfies
the identities
$$
\Xi(\varphi)=-\Xi(-\varphi), \qquad
\Xi(\eta-\varphi)=\Xi(\eta+\varphi).
$$
This function is also $4\eta$-periodic:
$\Xi(4\eta+\varphi)=\Xi(\varphi)$.

As it was shown in \cite{PaR} this curve is algebraic
in $e^H$ and $e^V$ and can be written as
\begin{eqnarray}
&& \Bigl((1-\nu\cosh^2V_0) \cosh^2H+\sinh^2V_0-(1-\nu)\cosh V_0
\cosh H \cosh V\Bigr)^2= \nonumber
\\
&& (1-\nu\cosh^2V_0)\sinh^2V_0\cosh^2V\sinh^2H(1-\nu\cosh^2H),
\label{antiferroelectic_boundary}
\end{eqnarray}
where $V_0$ is the positive value of $V$ on the curve when $H=0$.
Notice that $\nu$ depends on the Boltzmann weights $a,b,c$ only
through $\eta$.

\section{Some asymptotics of the free energy}

\subsection{The scaling near the boundary of the $D$-region}
Assume that $\vec{H}_0=(H_0,V_0)$ is  a regular point
at the boundary between the disordered region and and the $A_1$
-region, see Fig.\ref{disordered_diagram}. Recall that this boundary is the curve defined by the equation
\[
g(H,V)=0
\]
where
\begin{equation}\label{inter_D_F}
g(H,V)=\ln(b/a+\frac{c^2/a^2}{e^{2H}-b/a})-2V.
\end{equation}

Denote the normal vector to the boundary of the $D$-region at $\vec{H}_0$ by $\vec{n}$ and the tangent vector pointing inside of
the region $D$ by $\vec{\tau}$.

We will study the asymptotic of the free energy along the curves $\vec{H}(r,s,t)=\vec{H_0}+r^2s\vec{n}+rt\vec{\tau}$, as $r\to 0$.
It is clear that  $\vec{H}(r,s,t)$ is in the  $D$-region  if $s\geq 0$ .

\begin{theorem}
Let $\vec{H}(r,s,t)$ be defined as above. The asymptotic of the
free energy of the $6$-vertex model in the limit $r\rightarrow 0$
is given by
\begin{equation}\label{asympt_free_energy_interface}
f(\vec{H}(r,s,t))=f_{lin}(\vec{H}(r,s,t)) +\eta(s,t)r^3+O(r^5),
\end{equation}
where $f_{lin}(H,V)=-\ln(a)-H-V$ and
\begin{equation}\label{scal_lim_bound}
\eta(s,t)=-\kappa \left(\theta s+t^2\right)^{3/2}.
\end{equation}
Here the constants $\kappa$ and $\theta$ depend on the Boltzmann
weights of the model and on $(H_0,V_0)$ and are given by
$$
\kappa=\frac{16}{3\pi}\,\partial_H^2 g(H_0,V_0)
$$
and
$$
\theta=\frac{4+\left(\partial_H g(H_0,V_0)\right)^2}{2\partial_H^2
g(H_0,V_0)},
$$
where $g(H,V)$ is defined in (\ref{inter_D_F}).

Moreover, $\partial_H^2g(H_0,V_0)>0$ and, therefore, $\theta>0$.
\end{theorem}

We refer the reader to \cite{PaR} for the details. This behavior is universal in a sense that the exponent $3/2$ is the same for
all points at the boundary.

\subsection{The scaling in the tentacle}
Here, to be specific we assume that $a>b$. The theorem below describes the asymptotic of
the free energy function when $H\to +\infty$ and
\begin{equation}
\frac{1}{2}\ln(b/a)-\frac{c^2}{2ab}\,e^{-2H}\leq V
\leq\frac{1}{2}\ln(b/a)+\frac{c^2}{2ab}\,e^{-2H}, \qquad
H\longrightarrow\infty. \label{tentacle}
\end{equation}
These values of $(H,V)$ describe points inside the right
``tentacle'' on the Fig. \ref{ferro_diagram_1}.

Let us parameterize these values of $V$ as
$$
V=\frac{1}{2}\ln(b/a)+\beta\frac{c^2}{2ab}e^{-2H},
$$
where $\beta\in[-1,1]$.

\begin{theorem}
When $H\to \infty$ and $\beta\in [-1,1]$ the asymptotic of the
free energy is given by the following formula:
\begin{eqnarray}\label{tent-ass}
f(H,V)=-\frac{1}{2}\ln(ab)-H-\nonumber
\\  \frac{c^2}{2ab}\,e^{-2H} \Bigl(\beta+
&&
\frac{2}{\pi}\sqrt{1-\beta^2}-\frac{2}{\pi}\beta\arccos(\beta)\Bigr)
+O(e^{-4H}), \nonumber
\end{eqnarray}
\end{theorem}
The proof is given in \cite{PaR}.

\subsection{The $5$-vertex limit}
The $5$-vertex model can be obtained as the limit of the
$6$-vertex model when $\Delta\rightarrow\infty$. Magnetic fields
in this limit behave as follows:
\begin{itemize}

\item $a>b+c$. In the parametrization (\ref{dg1par-1}) after
changing variables $H=\frac{\eta}{2}+l$, and $V=-\frac{\eta}{2}+m$
take the limit $\eta\to \infty$ keeping $\lambda$ fixed. The
weights will converge (up to a common factor) to:
\[
a_1:a_2:b_1:b_2:c_1:c_2\to e^{\lambda+l+m}:e^{\lambda-l-m}:
(e^\lambda-e^{-\lambda})e^{l-m}:0:1:1
\]

\item $a+c<b$. In the parametrization (\ref{dg1par-2}) after
changing variables $H=\frac{\eta}{2}+l$, and $V=\frac{\eta}{2}+m$
take the limit $\eta\to \infty$ keeping $\xi=\lambda-\eta$ fixed.
The weights will converge (up to a common factor) to:
\[
a_1:a_2:b_1:b_2:c_1:c_2\to
(e^\xi-e^{-\xi})e^{l+m}:0:e^{\xi+l-m}:e^{\xi-l+m}:1:1
\]

\end{itemize}
The two limits are related by inverting horizontal arrows. From
now on we will focus on the 5-vertex model obtained by the limit
from the 6-vertex one when $a>b+c$.

The phase diagram of the 5-vertex model is easier then the one for
the 6-vertex model but still sufficiently interesting. Perhaps the
most interesting feature is that the existence of the tricritical
point in the phase diagram.

We will use the parameter
\[
\gamma=e^{-2\lambda}
\]
Notice that $\gamma<1$.

The frozen regions on the phase diagram of the $5$-vertex model,
denoted on Fig. \ref{amoeba_5v} as $A_1$, $A_2$, and $B_1$, can be
described by the following inequalities:
\begin{eqnarray}
&& \mbox{$A_1$-region:}\qquad m\geq -l,\qquad l\leq 0, \nonumber
\\
&& \qquad\qquad\qquad e^{2m}\geq 1-\gamma(1-e^{-2l}), \qquad l>1;
\nonumber
\\
&& \mbox{$A_2$-region:}\qquad m\leq -l, \qquad l\leq 0,
\\
&& \qquad\qquad\qquad e^{2m}\leq 1-\frac{1}{\gamma}(1-e^{-2l}),
\qquad l>1; \nonumber
\\
&& \mbox{$B_1$-region:}\qquad
(e^{2l}-\frac{1}{1-\gamma})(e^{-2m}-\frac{1}{1-\gamma})\geq
\frac{\gamma}{(1-\gamma)^2}, \qquad e^{2l}>\frac{1}{1-\gamma};
\nonumber
\end{eqnarray}

As it follows from results \cite{HWKK} the limit from the 6-vertex
model to the 5-vertex model commutes with the thermodynamical
limit and for the free energy of the $5$-vertex model we can use
the formula
\begin{equation}
\label{free_en_5v} f_5(l,m)=\lim_{\eta\rightarrow+\infty}
(f(\eta/2+l,-\eta/2+m)-f(\eta/2,-\eta/2)),
\end{equation}
where $f(H,V)$ is the free energy of the $6$-vertex model.

\begin{figure}[t]
\begin{center}
\includegraphics[width=8cm, height=8cm]{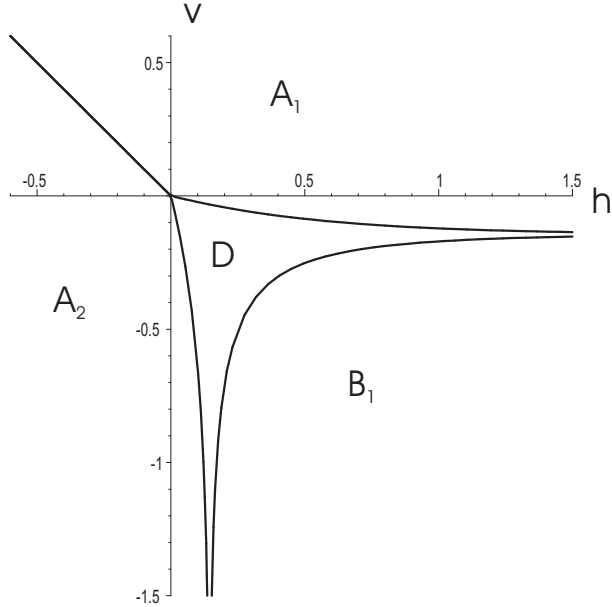}
\end{center}
\caption{The phase diagram of the $5$-vertex model with
$\gamma=1/4$ ($\beta=e^{-2h}$).} \label{amoeba_5v}
\end{figure}

\subsection{The asymptotic of the free energy near the
tricritical point in the 5-vertex model }

The disordered region $D$ near the tricritical point forms a
corner
$$
-\frac{1}{\gamma}\,l+O(l^2)\leq m\leq -\gamma\,l+O(l^2), \qquad
h\rightarrow 0+.
$$
The angle $\theta$ between the boundaries of the disordered region
at this point is given by
$$
\cos(\theta)=\frac{2\gamma}{1+\gamma^2}.
$$

One can argue that the finiteness of the angle $\theta$ manifests
the presence of interaction in the model. In comparison,
translation invariant dimer models most likely can only have cusps as such singularities.

Let $\gamma\leq k\leq \frac{1}{\gamma}$ and
$$
m=-kl,
$$
As it was shown in \cite{BS} for $m=-l$ and in \cite{PaR} for
$m=-kl$ with $\gamma\leq k\leq \frac{1}{\gamma}$ the asymptotic of the free energy as $l\to +0$ is given by
\begin{equation}
f(l,-kl)=c_1(k,\gamma)l +c_2(k,\gamma)l^{5/3}+O(l^{7/3}).
\end{equation}
where
\begin{equation}
c_1(k,\gamma)=\frac{1}{1-\gamma}\left(-(1+k)(1+\gamma)+4\sqrt{k\gamma}\right),
\end{equation}
and
\begin{equation}
c_2(k,\gamma)=(6\pi)^{2/3}\frac{2\gamma^{5/6}(1-\gamma)
k^{3/2}(\sqrt{k}-1/\sqrt{\gamma})^{4/3}}{5(\sqrt{k}-\sqrt{\gamma})^{4/3}}.
\end{equation}

The scaling along any ray inside the corner near the tricritical
point in the 6-vertex model differ from this only by in
coefficients. The exponent $h^{5/3}$ is the same.

\subsection{The limit $\Delta\to -1^-$}
If $\Delta=-1$, the region $A$ consists of one point located at the origin.

It is easy to find the asymptotic of the function $\Xi(\varphi)$ when $\eta\rightarrow 0+$, or $\Delta\rightarrow -1-$.
In this limit $K'\rightarrow \pi/2$,
$K\rightarrow \frac{\pi^2}{2\eta}$. Using the asymptotic
$\frac{1}{{\rm
dn}(u|1-m)}\sim 1+\frac{1}{2}(1-m)\sin^2(u)$ when $m\rightarrow 1-$, and
$\cosh^{-1}(x)\sim \pm\sqrt{2(x-1)}$, when $x\rightarrow 1+$, we
have
$$
\theta_0=\frac{|b-a|\eta}{a+b}.
$$
and \cite{LW}:
$$
\Xi(\varphi)\sim 4e^{-\frac{\pi^2}{2\eta}}
\sin(\frac{\pi}{2\eta}\varphi).
$$

The gap in the spectrum of elementary excitations vanishes
in this limit at the same rate as $\Xi(\varphi)$.
At the lattice distances of order $e^{\frac{\pi^2}{2\eta}}$
the theory has a scaling limit and become the relativistic $SU(2)$
chiral Thirring model. Correlation function of vertical and horizontal edges in the $6$-vertex model become correlation functions of the currents in the Thirring model.


\subsection{The convexity of the free energy} The following identity
holds in the region $D$ \cite{BIR},\cite{NK}:
\begin{equation}\label{hessian}
f_{H,H}f_{V,V}-f_{H,V}^2=\left(\frac{2}{\pi g}\right)^2.
\end{equation}

Here $g=\frac{1}{2D_0^2}$. The constant $D_0$ does not vanish in
the $D$-region including its boundary. It is determined by the
solution to the integral equation for the density $\rho(z)$.

Directly from the definition of the free energy we have
\[
f_{H,H}=\lim _{N,M\to \infty}\frac{<(n(L)-n(R))^2>}{NM},
\]
where $n(L)$ and $n(R)$ are the number of arrows pointing to the
left and the number of arrows pointing to the right, respectively.

Therefore, the matrix  $\pa_i\pa_j f$ of second derivatives with
respect to $H$ and $V$  is positive definite.

As it follows from the asymptotical behavior of the free energy
near the boundary of the $D$-phase, despite the fact that the
Hessian is nonzero and finite at the boundary of the interface,
the second derivative of the free energy in the transversal
direction at a generic point of the interface develops a
singularity.

\section{The Legendre Transform of the Free Energy}
The Legendre transform of the free energy
$$
\sup_{H,V}\Bigl(xH+yV+f(H,V)\Bigr)
$$
as a function of $(x,y)$ is defined for $-1\leq x , y\leq 1$.

The variables $x$ and $y$ are known as polarizations and are
related to the slope of the Gibbs measure as $x=2h-1$ and
$y=2v-1$. We will write the Legendre transform of the free energy
as a function of $(h,v)$
\begin{equation}\label{surface_tension}
\sigma(h,v)=\sup_{H,V}\Bigl((2h-1)H+(2v-1)V+f(H,V)\Bigr).
\end{equation}
$\sigma(h,v)$ is defined on $0\leq h,v\leq 1$.

For the periodic boundary conditions the surface tension function
has the following symmetries:
$$
\sigma(x,y)=\sigma(y,x)=\sigma(-x,-y)=\sigma(-y,-x).
$$
The last two equalities follow from the fact that if all arrows
are reversed, $\sigma$ is the same, but the signs of $x$ and $y$
are changed. It follows that $\sigma_h(h,v)=\sigma_v(v,h)$ and
$\sigma_v(h,v)=\sigma_h(v,h)$.

The function $f(H,V)$ is linear in the domains that correspond to
conic and corner singularities of $\sigma$. Outside of these
domains (in the disordered domain $D$) we have
\begin{equation}\label{f-sig-comp}
\nabla\sigma\circ\nabla f={\rm id}_D, \ \nabla f\circ\nabla \sigma
={\rm id}_{\nabla f(D)}.
\end{equation}
Here the gradient of a function as a mapping $\RR^2\to \RR^2$.

When the 6-vertex model is formulated in terms of the height
function, the Legendre transform of the free energy can be
regarded as a surface tension. The surface in this terminology is
the graph of of the height function.

\subsection{}
Now let us  describe some analytical properties of the function
$\sigma(h,v)$ is obtained as the Legendre transform of the free
energy. The Legendre transform maps the regions where the free
energy is linear with the slope $(\pm 1, \pm 1)$ to the corners of
the unit square ${\mathcal D}=\{(h,v)|\quad 0\leq h\leq 1, 0\leq
v\leq 1 \}$. For example, the region $A_1$ is mapped to the corner
$h=1$ and $v=1$ and the region $B_1$ is mapped to the corner $h=1$
and $v=0$. The Legendre transform maps the tentacles of the
disordered region to the regions adjacent to the boundary of the
unit square. For example, the tentacle between $A_1$ and $B_1$
frozen regions is mapped into a neighborhood of $h=1$ boundary of
$\mathcal D$, i.e. $h\to 1$ and $0<v<1$.

Applying the Legendre transform to asymptotics of the free energy
in the tentacle between $A_1$ and $B_1$ frozen regions we get
$$
H(h,v)=-\frac{1}{2}\ln\left(\frac{\pi
ab}{c^2}\frac{1-h}{\sin{\pi(1-v)}}\right),\qquad
V(h,v)=\frac{1}{2}\ln(b/a)+\frac{\pi}{2}(1-h)\cot(\pi(1-v)),
$$
and
\begin{equation}\label{sur_ten_boundary}
\sigma(h,v)=(1-h)\ln\left(\frac{\pi
ab}{c^2}\frac{1-h}{\sin(\pi(1-v))}\right) -(1-h)+v\ln(b/a)-\ln(b),
\end{equation}
Here $h\to 1-$ and $0<v<1$. From (\ref{sur_ten_boundary}) we see
that $\sigma(1,v)=v\ln(b/a)-\ln(b)$, i.e. $\sigma$ is linear on
the boundary $h=1$ of $\mathcal D$. Therefore, its asymptotics
near the boundary $h=1$ is given by
$$
\sigma(h,v)=v\ln(b/a)-\ln(b)+(1-h)\ln(1-h)+O(1-h),
$$
as $h\to 1-$ and $0<v<1$. We note that this expansion is valid
when $(1-h)/\sin(\pi(1-v))\ll 1$.

Similarly, considering other tentacles of the region $D$, we
conclude that the surface tension function is linear on the
boundary of $\mathcal D$.

\subsection{}
Next let us find the asymptotics of $\sigma$ at the corners of
$\mathcal D$ in the case when all points of the interfaces between
frozen and disordered regions are regular, i.e. when $\Delta<1$.
We use the asymptotics of the free energy near the interface
between $A_1$ and $D$ regions
(\ref{asympt_free_energy_interface}).

First let us fix the point $(H_0,V_0)$ on the interface and the
scaling factor $r$ in (\ref{asympt_free_energy_interface}). Then
from the Legendre transform we get
$$
1-h=-\frac{3}{4}\,r\,\frac{\kappa(\theta s+t^2)^{1/2}}{(\partial_H
g)^2+4} (\theta\partial_H g+4rt)
$$
and
$$
1-v=-\frac{3}{2}\,r\,\frac{\kappa(\theta s+t^2)^{1/2}}{(\partial_H
g)^2+4} (-\theta+r\partial_H g t).
$$
It follows that
$$
\frac{1-h}{1-v}=\frac{\theta\partial_H g
+4rt}{2(-\theta+rt\partial_H g)}.
$$
In the vicinity of the boundary $r\to 0$ and, hence,
\begin{equation}\label{slope_corner}
\frac{1-h}{1-v}=-\frac{\partial_H g }{2}= \frac{1-b/a\,\,
e^{-2V_0}}{1-b/a\,\, e^{-2H_0}}
\end{equation}
as $h,v\to 1$. Thus, under the Legendre transform, the slope of
the line which approaches the corner $h=v=1$ depends on the
boundary point on the interface between the frozen and disordered
regions.

It follows that the first terms of the asymptotics of $\sigma$ at
the corner $h=v=1$ are given by
$$
\sigma(h,v)=-\ln a-2(1-h)H_0(h,v)-2(1-v)V_0(h,v),
$$
where $H_0(h,v)$ and $V_0(h,v)$ can be found from
(\ref{slope_corner}) and $g(H_0,V_0)=0$.

When $|\Delta|< 1$ the function $\sigma$ is strictly convex and
smooth for all $0<h,v<1$. It develops conical singularities near
the boundary.

When $\Delta < -1$, in addition to the singularities on the
boundary, $\sigma$ has a conical singularity at the point
$(1/2,1/2)$. It corresponds to the ``central flat part'' of the
free energy $f$, see Fig. \ref{antiferroelectric_diagram}.

When $\Delta >1$ the function $\sigma$ has corner singularities
along the boundary as in the other cases. In addition to this, it
has a corner singularity along the diagonal $v=h$ if $a>b$ and
$v=1-h$ if $a<b$. We refer the reader to \cite{BS} for further
details on singularities of $\sigma$ in the case when $\Delta>1$.

When $\Delta=-1$ function $\sigma$ has a corner singularity at $((1/2,1/2)$.

\section{The limit shape phenomenon}
\subsection{The Height Function for the 6-vertex model}

Consider the square grid $L_\epsilon\subset \RR^2$ with
the step $\epsilon$. Let $D\subset \RR^2$ be a domain in $\RR^2$. Denote by $D_\epsilon$ a domain
in the square lattice which corresponds to the intersection
$D\cap L_\e$, assuming that the intersection is generic, i.e.
the boundary of $D$ does not intersect vertices of $L_\e$.

Faces of $D_\e$ which do not intersect the boundary $\pa D$ of $D$
are called {\it inner faces}. Faces of $D_\e$ which intersect the boundary of $D$ are called boundary faces.

A height function $h$ is an integer-valued function on the faces
of $D_\e$ of the grid $L_N$ (including the outer faces) which is

\begin{itemize}

\item non-decreasing when going up or to the right,

\item if $f_1$ and $f_2$ are neighboring faces then $h(f_1)-h(f_2)=-1,0,1$.

\end{itemize}

{\it Boundary value} of the height function is its restriction
to the outer faces.
Given a function $h^{(0)}$ on the set of boundary faces denote ${\mathcal H}(h^{(0)})$
the space of all height functions with the boundary value $h^{(0)}$. Choose a marked face $f_0$ at the boundary. A height function is normalized at this face if $h(f_0)=0$.

\begin{figure}[t]
\begin{center}
\includegraphics[height=5cm, width=5cm]{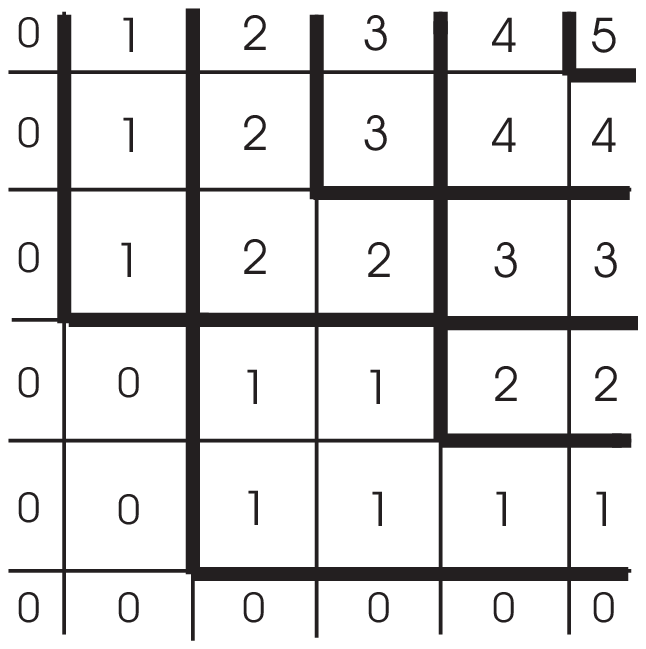}
\end{center}
\caption{The values of the height function for the configuration of paths
given on Fig. \ref{dwbc_conf}.}
\label{hf_example}
\end{figure}

\begin{proposition}
There is a bijection between the states of the $6$-vertex model with fixed boundary conditions, and height functions with corresponding boundary values normalized at $f_0$ .
\end{proposition}

\begin{proof}
Indeed, given a height function consider its ``level curves,'' i.e. paths on $D_\e$, where the height function changes its value by
$1$, see Fig. \ref{hf_example}. Clearly this defines a state for
the $6$-vertex model on $D_\e$ with the boundary conditions
determined by the boundary values of the height function.

On the other hand, given a state in the $6$-vertex model, consider the
corresponding configuration of paths. That there is a unique height
function whose level curves are these paths and which satisfies
the condition $h=0$ at $f_0$.

It is clear that this correspondence is a bijection.
\end{proof}

There is a natural partial order on the set of height functions
with given boundary values. One function is bigger than the other,
$h_1\leq h_2$,
if it is entirely  above the other, i.e. if $h_1(x)\leq h_2(x)$
for all $x$ in the domain. There exist the minimum
$h_{\rm min}$ and the maximum $h_{\rm max}$ height functions such
that $h_{\rm min}\leq h\leq h_{\rm max}$ for all height functions
$h$. for DW boundary conditions maximal and minimal height functions ar shown on Fig. \ref{DW_min_max}.

\begin{figure}[b]

\begin{center}
\includegraphics[width=8cm, height=4cm]{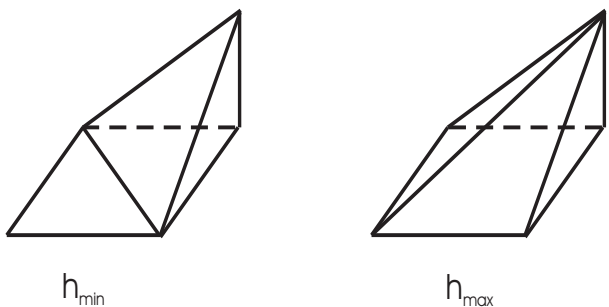}
\end{center}
\caption{The minimum and maximum height functions for
the DW boundary conditions.}
\label{DW_min_max}
\end{figure}

The characteristic function of an edge is related to the hight function as
\begin{equation}\label{e-f}
\sigma_e=h(f^+_e)-h(f^-_e)
\end{equation}
where faces $f^\pm_e$ are adjacent to $e$, $f^+_e$ is to the right
of $e$ for the vertical edge and on the top for the horizontal one.
The face $f^-_e$ is on the left for the vertical edge and below $e$
for the horizontal edge.

Thus, we can consider the $6$-vertex model on a domain as a theory of fluctuating discrete
surfaces constrained between $h_{\rm min}$ and $h_{\rm max}$.
Each surface occurs with probability given by the
Boltzmann weights of the $6$-vertex model.

The height function does not exist when the region i snot simply-connected or when a square lattice is on a surface with
non-trivial fundamental group. One can draw analogy between states
of the $6$-vertex model and $1$-forms. States on a domain with trivial fundamental group can be regarded as exact forms $\omega=dh$ where $h$ is a height function.

\subsection{Stabilizing Fixed Boundary Conditions}

Recall that the height function is a monotonic integer-valued
function on the faces of the grid, which satisfies the Lipschitz
condition (it changes at most by $1$ on any two adjacent faces).

The normalized height function is a piecewise constant function
on $D_\e$ with the value
$$
h^{\rm norm}(x,y)=\e \,h_\e(n,m).
$$
where $h_\e(n,m)$ is a height function on $D_\e$.

Normalized height functions on the boundary of $D$ satisfy the
inequality
$$
|h(x,y)-h(x',y')|\leq |x-x'|+|y-y'|.
$$

As for non-normalized height functions there is a natural partial ordering on the set of all normalized height functions
with given boundary values: $h_1\geq h_2$ if $h_1(x)\geq h_2(x)$ for all $x\in D$.
We define the operations
$$
h_1\vee h_2=\min_{x\in D}(h_1(x),h_2(x)), \qquad
h_1\wedge h_2=\max_{x\in D}(h_1(x),h_2(x)).
$$
It is clear that
$$
h_1\vee h_2\leq h_1,h_2\leq h_1\wedge h_2
$$
It is also clear that in this partial ordering there is a unique minimal and maximal
height functions, which we denote by $h^{min}$ and $h^{max}$, respectively.

The boundary value of the height function defines a piecewise
constant function on boundary faces of $D_\e$.
Consider a sequence of domains $D_\e$ with $\e\to 0$.
{\it Stabilizing fixed boundary conditions} for height functions
is a sequence  $\{h^{(\e)}\}$ of functions on boundary faces
of $D_\e$ such that $h^{(\e)}\to h^{(0)}$ where $h^{(0)}$
is a continuous function of $\pa D$.

It is clear that in the limit $\e\to 0$ we have at least two characteristic scales. At the {\it macroscopical} scale we can "see" the region $D\subset \RR^2$. Normalized height functions
in the limit $\e\to 0$ will be functions on $D$. As we will see in the next section, the height function in the $6$-vertex model
develops deterministic limit shape $h_0(x,y)$ at this scale. It is reflected in the structure of local correlation functions. Let $e_i$ be edges with coordinates $(n_i,m_i)=(\frac{x_i}{\e},\frac{y_i}{\e})$.
When  $\e\to 0$ and $(x_i,y_i)$ are fixed
\[
<\sigma_{e_1}\dots \sigma_{e_k}>= \pa_{i_1} h_0(x_1,y_1)\dots \pa_{i_k} h_0(x_k,y_k)+O(e^{-\frac{c}{\e}})
\]

The randomness remain at the smaller scale, and in particular at the lattice, microscopical, scale. If coordinates of edges $e_i$
are $(n_i,m_i)=(\frac{x}{\e}+\Delta n_i,\frac{y}{\e}+\Delta m_i)$
we expect that in the limit $\e\to 0$ correlation functions have
the limit
\[
<\sigma_{e_1}\dots \sigma_{e_k}>\to \sigma_{e_1}\dots \sigma_{e_k}>_{(\pa_xh_0(x,y),\pa_y h_0(x,y))}
\]
where the correlation function at the right side is taken
with respect to the translation invariant Gibbs nmeasure
with the average polarization $(\pa_xh_0(x,y),\pa_y h_0(x,y))$.
The correlator in the r.h.s. depends on the polarization and on
$\Delta n_i -\Delta n_j , \Delta m_i -\Delta m_j$, i.e. it is translation invariant.

\subsection{The Variational Principle}

\subsubsection{ } Here we will outline the derivation of the
variational problem which determines the limit shape of the height function.

First, consider the sequence of rectangular domain of size $M,M$
when $N, M\to \infty$ such that $a=N/M$ is finite and the boundary values of the height function stabilize to the function
$\phi(x,y)=hy+vx, x\in [0,a], y\in [0,1]$. The partition function of such system has the asymptotic
\begin{equation}\label{asym}
Z_{N,M}\propto \exp(NM\sigma(h,v))
\end{equation}
where $ \sigma(h,v)$ is the Legendre transform of the free
energy for the torus.

For a domain $D_\e$, choose a subdivision
of it into a collection of small rectangles.
Taking into account (\ref{asym}), the partition function of the $6$-vertex model with zero electric fields and stabilizing boundary conditions can be written as
\[
Z_{D_\e}\simeq \sum_h e^{|D_\e|\sum_i\sigma(\Delta h_{x_i}, \Delta h_{y_i})\Delta_{x_i}\Delta_{y_i}}
\]
When $\e\to 0$ the size of $|D_\e|$ the region is increasing and
the the leading contribution to the sum comes from the height function which minimizes the functional
\[
I(h)=\int_D \sigma(\nabla h) d^2x
\]

One can introduce the extra weight $q^{vol(h)}$ to the partition function of the $6$-vertex model. It corresponds to inhomogeneous
electric fields \cite{PaR}. If $\e\to 0$ and $q=\exp(-\lambda \e)$
the leading contribution to the partition function
comes from the height function which minimizes the
functional
\[
I_\lambda(h)=\int_D \sigma(\nabla h(x)) d^2x+\lambda \int_D h(x)d^2x
\]

\subsubsection{The variational principle}
Thus, in order to find the limit shape in the thermodynamical limit of the $6$-vertex model on a disc with Dirichlet boundary conditions $\varphi_0$, we should minimize the functional
\begin{equation}\label{var-pr}
I_\lambda[\varphi]=\int_D\sigma(\nabla \varphi)d^2x+\lambda\int_D \varphi d^2x,
\end{equation}
on the space $L(D,\varphi_0)$
of functions satisfying the Lipshitz condition
$$
|\varphi(x,y)-\varphi(x',y')|\leq |x-x'|+|y-y'|
$$
and the boundary conditions monotonically increasing
in $x$ and $y$ directions
$$
\varphi|_{\partial D}=\varphi_0.
$$

Since $\sigma$ is convex, the minimizer is unique when it exists. Thus, we should expect that the variational problem (\ref{var-pr})
has a unique solution.

The large deviation principle applied to this situation
should result in the convergence in probability, as $\e\to 0$ of random normalized height functions
$h(x,y)$  to the minimizer of (\ref{var-pr}).

\subsubsection{} If the vector $\nabla h(x,y)$ is not a singular point of $\sigma$,
the minimizer $h$ satisfies the Euler-Lagrange equation in a neighborhood
of $(x,y)$
\begin{equation}\label{E-L_eq}
{\rm div}(\nabla\sigma\circ\nabla h)=\lambda.
\end{equation}
We can also rewrite this equation in the form
\begin{equation}\label{fun_g}
\nabla\sigma(\nabla h(x,y))=\frac{\lambda}{2}\,(x,y)+(-g_y(x,y),g_x(x,y)),
\end{equation}
where $g$ is an unknown function such that $g_{xy}(x,y)=g_{yx}(x,y)$.
It is determined by the boundary conditions for $h$.

Applying (\ref{f-sig-comp}), it follows that
\begin{equation}
\nabla h(x,y)=\nabla f\Bigl(\frac{\lambda}{2}\,x-g_y(x,y),
\frac{\lambda}{2}\,y+g_x(x,y)\Bigr).
\label{height}
\end{equation}

From the definition of the slope, see (\ref{slope}), it follows that $|f_H|\leq 1$
and $|f_V|\leq 1$. Thus, if the minimizer $h$ is differentiable at $(x,y)$, it
satisfies the constrains $|h_x|\leq 1$ and $|h_y|\leq 1$. It is given that
$h_x,h_y\geq 0$, hence, $0\leq h_x\leq 1$ and $0\leq h_y\leq 1$.

In particular, one can choose $g(x,y)=0$. In this case the function (\ref{height}) is the minimizer of the rate functional $I_\lambda[h]$ with very special boundary conditions.
For infinite region $D$ (and finite $\lambda$) this height
function reproduces the free energy as a function of electric fields.

The  limit shape height function is a real analytic function
almost everywhere. One of the corollaries of (\ref{height}) is that curves on which it is not
analytic ( interfaces of the limit shape) are images of boundaries between different phases in the phase diagram of $f(H,V)$.  The ferroelectric and
antiferroelectric regions, where $f$ is linear, correspond to the regions where
$h$ is linear.

From (\ref{height}) and the asymptotic of free energy near boundaries  between
different phases, we can make some general
conclusions about the the structure of limit shapes.

First conclusion is that near regular pieces of
boundary between regions where the height function is real analytic it
behaves
as  $h(x,y)-h(x_0,y_0)\simeq d^{3/2}$ where $d$ is the distance from
$(x,y)$ to the closest point $(x_0,y_0)$ at the boundary.

Second obvious conclusion is that if $(x,y)$ is inside the corner
singularity in the boundary between smooth parts, the
height function behaves near the corner as $h(x,y)-h_(x_0,y_0)
\simeq d^{5/3}$ where $d$ is the distance from $(x,y)$ to the corner,
and $(x_0,y_0)$ are coordinates of the corner.

\subsection{Limit shapes for inhomogeneous models}
So far the analysis of limit shapes was done in the homogeneous
$6$-vertex model. The  variational problem determining limit shapes
involves the free energy as the function of the polarization.
The extension of the analysis outlined above to an inhomogeneous
periodically weighted case is straightforward. The variational
principle is the same same but the function $\sigma$ is determined
from the Bethe ansatz for the inhomogeneous model. It is again the Legendre
transform of the free energy as a function of electric fields.

The computation of the free energy from the large $N$ asymptotic of Bethe ansatz equations is similar and based on similar conjecture about the accumulation of solutions to Bethe equations on
a curve. In the free fermionic case it is illustrated in
section \ref{ff}. If the size of the fundamental domain is
$k\times m$ one should expect $2(k+m)$ cusps in the
boundary between ordered and  disordered regions.

\subsection{Higher spin $6$-vertex model}
The weights in the $6$-vertex model can be identified with
matrix elements of the $R$-matrix of $U_q(\widehat{sl}_2)$
in the tensor product of two $2$-dimensional representations of this algebra. We denoted such matrix $R^{(1,1)}(z)$.

In a similar fashion one can consider matrix elements of $R^{(l_1,l_2)}(z)$ as weights of the higher spin generalization
of the $6$-vertex model. We will call it higher spin $6$-vertex model.

A state in such models is an assignment of a weight
of the irreducible representation $V^{(l_1)}$ to
every horizontal edge, and of a weight of the irreducible
representation of $V^{(l_2)}$ to every vertical edge.

Natural local observables in such model are sums of products
of spin functions of an edge:
\[
s_e(S)=\mbox{ the weight } s \mbox{ assigned to } e \mbox{ in the state } S ,
\]

Spin variables define the height function locally as
\[
s_e=h(f_e^+)-h(f_e^-)
\]
where $f_e^\pm$ are faces adjacent to $e$.
When the domain has trivial fundamental group the
local height function extends to a global one.
Otherwise one should the region into simply-connected pieces.

The height function has the property
\[
|h(n,m)-h(n',m')|\leq l_1|n-n'|+l_2|m-m'| ,
\]

Given a sequence of domains $D_\e$ with $\e\to 0$ one
should expect that the normalized height function $\e h(\e n, \e m)$ develops the limit shape $h_0(x,y)$ which minimizes the
functional
\[
I^{(l_1,l_2)}[h]=\int_D \sigma^{(l_1,l_2)}(\nabla h) d^2x
\]
subset to the Dirichlet boundary conditions and
the constraint
\[
|h(x,y)-h(x',y')|\leq l_1 |x-x'|+l_2 |y-y'| ,
\]

The properties of such model and of its free energy as
function of electric fields is an interesting problem
which need further research.

\section{Semiclassical limits}

\subsection{The semiclassical limits in Bethe states}

The spectrum of quantum spin chains described in section \ref{Bethe-ans} in terms of Bethe equations has a natural
limit when $m_i=R_i/h$, $q=e^h$, and $h\to 0$ with fixed $R_i$.
In this limit the quantum spin chain with the quantum monodromy matrix (\ref{mon-matr}) becomes the classical spin chain with the monodromy matrix (\ref{cl-tran}).

Semiclassical eigenvectors of quantum transfer-matrices
correspond to solutions to Bethe equations which accumulate
along contours representing branch cuts on the spectral curve
of $T(z)$, see \cite{RS} and for more recent results \cite{S}.

\subsection{The semiclassical limits in the higher spin $6$-vertex model}

Relatively little known about the semiclassical limit of the higher spin $6$-vertex model. This is the limit when $l_1=R_i/h$, and $h\to 0$. Here, depending on the values of $\Delta$, $h=\eta$ or $h=\gamma$.

The first problem is to find the asymptotic of the conjugation action of the $R$-matrix in this limit. The mapping
\[
x\to R^{(l_1,l_2)}(u)xR^{(l_1,l_2)}(u)^{-1}
\]
is an automorphism of $End(V^{(\l_1)}\otimes V^{(l_2)})$ which
as $h\to 0$ becomes the Poisson automorphism $\rho^{(R_1,R_2)}(u)$ of the Poisson algebra of functions on $S^{(R_1)}\otimes S^{(R_2)}$.

When $\Delta=1$ this automorphism was computed in \cite{Skl}.
It is easy to extend this results for $\Delta\neq 1$. For
constant $R$-matrices see \cite{Re}.

Next problem is to find the semiclassical asymptotic for the
$R$-matrix considered as an "evolution operators".
For this one should choose two Lagrangian submanifolds in
$S^{(R_1)}\otimes S^{(R_2)}$, one corresponding to the initial data and the other corresponding to the target data. Points in these manifolds parameterize corresponding semiclassical states.
Then matrix elements of the $R$-matrix should have the
asymptotic
\[
R^{(l_1,l_2)}(u)(\sigma_1,\sigma_2)=const \ \ \exp{\frac{S^{(R_1,R_2)}(\sigma_1,\sigma_2)}{h}}\sqrt{H(\sigma_1,\sigma_2)}
(1+O(h))
\]
where $S$ is the generating function for the mapping $\rho$ and
$H$ is the Hessian of $S$. This asymptotical behavior of
the $R$-matrix defines the semiclassical limit of the partition function and needs further investigation.

\subsection{The large $N$ limit in spin-$1/2$ spin chain as a semiclassical limit}

\subsubsection{}

Local spin operators in a spin-$1/2$ spin chain of length $N$ are
\[
S^a_n=1\otimes \dots \sigma^a_n\otimes \dots \otimes 1, \ \ n=1,\dots, N
\]
They commute as
\[
[S^a_n,S^b_m]=i\sum_c \e_{abc} \delta_{nm}S^c_n
\]

As $N\to \infty$ the operators $S^a_n$ converge to local continuous
classical spin variables
\[
S^a_n\to S^a(\frac{n}{N})
\]
where $S^a(x)$ are local functionals on the classical phase space
of the continuum spin system. The commutation relations become the relations between Poisson brackets:
\[
\{S^a(x), S^b(y)\}=\delta(x-y)\sum_c f^{ab}_c S^c(x)
\]
where $S^a(x)$ are continuous local classical spin variables. This combination of classical and continuous limit looks more convincing in terms
of observables
\[
S_N^a[f]=\frac{1}{N} \sum_{n=1}^N f(\frac{n}{N})S^a_n,
\]
As $N\to \infty$, such operators becomes
$S^a[f]=\int_0^1f(x)S^a(x)dx$ and the commutation relation
\[
[S_N^a[f],S_N^b[g]]=\frac{i}{N}\sum_c\e_{abc}S_N^c[fg]
\]
becomes
\[
\{S^a[f],S^b[g]\}=\sum_c\e_{abc}S^c[fg]
\]
Here we used the asymptotic of commutators in the semiclassical limit: $[a,b]=ih\{a,b\}+\dots$
and the fact that $\frac{1}{N}$ plays the role of the Plank constant.

\subsubsection{} Consider the limit $N\to \infty$ of the row-to-row transfer-matrix in the 6-vertex model. Assume that at the same time $\Delta\to 1$
such that $\Delta=1+\frac{\kappa^2}{2N^2}+\dots$. To be specific consider
the case when $\Delta >1$, so that $\kappa=N\eta$ is finite and real.

We can write the transfer-matrix of the 6-vertex model as
\[
t(u)=(\sinh(u))^Ntr(T_N(u)),
\]
where $T_N(u)$ is the solution to the difference equation
\[
T_{n+1}(u)=\frac{R_n(u)}{\sinh(u)}T_n(u)
\]
with the initial condition $T_0(u)=1$. Here
$R(u)$ is the $R$-matrix of the $6$-vertx model. For small $\eta$ we have:
\[
\frac{R_n(u)}{\sinh(u)}=1+\eta(\frac{1}{2}\coth(u)\sigma^3\sigma^3_n+
\frac{\sigma^+\sigma^-_n+\sigma^-\sigma^+_n}{\sinh(u)})+O(\eta^2)
\]
Taking this into account we conclude that as $N\to \infty$, $T_n(u)\to T(u|x)$ where $T(u|x)$ is the solution to
\[
\frac{\pa T(u|x)}{\pa x}=(\frac{1}{2}\coth(u)\sigma^3S^3(x)+
\frac{\sigma^+S^-(x)+\sigma^-S^+(x)}{\sinh(u)})T(u|x)
\]
with the initial condition $T(u|0)=1$. Here $\sigma^z$ is the same as in  (\ref{sigma-m}), $\sigma^\pm=
\frac{1}{2}(\sigma^x\pm i\sigma^y)$, and $S^\pm=\frac{1}{2}(S^1\pm iS^2)$

The row-to-row transfer-matrix of the 6-vertex model have the following asymptotic in this
limit:
\begin{equation}\label{t-fun}
t(u)\to (\sinh(u))^N \tau(u)
\end{equation}
where $\tau(u)=tr(T(u|1)$.

\subsubsection{} Now let us study the evolution of local spin operators in this limit.

Consider the partition function of the $6$-vertex model  on
a cylinder as a liner operator
\[
Z_{M,N}^{(C)}=t(u)^M
\]
acting in ${\CC^2}^{\otimes N}$. Consider the row-to-row transfer-matrix as the evolution operator by one step and
$Z_{M,N}^{(C)}$ as the evolution operator by $M$ steps. In the
Heisenberg picture, local spin operators evolve as:
\begin{equation}\label{d-evol}
S^a_{n,m+1}=t(u)S^a_{n,m}t(u)^{-1}
\end{equation}

In the continuum semiclassical limit $N\to \infty$ the transfer-matrix becomes a functional (\ref{t-fun}) on the
phase space of continuous classical spins. The equation
(\ref{d-evol}) can be written as
\[
S^a_{n,m+1}-S^a_{n,m}=[t(u),S^a_{n,m}]t(u)^{-1}
\]
As $N\to \infty$ and $x=n/N$ and $t=m/N$ are fixed, this
equation become the evolution equation for continuum spins:
\[
\frac{\pa S^a(x,t)}{\pa t}=\{H(u), S^a(x,t)\}
\]
where $H(u)=\log t(u)$.

The evolution with respect to the partition function on a cylinder of height $M$
\begin{equation*}
S^a_{n,m+M}=Z_{M,N}^{(C)}S^a_{n,m}{Z_{M,N}^{(C)}}^{-1}
\end{equation*}
becomes the evolution in time $T=M/N$: $S^a(x,t)\mapsto S^a(x,t+T)$.

If we choose a Lagrangian submanifold in the phase space
of continuous spin system corresponding to the initial and target data (say, a version of initial and target $q$-coordinates),
the asymptotic of the partition function should be of the form
\[
Z_{N,M}^{(C)}(\sigma_1,\sigma_2)=const e^{-NS_T(\sigma_1, \sigma_2)}
\sqrt{Hess(\sigma_1,\sigma_2)}(1+O(1/N))
\]
where $S_T$ is the Hamilton-Jacobi action for the Hamiltonian $H(u)$, $Hess(\sigma,\tau)$ is the Hessian of $S_T$, and in the
left side we have the matrix element of $Z_{N,M}^{(C)}$ between  semiclassical states corresponding to $\sigma_1$ and to $\sigma_2$.

The Bethe equations and the spectrum of the Heisenberg Hamiltonian
in this limit was studied in \cite{Ka}\cite{De}.

\section{The free fermionic point and dimer models}

Decorate the square grid inserting a box with two faces
to each vertex, as it is shown on Fig. \ref{reference_config}.
Recall that a dimer configuration on a graph is a perfect matching on a set of vertices connected by edges. In other words, it is a
collection of "occupied edges" (by dimers) such that two occupied
edges never meet, and any vertex in an endpoint to an occupied edge.

Dimer configurations on the decorated square grid project to
$6$-vertex configurations on a square grid as it is shown on Fig.
\ref{projection}.

\begin{figure}[t]
\begin{center}
\includegraphics[width=7cm, height=10cm]{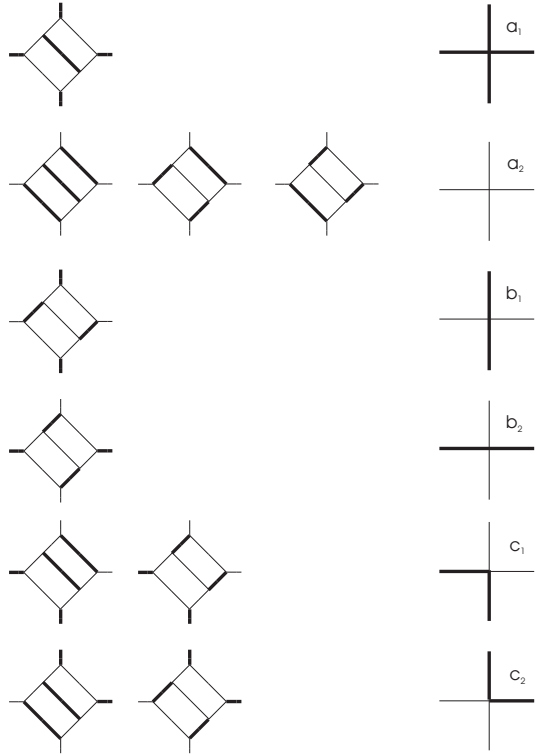}
\end{center}
\caption{Projection of dimer configurations on $G^{(0)}_{nm}$
onto the different types of vertices at the $(n,m)$-vertex.}
\label{projection}
\end{figure}

\begin{figure}[t]
\begin{center}
\includegraphics[width=4cm, height=4cm]{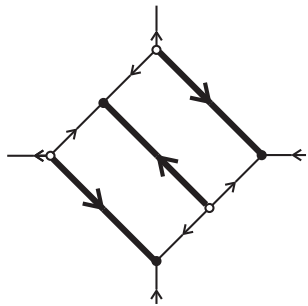}
\end{center}
\caption{Reference dimer configuration on $G^{(0)}_{nm}$.}
\label{reference_config}
\end{figure}

It is easy to check that any edge weight system on the decorated lattice projects to $6$-vertex weights at the free fermionic
point, when $a^2+b^2-c^2=0$ at every vertex. Recall that edge weights is a mapping $w: \mbox{Edges}\to \RR_{\geq 0}$. The weight of a dimer configuration is
\[
W(D)=\prod_{e\in D} w(e)
\]
The statement above means
\[
\sum_{D\in \pi^{-1}(S)} W(D)=\prod_v w_v(S)
\]
where $S$ is a $6$-vertex configuration on the square grid,
$\pi$ is the projection from dimer configurations on the decorated square grid to the $6$-vertex configurations, and the $6$-vertex
weights $w_v$ are given by an explicit formula. The weights $w_v$ satisfy the free fermionic condition.

A pair of dimer configurations $D, D_0$ on bipartite graphs define the height function. This height function agrees with the
height function of the $6$-vertex model:
\[
h_S(f)=h_{D,D_0}(f)
\]
where $S$ is a $6$-vertex state on the square lattice, and
$h_S$ is the corresponding height function, $D\in \pi^{-1}(S)$,
$D_0$ is shown on Fig. \ref{reference_config}, and $f$ is the
face of the decorated lattice which projects to a face of the
square grid.

\appendix

\section{Symplectic and Poisson manifolds}\label{be}

\subsection{}Recall that an even dimensional manifold equipped
with a closed non-degenerate 2-form is called {\it symplectic}.

Let $(\M, \omega)$ be a $2n$-dimensional symplectic manifold. In local coordinates:
\[
\omega=\sum_{ij=1}^{2n} \omega_{ij}dx^i\wedge dx^j,\ \ det(\omega)\neq 0, \ \ d\omega=\sum_{k=1}^{2n} \sum_{ij=1}^{2n} \frac{\partial\omega_{ij}}{\partial x^k}dx^i\wedge dx^j=0,
\]

The last identity is equivalent to the Jacobi identity
for the bracket
\[
\{f,g\}=\sum_{ij=1}^{2n} (\omega^{-1})^{ij}\frac{\partial f}{x^i} \frac{\partial g}{x^j}
\]

\subsection{} A smooth manifold $M$ with bi-vector field
$p$ (a section of the bundle $\wedge^2TM$) such that the bracket between two smooth functions
\[
\{f,g\}=p(df\wedge dg)
\]
satisfies the Jacobi identity is called a Poisson manifold.
I local coordinates $x^1,\dots, x^n$, $p(x)=\sum{i,j=1}^np^{ij}(x)\frac{\pa}{\pa x^i}\wedge \frac{\pa}{\pa x^i}$.

\subsection{} A Poisson tensor on a smooth manifold $M$ defines a subspace mapping $p: T^*M\to TM$. The image is a system of subspaces $p(T^*M)\subset TM$ which is a distribution
on $M$. Leaves of this distribution are spanned by curves which are flow lines of piece-wise Hamiltonian vector fields. They are smooth submanifolds. Symplectic leaves of the Poisson manifold $M$
are leaves of this distribution.

Well know examples of symplectic leaves are co-adjoint orbits
in the dual space to a Lie algebra.

\section{Classical integrable systems and their quantization}
\subsection{Integrable systems in Hamiltonian mechanics}
The notion of integrability is most natural in the Hamiltonian
formulation of classical mechanics. For details see \cite{Arn}
\cite{RSTS}.

In Hamiltonian formalism of classical mechanics the dynamics is taking place on the phase space and equations of motion are of the first order. When it is a system of particles moving on a manifold
$M$,  local coordinates are positions and momenta of particles.
Globally, the phase space in this case is the cotangent bundle to the manifold $M$, see for example \cite{Arn}.

Hamiltonian formulation of spinning tops or other systems
with more complicated constraints involves more complicated
phase spaces. In all cases a phase space has a structure
of a structure of a symplectic manifold, that is it comes
together with a non-degenerate closed 2-form on it.

Any symplectic manifold (and so any phase space of a Hamiltonian
system admits local coordinates (Darboux coordinates) in which
the 2-form  has the form
\[
\omega=\sum_{i=1}^n dp_i\wedge dq^i
\]
These coordinates can be interpreted as momenta and positions,
though in case of spinning tops this interpretation does
not have a lot of physical meaning but Darboux coordinates is a convenient mathematical tool.

The dynamics in a Hamiltonian system is determined by
the energy function $H$. The trajectories of such system
in local Darboux coordinates are solutions to differential
equations:
\[
\frac{d q^i}{ dt}=\frac{\pa H}{\pa p_i}, \ \ \frac{d p_i}{ dt}=-\frac{\pa H}{\pa q^i}
\]

Geometrically, trajectories are flow lines of the Hamiltonian vector field $v_H=\omega^{-1}(dH)\in \Gamma(\wedge^2TM)$
where $\omega^{-1}: TM\to T^*M$ is the bundle isomorphism induced by the symplectic form $\omega$.

Let $\M$ be a $2N$-dimensional symplectic manifold.

\begin{definition} An integrable system on $\M$ is
a collection of $N$ independent functions on $\M$
which commute with respect to the Poisson bracket.
\end{definition}

Recall that

\begin{itemize}
\item The ``level surfaces''
\begin{equation}
{\M }(c_1,\dots,c_n)={\{} x \in {\M} , I_i(x)=c_i{ \}}
\end{equation}
are invariant with respect to the flow of any Hamiltonian
$H=F(I_1,\dots,I_N)$.
\item For every such Hamiltonian and every level surface
${\M}(c_1,\dots,c_n)$ there exists an affine coordinate system $(p_1,\dots,p_n)$
in which the Hamiltonian flow generated by $H$ is linear:
$p_i=$constant.
\item The coordinate system  $(p_1,\dots,p_n)$ on  ${\M}(c_1,\dots,c_n)$
can be completed to a canonical coordinate system $(p_i,q_j)$
in every sufficiently small neighborhood of ${\M}(c_1,\dots,c_n)$.
These coordinates are called action-angle variables and in some cases
these coordinates are global.
\end{itemize}

\section{Poisson Lie groups}
\subsection{Poisson Lie groups}
A Lie group $G$ is called a Poisson Lie group if
\begin{itemize}

\item $G$ has a Poisson structure, i. e. it is given together
with the Poisson tensor $p\in \Gamma(\wedge^2TG)$, in local coordinates $p(x)=p^{ij}(x)\frac{\pa }{\pa x^i}\wedge\frac{\pa }{\pa x^j}\in \wedge^2 T_xG$. This tensor the Poisson bracket
on smooth functions on $G$ (Lie bracket satisfying the
Leibnitz rule with respect to the point-wise multiplication)
\[
\{f,g\}(x)=\sum_{ij}p^{ij}(x)\frac{\pa f}{\pa x^i}\frac{\pa g}{\pa x^j}(x)
\]

\item The Poisson structure is compatible with the group
multiplication. This means that the multiplication mapping
$G\times G\to G$ bring the Poisson tensor on $G\times G$
to the Poisson tensor on $G$, or:
\[
\sum_{ij}p^{ij}(xy)\frac{\pa f}{\pa z^i}\frac{\pa g}{\pa z^j}(z)|_{z=xy}=
\sum_{ij}p^{ij}(x)\frac{\pa f}{\pa x^i}\frac{\pa g}{\pa x^j}(xy)+
\sum_{ij}p^{ij}(y)\frac{\pa f}{\pa y^i}\frac{\pa g}{\pa y^j}(xy)
\]
for any pair of functions $f,g$.
\end{itemize}

For more details on Poisson Lie groups and for numerous
examples see \cite{Dr}\cite{ChariP}\cite{KorSoib}.

\subsection{Factorizable Poisson Lie groups and classical $r$-matrices}
The class of Poisson Lie groups relevant to integrable systems
has so-called $r$-matrix Poisson brackets. Let $\g$ be a Lie algebra corresponding to the Lie group $G$. A classical $r$-matrix for $G$ is an element $r\in \g\otimes \g$ satisfying the bilinear identity
\begin{equation}\label{cYB}
[r_{12}, r_{13}]+[r_{12},r_{23}]+[r_{13},r_{23}]=0
\end{equation}
The Poisson tensor
\[
p(x)=r-Ad_x(r)\in \g\wedge \g
\]
defines a Poisson Lie structure on $G$ if $r$ satisfies (\ref{cYB})
and $r+\sigma(r)$ is an invariant tensor.

Two remarks: we identified $T_xG with T_eG=\g$ using
left translations on $G$; $Ad_x$ is the diagonal adjoint action of $G$ on $\g\wedge \g$, $u\wedge v\to xux^{-1} \otimes xvx^{-1}$
when $G$ is a matrix Lie algebra.

Let $\pi^V:G\to GL(V)$ and $\pi^W: G\to GL(W)$ be two (finite-dimensional) representations of $G$
and let $\pi^V_{ij}, \pi^W_{ab}$ be matrix elements of $G$ in a linear basis: $\pi^V_{ij}: g\in G\to \pi_{ij}^V(g)\in \CC$.
The $r$-matrix Poisson brackets between two
such functions are:
\[
\{\pi^V_1, \pi^W_2\}=[(\pi^V\otimes \pi^W)(r), \pi^V_1,\pi^W_2]
\]
Taking trace in this formula we see that characters of finite
dimensional representations form a Poisson commutative algebra on
$G$. Restricting this subalgebra to a symplectic leaf of $G$ we
will obtain an integrable system if the number of independent functions among $\chi_\lambda$ after this restriction is equal
to half of the dimension of the symplectic leaf.

the principal advantage of this approach is that it gives an
algebraic way to construct classical $r$-matrices through the double construction of Lie bi-algebras.

In then next section we will see how $r$-matrices with spectral
parameter appear naturally from Lie bi-algebras on loop algebras.

\subsection{Basic example $G=LSL_2$}

Our basic example is an infinite dimensional Poisson Lie
group  $LGL_2$ of mappings $S^1=\{z\in \CC| |z|=1\}\to GL_2$
which are holomorphic inside the unit disc.

The Lie algebra $Lgl_2$ (we consider maps which are Laurent polynomials in $t$) of the Lie group $LGL_2$ has a linear basis
$e_{ij}[n]$ with $e_{ij}[n](z)=e_{ij}t^{-n-1}$. It also has an invariant scalar product $(x[n],y[m])=\delta_{n,m}$.
The following element of the completed tensor product $Lgl_2\otimes Lgl_2$

\begin{equation}\label{cl-r}
r=\sum_{i\geq j; i,i=1,2}e_{ij}[0]\otimes e_{ji}[0]+\sum_{i,j=1,2;n\in \ZZ, n\geq 1} e_{ij}[n]\otimes e_{ji}[-n]
\end{equation}
satisfies the classical Yang-Baxter equation \eqref{cYB}.
We will not explain here why, but this follows from the Drinfeld's double construction for Lie bialgebras.

Let $\pi^V: GL_2\to GL(V)$ be a representation of $GL_2$ and $a\in \CC^\vee$ be a non-zero complex number. The evaluation representation $\pi_{V,a}$ of $Lgl_2$ acts in $V$ as
\[
\pi_{V,a^2}(x[n])=\pi^V(x)a^{-2n-2}
\]
where $a$ is a non-zero complex number.

Evaluating the element \eqref{cl-r} in $\pi_{V,a}\otimes \pi_{V,b}$ where $V$ is the two dimensional representation
we get
\[
\pi_{V,a}\otimes \pi_{V,b}(r)=r(a/b)+f(a/b)1
\]
where $f(z)$ is a scalar function and
\begin{equation}\label{Lsl2-r}
r(z)=\frac{z+z^{-1}}{z-z^{-1}}\sigma^z\otimes \sigma^z+
\frac{2}{z-z^{-1}}(\sigma^+\otimes \sigma^-+\sigma^-\otimes \sigma^+)
\end{equation}

Here $\sigma^z=e_{11}-e_{22}$, $\sigma^+=e_{12}$, and
$\sigma^-=e_{21}$ are Pauli matrices.

The Yang-Baxter equation for $r$ implies
\[
[r_{12}(z), r_{13}(zw)]+[r_{12}(z),r_{23}(w)]+[r_{13}(zw),r_{23}(w)]=0
\]

The Poisson bracket between coordinate functions $g_{ij}(z)$ are
\begin{equation}\label{loop-P}
\{g_1(z), g_2(w)\}=[r(z/w), g_1(z)g_2(w)]
\end{equation}

One of the important properties of such Poisson brackets is that characters of finite dimensional representations define families
of Poisson commuting functions on $LG_2$:
\[
t_V(z)=tr_V(\pi^V(g(x)), \ \ \{t_V(z), t_W(w)\}=0
\]
The coefficients of these functions will produce Poison commuting
integrals of motions for integrable systems, when restricted to symplectic leaves of $LGL_2$.

\subsection{Symplectic leaves}

It is a well know, classical fact, which can be traced back to works of Lie, that symplectic leaves of the Poisson manifold which is the dual space to a Lie algebra are co-adjoint orbits.
Similarly, symplectic leaves of a Poisson Lie group $G$ are
orbits of the dressing action of the dual Poison Lie group
on $G$.

The structure of symplectic leaves of Poisson Lie groups is
well known for finite dimensional simple Lie algebras, see \cite{KorSoib}\cite{Yak}. For the construction of integrable
spin chains related to $SL_2$ we will need only some
special symplectic leaves of $LGL_2$. These symplectic leaves are
symplectic leaves of finite dimensional Poisson submanifolds
of polynomial maps of given degree.

\section{Quantization}

\subsection{Quantization}

This section is a brief outline of the quantization
of Hamiltonian systems. There is also a very important point
of view of a path integral quantization. We will not
discuss it here.

\subsubsection{Quantized algebra of observables}

Let $M$ be a symplectic manifold, which is the phase space of our
mechanical system. We want to describe possible quantum mechanical
systems which reproduce our system in the classical limit. This procedure is called quantization.

Let $A$ be a Poisson algebra over $\CC$ i.e. it is a complex vector space with a commutative multiplication $ab$ and with the Lie
bracket $\{a,b\}$ such that $\{a,bc\}=b\{a,c\}+\{a,b\}c$.
Let $X\subset \CC$ is a neighborhood of $0\in \CC$.

\begin{definition} A deformation quantization of $A$ is a
family of associative algebras $A_h$, parameterized by $h\in X$
together with two families of linear maps $\phi_h: A_h\to A$
and $\psi_h: A\to A_h$ such that
\begin{itemize}
\item $\lim_{h\to 0} \phi_h\circ \psi_h\to id_A$,
\item $\lim_{h\to 0} \phi_h(\psi_h(a)\psi_h(b))=ab$,
\item $\lim_{h\to 0} \frac{\phi_h(\psi_h(a)\psi_h(b)-\psi_h(a)\psi_h(b))}{ih}=\{a,b\}$,
\end{itemize}
\end{definition}

Denote by $C(M)$ the classical algebra of observables, i.e.
the algebra of real valued functions on the phase space. If $M$ is $T^*N$, this will be the algebra of functions on $M$ which are polynomial in the cotangent direction and smooth on $N$. If $M$ is an affine algebraic, manifold, $C(M)$ will be algebra of polynomial
functions. If $M$ is a smooth manifold it is $C^\infty M$, etc..
The space $C(M)$ has a natural structure of a Poisson algebra with the point-wise multiplication and the Poisson bracket determined by the symplectic form on $M$.

Denote its complexification by $C(M)_\CC$.
The real subalgebra $C(M)\subset C(M)_\CC$ is the set of fixed
points of complex conjugation, i.e. real valued functions. Denote complex conjugation by
$\sigma$, i. e. $\sigma(f)(x)=\overline{f(x)}$.

Deformation quantizations which are relevant to quantum mechanics
are real deformation quantizations:

\begin{definition} A real deformation quantization of $C(M)$ is
a pair $(A_h, \sigma_h, h\in X\subset \RR)$ where $A_h$ is a deformation quantization (as above) of $C(M)_\CC$, $h$ is a real deformation parameter,  and
$\sigma_h$ is a $\CC$-anti-linear anti-involution of $A_h$, i.e.
\[
\sigma_h^2=1, \  \ \sigma_h(ab)=\sigma_h(b)\sigma_h(a), \ \
\sigma_h(sa)=\bar{s}\sigma_h(a),
\]
The the subspace $C(M)_h\subset A_h$ of fixed points
of $\sigma_h$ is called the space of quantum observables.
Linear maps $\phi_h$ and $\psi_h$ from the definition f $A_h$ should satisfy extra condition
\[
\lim_{h\to 0}\phi_h\circ\sigma_h\circ \psi_h=\sigma
\]
\end{definition}

The space of fixed points of $\sigma_h$ is called the {\it space
of quantum observables}. Notice that the multiplication does not preserve this space. However, it
is closed with respect to operations $AB+BA$ and $i(AB-BA)$. Such
structures are called Jordan algebras.

Of course these definition still need a clarification. The vector space $C(M)$ is infinite dimensional and we have to specify in which topology our linear isomorphism is continuous and in which sense we should take limits. The way how to handle this is either dictated by the nature of a specific problem, so we will discuss it later in relation to integrable spin chains which is the main subject here.

\subsection{Examples of family deformations}

\subsubsection{} Take $M=\RR^2$, $A=Pol_\CC (\RR^2)=\CC[p,q]$ with the standard
symplectic form $dp\wedge dq$ giving the bracket $\{p,q\}=1$ (this
determines the bracket). We have a natural monomial basis $p^nq^m$
on $A$. Define
 \[
  A_h = < p,q| pq-qp=h>
 \]
It is clear that this is a family of associative algebras. To identify
this family with a deformation quantization of $A$ we should find
$\phi_h$ and $\psi_h$. For this choose monomial bases $p^nq^m$ in
$A_h$ and $A$. Define linear maps
$\phi_h$ and $\psi_h$ as linear isomorphisms $A_h\simeq A$ identifying
the monomial bases. It is clear that this choice makes $A_h$ into a
defromation quantization of $A$.

\subsubsection{} Let $\g$ be a Lie algebra, and consider $Pol(\g^*)=\CC[\g]$. If
$\{e_i\}$ is a basis for $\g$, then we can think of the $e_i$ as
coordinate functions $x_i$ on $\g^*$. A theorem of Kostant,
$\{f,g\}(x) = < x,[df(x),dg(x)]>$, $\{x_i,x_j\}=\sum_k c_{ij}^k
x_k$. $Pol(\g^*)$ then gets a Poisson bracket.

We can get a deformation quantization
 \[
  A_h = <x_1,\dots, x_n| x_ix_j-x_jx_i=h \sum_k c_{ij}^k x_k>
 \]
 Note that $A_h\cong U\g$ for any $h\neq 0$ (you just have to rescale the $x$'s by $h$). On the other hand, Choosing the monomial basis $x_1^{a_1}\cdots x_n^{a_n}$ in $\CC[x_1,\dots, x_n]$ and the PBW basis in $A_h$. Identifying them, we get $A_h \cong \CC[x_1,\dots, x_n]$, which is how the PBW theorem is usually formulated.

 \[
  U\g \cong Pol(\g^*)\cong S(\g)
 \]
We get a linear isomorphism $\phi_h\colon A_h\cong A$. It is easy
to check that this is a deformation quantization.

\subsection{Quantization of integrable systems}

As we have seen above the quantization of a classical Hamiltonian system on $\M$ with the Hamiltonian $H\in C(M)$consists of the following:
\begin{itemize}
\item A family of associative algebras $A_h$ quantizing the algebra of function on $\M$.
\item The choice of quantum Hamiltonian $H_h\in A_h$ for each $h$,such that $\phi_h(H_h)\to H$ as $h\to 0$.
\end{itemize}

The quantization is integrable if for each $h$ there is a maximal
commutative subalgebra of $C_h(M)$ quantizing
the subalgebra of classical integrals in $C(M)$, which contains $H_h$.

To be more precise, let $I(M)\in C(M)$ be the subalgebra generated by Poisson commuting integrals in the classical algebra of observables, i.e. $\{H,F\}=0$ for each $F\in I(M)$ and $\{G,F\}=0$ for each $F,G\in I(M)$. Its integrable quantization consists of a commutative subalgebra  $I_h\in A_h$ such that $FH=HF$ for each $F\in I_h$, such that $\lim_{h\to 0} \phi_h(I_h)=I$. For the precise definition in case of formal deformation quantization see \cite{RY}.

When quantized algebra of observables is represented
in a Hilbert space, the important problem is the computation of
the spectrum of commuting Hamiltonians.

\end{document}